\documentclass[11pt]{article}
\usepackage[margin=1.25in]{geometry}
\RequirePackage[OT1]{fontenc}
\RequirePackage{amsthm,amsmath,amsmath,amsfonts,amssymb,verbatim,soul}
\RequirePackage{natbib}
\RequirePackage[colorlinks,citecolor=blue,urlcolor=blue]{hyperref}

\RequirePackage{tablefootnote}

\RequirePackage{graphicx}

\usepackage{times}
\usepackage{color}
\usepackage{enumitem}
\usepackage{multirow}
\usepackage{subcaption}
\usepackage{appendix}

\usepackage{titling}

\usepackage{authblk}

\newtheorem{definition}{Definition}

\newtheorem{theorem}{Theorem}
\newtheorem{lemma}{Lemma}
\newtheorem{corollary}{Corollary}

\newtheorem{proposition}{Proposition}

\newtheorem{remark}{Remark}

\usepackage[linesnumbered,boxed,norelsize]{algorithm2e}
\SetKwRepeat{Do}{do}{while}%
\SetKwInput{KwInput}{Input}
\SetKwInput{KwOutput}{Output}
\SetKwInput{KwFunction}{Function}
\SetKwInput{KwParameters}{Parameters}
\SetKwComment{Comment}{$\triangleright$\ }{}
\SetCommentSty{itshape}
\let\oldnl\nl% Store \nl in \oldnl
\newcommand{\nonl}{\renewcommand{\nl}{\let\nl\oldnl}}

\newcommand{\vct}{\boldsymbol }

\newcommand{\ud}{\mathrm d}

\newcommand{\kl}{\mathrm{KL}}

\newcommand{\SPAN}{\mathrm{span}}

\renewenvironment{proof}{\noindent {\it {Proof.} }}{\hfill $\Box$ \\}

\def\R{\mathbb{R}}

\def\op{\mathrm{op}}

\renewcommand{\hat}{\widehat}
\renewcommand{\tilde}{\widetilde}
\renewcommand{\bar}{\overline}

\definecolor{DSgray}{cmyk}{0,1,0,0}

\begin{document}
%%%%%%%%%%%%%%%%

\title{Dynamic Assortment Optimization with Changing Contextual Information}

\author[1]{Xi Chen \thanks{Author names listed in alphabetical order.}}
\author[2]{Yining Wang}
\author[3]{Yuan Zhou}
\affil[1]{Stern School of Business, New York University}
\affil[2]{Machine Learning Department, Carnegie Mellon University}
\affil[3]{Computer Science Department, Indiana University at Bloomington}

\date{}

%\HISTORY{This paper was first submitted on April 12, 1922 and has been with the authors for 83 years for 65 revisions.}

\maketitle
%%%%%%%%%%%%%%%%%%%%%%%%%%%%%%%%%%%%%%%%%%%%%%%%%%%%%%%%%%%%%%%%%%%%%%

% Samples of sectioning (and labeling) in OPRE
% NOTE: (1) \section and \subsection do NOT end with a period
%       (2) \subsubsection and lower need end punctuation
%       (3) capitalization is as shown (title style).
%
%\section{Introduction.}\label{intro} %%1.
%\subsection{Duality and the Classical EOQ Problem.}\label{class-EOQ} %% 1.1.
%\subsection{Outline.}\label{outline1} %% 1.2.
%\subsubsection{Cyclic Schedules for the General Deterministic SMDP.}
%  \label{cyclic-schedules} %% 1.2.1
%\section{Problem Description.}\label{problemdescription} %% 2.

% Text of your paper here

\begin{abstract}
In this paper, we study the dynamic assortment optimization problem under a finite selling season of length $T$. At each time period, the seller offers an arriving customer an assortment of substitutable products under a cardinality constraint, and the customer makes the purchase among offered products according to a discrete choice model. Most existing work associates each product with a real-valued fixed mean utility and assumes a multinomial logit choice (MNL) model. In many practical applications,  feature/contextual information of products is readily available. In this paper, we incorporate the feature information by assuming a linear relationship between the mean utility and the feature. In addition, we allow the feature information of products to change over time so that the underlying choice model can also be non-stationary. To solve the dynamic assortment optimization under this changing contextual MNL model, we need to simultaneously learn the underlying unknown coefficient and make the decision on the assortment. To this end, we develop an upper confidence bound (UCB) based  policy and establish the regret bound on the order of $\tilde{O}(d\sqrt{T})$, where $d$ is the dimension of the feature and $\tilde{O}$ suppresses logarithmic dependence. We further establish a lower bound $\Omega(d\sqrt{T}/{K})$, where $K$ is the cardinality constraint of an offered assortment, which is usually small. When $K$ is a constant, our policy is optimal up to logarithmic factors. In the exploitation phase of the UCB algorithm, we need to solve a combinatorial optimization for assortment optimization based on the learned information.  We further develop an approximation algorithm and an efficient greedy heuristic. The effectiveness of the proposed policy is further demonstrated by our numerical studies.

\noindent\textbf{keywords}: Dynamic assortment optimization, regret analysis, contextual information, bandit learning, upper confidence bounds.
\end{abstract}

\section{Introduction}

In operations, an important research problem facing a retailer is the selection of products/advertisements for display. For example, due to the limited shelf space, stocking restrictions, or available slots on a website, the retailer needs to carefully choose an assortment from the set of substitutable products. In an assortment optimization problem, choice model plays an important role  since it characterizes a customer's choice behavior. However,  in many scenarios, customers' choice behavior (e.g., mean utilities of products) is not given as \emph{a priori} and cannot be easily estimated due to the insufficiency of historical data. This motivates the research of dynamic assortment optimization, which has attracted a lot of attentions from the revenue management community in recent years. A typical dynamic assortment optimization problem assumes a finite selling horizon of length $T$ with a large $T$. At each time period, the seller offers an assortment of products (with the size upper bounded by $K$) to an arriving customer. The seller observes the customer's purchase decision, which further provides useful information for learning  utility parameters of the underlying choice model. The multinomial logit model (MNL) has been widely used in dynamic assortment optimization literature, see, e.g., \cite{Caro2007,  Rusmevichientong2010, Saure2013, Agrawal16MNLBandit, Agrawal17Thompson,Chen:18tight,Chen:18near}.

In the age of e-commerce, side information of products is widely available (e.g., brand, color, size, texture, popularity, historical selling information), which is important in characterizing customers' preferences for products. Moreover, some features are not static and could change over time (e.g., popularity score or ratings). The feature/contextual information of products will facilitate accurate assortment decisions that are tailored to customers' preferences. In particular, we assume at each time $t=1,\ldots, T$, each product $j$  is associated with a $d$-dimensional feature vector $v_{tj} \in \mathbb{R}^d$. To incorporate the feature information, following the classical conditional logit model \citep{McFa73}, we assume that the mean utility of product $j$ at time $t$ (denoted by $u_{tj}$) obeys a linear model
\begin{equation}\label{eq:linear}
u_{tj}= v_{tj}^\top\theta_0.
\end{equation}
Here, $\theta_0 \in \mathbb{R}^d$ is the unknown coefficient to be learned. Based on this linear structure of the mean utility, we adopt the MNL model as the underlying choice model (see Section \ref{sec:model} and Eq.~\eqref{eq:context_MNL} for more details).  As compared to the standard MNL, this \emph{changing contextual MNL} model not only incorporates rich contextual information but also allows the utility to evolve over time. The changing utility is an attractive property as it captures the reality in many applications but also brings new technical challenges in learning and decision-making.
For example, in existing works of \citep{Agrawal16MNLBandit} for plain MNL choice models,
 upper confidence bands are constructed by providing the same assortment repetitively to incoming customers
until a no-purchase activity is observed.
Such an approach, however, can no longer be applied to MNL with changing contextual information as the utility parameters of products constantly evolve with time.
To overcome such challenges, we propose a policy that performs optimization at every single time period, without repetitions of assortments in general.

Our model also allows the revenue for each product $j$ to change over time. In particular, we associate the revenue parameter $r_{tj}$ for the product $j$ at time $t$.

This model generalizes the widely adopted (generalized) linear contextual bandit from machine learning literature (see, e.g., \cite{filippi2010parametric, chu11,Abbasi-Yadkori:2011,Agrawal:13,li2017provably} and references therein) in a non-trivial way since the MNL cannot be written in a generalized linear model form (when
an assortment contains more than one product, see Section \ref{sec:related} for more details). It is also worthwhile noting that this model incorporates a personalized MNL model proposed by \cite{Wang:17:person} as a special case, where each product $j$ is associated with a fixed but unknown coefficient $\theta(j)$ and each arriving customer at time $t$ with an observable feature vector $x_t$ (see Section \ref{sec:related} for a more detailed discussion).  On the other hand, we choose to motivate our model from product contextual information since in practice, obtaining products' features is usually easier (and less sensitive) than extracting customers' preferences.

Given this contextual MNL choice model, the key challenge is how to design a policy that simultaneously learns the unknown coefficient $\theta_0$ and sequentially makes the decision on offered assortment. The performance of a dynamic policy is usually measured by the \emph{regret}, which is defined as the gap between the expected revenue generated by the policy and the oracle expected revenue when $\theta_0$ (and thus the mean utilizes) is known as a priori.

The first contribution of the paper is the construction of an upper confidence bound (UCB) policy. Our UCB policy is based on the maximum likelihood estimator (MLE)  and thus is named MLE-UCB. Although UCB  has been a well-known technique for bandit problems, how to adopt this high-level idea to solve a problem with specific structures certainly requires technical innovations (e.g., how to build a confidence interval varies from one problem to another). In particular, our MLE-UCB contains two stages. The first stage is a \emph{pure exploration stage} in which assortments are randomly offered and  a \emph{``pilot MLE''} is computed based on the observed purchase actions. As we will show in Lemma \ref{lem:pilot}, this pilot estimator serves as a good initial estimator of $\theta_0$. After the exploration phase, the MLE-UCB enters the \emph{simultaneous learning and decision-making phase}. We carefully construct an upper confidence bound of the expected revenue when offering an assortment. The added interval is based on the Fisher information matrix of the computed MLE from the previous step.  Then we solve a combinatorial optimization problem to search the assortment that maximizes the upper confidence bound. By observing the customer's purchase action based on the offered assortment, the policy updates the estimated MLE. In this update, we propose to compute a \emph{``local MLE''}, which requires the solution to be close enough to our pilot estimator. The local MLE plays an important role in MLE-UCB policy since it guarantees that the obtained estimator at each time period is also close to the unknown true coefficient $\theta_0$.

%\xnote{add more details on the policy}

%In our regret analysis, we are mainly interested in the case when the selling horizon $T$ is much larger than the dimensionality $d$ and the size constraint $K$. The large $T$ case is also the regime of common interest in dynamic assortment optimization literature. \footnote{Although we explicitly express our regret in terms of  all problem-related parameters,  considering this regime can help us identify the dominating term in the regret.}

Under some mild assumptions on  features and coefficients, we are able to establish a regret bound $\tilde{O}(d \sqrt{T})$, where the $\tilde{O}$ notation suppresses logarithmic dependence on $T$, $K$ (cardinality constraint), and some other problem dependent parameters\footnote{For the ease of presentation in the introduction, we only present the dominating term under the common scenario that the selling horizon $T$ is larger than the dimensionality $d$ and the cardinality constraint $K$. Please refer to  Theorem \ref{thm:mle-ucb} for a more explicit expression of the obtained regret.}.  One remarkable aspect of our regret bound is that our regret has no dependence on the total number of products $N$ (not even in a logarithmic factor).  This makes the result attractive to online applications where $N$ is large (e.g., online advertisement). Moreover, it is also worthwhile noting the dependence of $K$ is only through a logarithmic term.

Our second contribution is to establish the lower bound result $\Omega(d\sqrt{T}/K)$. When the maximum size of an assortment $K$ is small (which  usually holds in practice), this result shows that our policy is almost optimal.

Moreover, at each time period in the exploitation phase, our UCB policy needs to solve a combinatorial optimization problem, which searches for the best assortment (under the cardinality constraint) that  minimizes the upper confidence bound of the expected revenue. Given the complicated structure of the upper confidence bound, there is no simple solution for this combinatorial problem. When $K$ is small and $N$ is not too large, one can directly search over all the possible sets with the size less than or equal to $K$. In addition to the solution of solving the combinatorial optimization exactly, {the third contribution of the work is to  provide an approximation algorithm based on dynamic programming that runs in polynomial time with respect to $N$, $K$, $T$.}  Although the proposed approximation algorithm has a theoretical guarantee, it is still not efficient for dealing with large-scale applications. To this end, we further describe a computationally efficient greedy heuristic for solving this combinatorial optimization problem.
The heuristic algorithm is based on the idea of local search by greedy swapping, with more details described in Sec.~\ref{subsec:heuristic}.

%\xnote{Need a bit more description on this approximation algorithm.}

%In addition to methodology and theoretical development, we further conduct numerical studies to demonstrate the effectiveness of the proposed policy. In our numerical studies, instead of running the FPTAS algorithm, we proposed an efficient local search heuristic to solve the combinatorial optimization problem.

\subsection{Related work}
\label{sec:related}

Due to the popularity of data-driven revenue management, dynamic assortment optimization, which adaptively learns unknown customers' choice behavior, has received increasing attention in the past few years. Motivated by fast-fashion retailing, the work by \cite{Caro2007} first studied dynamic assortment optimization problem, but it makes a strong assumption that the demands for different product are independent. Recent works by \cite{Rusmevichientong2010, Saure2013, Agrawal16MNLBandit, Agrawal17Thompson,Chen:18tight,Chen:18near} incorporated  MNL models into dynamic assortment optimization and formulated the problem into a online regret minimization problem. In particular, for capacitated MNL,  \cite{Agrawal16MNLBandit} and \cite{Agrawal17Thompson} proposed  UCB and Thompson sampling techniques  and established the regret bound $\tilde{O}(\sqrt{NT})$ (when $T  \gg N^2$). \cite{Chen:18tight} further established a matching lower bound of $\Omega(\sqrt{NT})$. It is interesting to compare our regret to the bound for the standard MNL case. When the total number of products $N$ is much larger than $d$ (i.e., $N > d^2$), by incorporating the contextual information, the regret reduces from $\tilde{O}(\sqrt{NT})$ to  $\tilde{O}(d \sqrt{T})$. The latter one only depends on $d$ and is completely independent of the total number of products $N$, which also demonstrates the usefulness of the contextual  information. \cite{Chen:18dynamic} further studied the dynamic assortment optimization under nested logit models. We also note that to highlight our key idea and focus on the balance between learning of $\theta_0$ and revenue maximization, we study the stylized dynamic assortment optimization problems following the existing literature \citep{Rusmevichientong2010, Saure2013, Agrawal16MNLBandit, Agrawal17Thompson}, which ignore operations considerations such as price decisions and inventory replenishment.

There is another line of recent research on investigating personalized assortment optimization. {By incorporating the feature information of each arriving customer,  both the static and dynamic assortment optimization problems are studied in \cite{Chen2015} and \cite{Wang:17:person}, respectively. It is worthwhile noting that although we do not approach our work from a personalized perspective\footnote{This is because in some applications, the product features are easier to obtain by the seller as compared to customer features.}, the \emph{personalized MNL} considered in \cite{Wang:17:person} can be viewed as a special of our model.
	
	In particular, the personalized MNL assumes that each product $j$ is associated with an unknown coefficient $\theta(j) \in \mathbb{R}^D$. When a customer arrives at time $t$ with the \emph{observed feature $x_t$}, the utility of product $j$ at time $t$ is $u_{tj}=x_t^\top \theta(j)$. Now we explain how to specialize our model to obtain the personalized MNL. Let us define $\theta_0:=\{\theta(1), \ldots, \theta(N)\} \in \mathbb{R}^{DN}$ and the feature vector $v_{tj}:=(0, \ldots, x_{t}, \ldots, 0) \in \mathbb{R}^{DN}$, which is a concatenation of $N$ $D$-dimensional vectors with the $j$-th vector being $x_t$ and all other vectors being 0. Then according to our linear model in Eq.~\eqref{eq:linear}, we have $u_{tj}= v_{tj}^\top\theta_0= x_t^\top \theta(j)$, which recovers the personalized MNL model. Using our regret bound $\tilde{O}(d\sqrt{T})$ with $d=DN$ as the dimensionality of $\theta_0$, we directly obtain the regret  $\tilde{O}(DN \sqrt{T})$ for the dynamic assortment optimization under the personalized MNL. As compared to the Bayesian regret bound $\tilde{O}(DN\sqrt{KT})$ in \cite{Wang:17:person} (see Theorem 3.3. therein), our approach still saves a factor of $\sqrt{K}$.}
We also remark that our results require a slightly stronger assumption on the contextual information vectors $\{v_{tj}\}$ compared to \cite{Wang:17:person}, which allows customer feature vectors $\{x_t\}$ to be  adversarially chosen. More specifically, {a stochastic assumption} is imposed on $\{v_{tj}\}$ only during the \emph{pure exploration} phase of our proposed policy. {After this pure exploration phase, the feature vectors $\{v_{tj}\}$ can also be adversarially chosen.}
We refer the readers to Sec.~\ref{subsec:regret} for further details.

In addition, the developed techniques in our work and \cite{Wang:17:person} are different, Our policy is based on UCB, while the policy in \cite{Wang:17:person} is based on Thompson sampling. Furthermore, other research studies personalized assortment optimization in an adversarial setting rather than stochastic setting.  For example, \cite{Golrezaei2014,Chen:16recom} assumed that each customer's choice behavior is known, but that the customers' arriving sequence (or customers' types) can be adversarially chosen and took the inventory level into consideration. Since the arriving sequence can be arbitrary, there is no learning component in the problem and both \cite{Golrezaei2014} and \cite{Chen:16recom} adopted the competitive ratio as the performance evaluation metric.

Another field of related research is the \emph{contextual bandit} literature, in which the  linear contextual bandit has been widely studied as a special case (see, e.g., \cite{Dani08stochastic,Rusmevichientong:10:bandit,chu11,Abbasi-Yadkori:2011,Agrawal:13} and references therein). Some recent work extends the linear contextual bandit to  \emph{generalized linear bandit} \citep{filippi2010parametric,li2017provably}, which assumes a generalized linear reward structure. In particular, the reward $r$ of pulling an arm given the observed feature vector of this arm $x$ is modeled by
\begin{equation}
\mathbb E[r|x] = \sigma(x^\top\theta_0),
\label{eq:glm}
\end{equation}
for an unknown linear model $\theta_0$ and a known link function $\sigma:\mathbb R\to\mathbb R$. For example, for a linear contextual bandit, $\sigma$ is the identity mapping, i.e., $\mathbb E[r|x]=\sigma(x^\top\theta_0)$. For the logistic contextual bandit, we have $r \in \{0,1\}$ and $\Pr(r=1|x)=\frac{\exp(x^\top\theta_0)}{1+\exp(x^\top\theta_0)}$.  { In a standard generalized linear bandit problem (see, e.g., \cite{li2017provably}) with $N$ arms, it is assumed that a context vector $v_{tj}$ is revealed at time $t$ for each arm $j \in [N]$. Given a selected arm $i_t \in [N]$ at time $t$, the expected reward follows Eq.~\eqref{eq:glm}, i.e., $\mathbb{E}[r_t | v_{t,i_t} ]=\sigma(v_{t,i_t}^\top\theta_0)$. At first glance, our contextual MNL model is a natural extension of the generalized linear bandit to the MNL choice model. However, when the size of an assortment $K\geq 2$, the contextual MNL \emph{cannot} be written in the form of Eq.~(\ref{eq:glm}) and the denominator in the choice probability (see Eq.~\eqref{eq:context_MNL} in the next section) has a more complicated structure. Therefore, our problem is technically \emph{not} a generalized linear model and is therefore more challenging.
Moreover, in contextual bandit problems, only one arm is selected by the decision-maker at each time period. In contrast, each action in an assortment optimization problem involves a set of items, which makes the action space more complicated.}

\subsection{Notations and paper organization}

Throughout the paper, we adopt the standard asymptotic notations. In particular, we use $f(\cdot) \lesssim g(\cdot)$ to denote that $f(\cdot) = O(g(\cdot))$. Similarly, by $f(\cdot) \gtrsim g(\cdot)$, we denote $f(\cdot) = \Omega(g(\cdot))$. We also use $f(\cdot) \asymp g(\cdot)$ for $f(\cdot) = \Theta(g(\cdot))$. Throughout this paper, we will use $C_0, C_1, C_2，\ldots$ to denote universal constants. For a vector $v$ and a matrix $M$, we will use $\|v\|_2$ and $\|M\|_{\op}$ to denote the vector $\ell_2$-norm and the matrix spectral norm (i.e., the maximum singular value), respectively. Moreover, for a real-valued symmetric matrix $M$, we denote the maximum eigenvalue and the minimum eigenvalue of $M$ by $\lambda_{\max}(M)$ and $\lambda_{\min}(M)$, respectively, and define $\|v\|_M^2:=v^T M v$ for any given vector $v$. For a given integer $N$, we denote the set $\{1,\ldots, N\}$ by  $[N]$. %\xnote{any other notations?}

The rest of the paper is organized as follows. In Section \ref{sec:model}, we introduce the mathematical formulation of our models and define the regret. In Section \ref{sec:UCB}, we describe the proposed MLE-UCB policy and provide the regret analysis. The lower bound result is provided in Section \ref{sec:lower}. In Section \ref{sec:comb}, we investigate the combinatorial optimization problem in MLE-UCB and propose the approximation algorithm and greedy heuristic. The multivariate case of the approximation algorithm is relegated to the appendix. In Section \ref{sec:numerical}, we provide the numerical studies. The conclusion and future directions are discussed in Section \ref{sec:conclusion}. Some technical proofs are provided in the online supplementary material.

\section{The problem setup}
\label{sec:model}

There are $N$ items, conveniently labeled as $1,2,\cdots,N$.
%Each item is associated with a known revenue parameter $r_i$, normalized so that $r_1,\cdots,r_N\in[0,1]$.
At each time $t$, a set of time-sensitive ``feature vectors'' $v_{t1}, v_{t2}, \cdots, v_{tN}\in\mathbb R^d$
and revenues $r_{t1},\cdots,r_{tN}\in[0,1]$
are observed, reflecting time-varying changes of items' revenues and customers' preferences.
A retailer, based on the features $\{v_{ti}\}_{i=1}^N$ and previous purchasing actions, picks an assortment $S_t\subseteq[N]$ under the cardinality constraint $|S_t|\leq K$ to present
to an incoming customer;
the retailer then observes a purchasing action $i_t\in S_t\cup\{0\}$ and collects the associated revenue $r_{i_t}$ of the purchased item
(if $i_t=0$ then no item is purchased and zero revenue is collected).

We use an MNL model with features to characterize how a customer makes choices.
Let $\theta_0\in\mathbb R^d$ be an \emph{unknown} time-invariant coefficient.
For any $S\subseteq[N]$, the choice model $p_{\theta_0,t}(\cdot|S)$ is specified as (let $r_0=0$ and $v_{t0}=0$)
\begin{equation}
p_{\theta_0,t}(j|S) = \frac{\exp\{v_{tj}^\top\theta_0\}}{1+\sum_{k\in S} \exp\{v_{tk}^\top\theta_0\}} \;\;\;\;\;\;\forall j\in S\cup\{0\}.
\label{eq:context_MNL}
\end{equation}
For simplicity, in the rest of the paper we use $p_{\theta,t}(\cdot|S)$ to denote the law of the purchased item  $i_t$ conditioned on given assortment $S$
at time $t$, parameterized by the coefficient $\theta\in\mathbb R^d$.
The expected revenue $R_t(S)$ of assortment $S\subseteq[N]$ at time $t$ is then given by
\begin{equation}\label{eq:rev}
R_t(S) := \mathbb E_{\theta_0,t}[r_{tj}|S] = \frac{\sum_{j\in S}r_{tj}\exp\{v_{tj}^\top\theta_0\}}{1+\sum_{j\in S}\exp\{v_{tj}^\top \theta_0\}}.
\end{equation}
{Note that throughout the paper,  we use $\mathbb E_{\theta_0,t}[\cdot |S]$  to denote the expectation with respect to the choice probabilities $p_{\theta_0,t}(j|S)$ defined in Eq.~\eqref{eq:context_MNL}.}

Our objective is to design policy $\pi$ such that the regret
\begin{equation}\label{eq:regret}
\mathrm{Regret}(\{S_t\}_{t=1}^T) =\mathbb E^\pi\sum_{t=1}^T R_t(S_t^*) - R_t(S_t) \;\;\;\;\;\text{where}\;\;S_t^*=\arg\max_{S\subseteq[N], |S|\leq K} R_t(S)
\end{equation}
is minimized. {Here, $S_t^*$ is an optimal assortment chosen when the full knowledge of choice probabilities is available (i.e., $\theta_0$ is known).}

\SetAlCapSkip{1em}
\section{An MLE-UCB policy and its regret}
\label{sec:UCB}
\begin{algorithm}[t]
	\KwInput{Number of pure explorations $T_0$, constraint radius $\tau$.}
	\KwOutput{Assortment selections $\{S_t\}_{t=1}^T\subseteq[N]$ satisfying $|S_t|\leq K$.}
	
	{{Pure exploration}}: for $t=1,\cdots,T_0$, pick $S_t=\{\ell_t\}$ for a single product $\ell_t$ sampled uniformly at random from $\{1,\cdots,N\}$
	and record purchasing actions $(i_1,\cdots,i_{T_0})$\;
	
	Compute a pilot estimator using global MLE: $\theta^*\in\arg\max_{\theta\in\mathbb R^d} \sum_{t'=1}^{T_0} \log p_{\theta,t}(i_{t'}|S_{t'})$\;
	
	\For{$t=T_0+1$ to $T$}{
		Observe revenue parameters $\{r_{tj}\}_{j=1}^N$ and preference features $\{v_{tj}\}_{j=1}^N$ at time $t$\;
		Compute local MLE $\hat\theta_{t-1} \in \arg\max_{\|\theta-\theta^*\|_2\leq\tau} \sum_{t'=1}^{t-1}\log p_{\theta,t}(i_{t'}|S_{t'})$\;
		For every assortment $S\subseteq[N]$, $|S|\leq K$, compute its upper confidence bound
		\begin{align*}
		&\bar R_t(S) :=\mathbb E_{\hat\theta_{t-1},t}[r_{tj}|S] + \min\left\{1, \omega\sqrt{\|\hat I_{t-1}^{-1/2}(\hat\theta_{t-1}) \hat M_t(\hat\theta_{t-1}|S) \hat I_{t-1}^{-1/2}(\hat\theta_{t-1})\|_\op}\right\};\\
		&\hat I_{t-1}(\theta) := \sum_{t'=1}^{t-1}\hat M_{t'}(\theta|S_{t'}); \;\; \hat M_t(\theta|S) := \mathbb E_{\theta,t}[v_{tj}v_{tj}^\top|S] - \{\mathbb E_{\theta,t}[v_{tj}|S]\}\{\mathbb E_{\theta,t}[v_{tj}|S]\}^\top;\\
		&\omega \asymp \sqrt{d\log(\rho\nu TK)};
		\end{align*}\label{alg:step6}
		
		Pick $S_t\in\arg\max_{S\subseteq[N], |S|\leq K} \bar R_t(S)$ and observe purchasing action $i_t\in S_t\cup\{0\}$\;
	}
	
	Remark: the expectations admit the following closed-form expressions:\newline
	$\mathbb E_{\theta,t}[r_{tj}|S] = \sum_{j\in S} p_{\theta,t}(j|S)r_{tj} =\frac{\sum_{j\in S}r_{tj}\exp\{v_{tj}^\top\theta\}}{1+\sum_{j\in S}\exp\{v_{tj}^\top\theta\}}$;\newline
	$\mathbb E_{\theta,t}[v_{tj}|S] = \sum_{j\in S}p_{\theta,t}(j|S)v_{tj} = \frac{\sum_{j\in S}v_{tj}\exp\{v_{tj}^\top\theta\}}{1+\sum_{j\in S}\exp\{v_{tj}^\top\theta\}}$;\newline
	$\mathbb E_{\theta,t}[v_{tj}v_{tj}^\top|S] = \sum_{j\in S}p_{\theta,t}(j|S)v_{tj}v_{tj}^\top = \frac{\sum_{j\in S}v_{tj}v_{tj}^\top\exp\{v_{tj}^\top\theta\}}{1+\sum_{j\in S}\exp\{v_{tj}^\top\theta\}}$.
	%and similarly for
	%$\mathbb E_{\theta,t}[v_{tj}|S]$ and $\mathbb E_{\theta,t}[v_{tj}v_{tj}^\top|S]$.

	\caption{The MLE-UCB policy for dynamic assortment optimization with changing features}
	\label{alg:mle-ucb}
\end{algorithm}

We propose an MLE-UCB policy, described in Algorithm \ref{alg:mle-ucb}.

The policy can be roughly divided into two phases.
In the first \emph{pure exploration} phase,
the policy selects assortments uniformly at random, consisting of only \emph{one item}.
The objective of the pure exploration is to establish a ``pilot'' estimator of the unknown coefficient $\theta_0$, i.e., a good initial estimator for $\theta_0$.
For the simplicity of the analysis, we choose one item for each assortment in this phase, which facilitates us to adapt existing analysis in \citep{filippi2010parametric,li2017provably}
as the MNL-logit choice model reduces to a \emph{generalized linear model} when only one item is present in the assortment.
In the second phase, we use a UCB-type approach that selects $S_t$ as the assortment maximizing an \emph{upper bound} $\overline R_t(S_t)$
of the expected revenue $R_t(S_t)$.
Such upper bounds are built using a \emph{local Maximum Likelihood Estimation (MLE)} of $\theta_0$. In particular, in Step 5, instead of computing an MLE, we compute a local MLE, where the estimator $\hat{\theta}_{t-1}$ lies in a ball centered at the pilot estimator $\theta^*$ with a radius $\tau$.
This localization also simplifies the technical analysis based on Taylor expansion,
which benefits from the constraint that $\hat\theta_{t-1}$ is not too far away from $\theta^*$.

To construct the confidence bound, we introduce the matrices $\hat M_t(\hat\theta_{t-1}|S)$ and $\hat I_{t-1}(\hat\theta_{t-1})$ in Step 6 of Algorithm \ref{alg:mle-ucb}, which are empirical estimates of the \emph{Fisher's information} matrices $-\mathbb E[\nabla_\theta^2\log p(\cdot|\theta)]$
corresponding to the MNL choice model $p(\cdot|S_t)$. The population version of the Fisher's information matrices are presented in Eq.~\eqref{eq：pop_fisher} in Sec.~\ref{sec:ana_local}.
These quantities play an essential role in classical statistical analysis of maximum likelihood estimators (see, e.g., \citep{van1998asymptotic}).

The proposed MLE-UCB policy has three hyper-parameters: the coefficient $\omega>0$ that controls the lengths of confidence intervals of $R_t(S)$,
the number of pure exploration iterations $T_0$, and the radius $\tau_0$ in the local MLE formulation.
While theoretical values of $\omega,T_0$ and $\tau$ are given in Theorem \ref{thm:mle-ucb}, which potentially depend on several unknown problem parameters,
in practice we recommend the usage of $T_0=\max\{d\log T, T^{1/4}\}$, $\omega=\sqrt{d\log T}$ and $\tau=1/K$.

%\xnote{better to add a bit more explanations of the Algorithm}
%{\red For any given $\theta$,  the vector $\mathbb E_{\theta,t}[v_{tj}|S] \in \mathbb{R}^d$ and the matrix $\mathbb E_{\theta,t}[v_{tj}v_{tj}^\top|S] \in \mathbb{R}^{d \times d}$ defined in  Algorithm \ref{alg:mle-ucb} can be explicitly expressed  as follows,
%\[
%\mathbb E_{\theta,t}[v_{tj}|S]= \sum_{j=1}^N \frac{\exp\{v_{tj}^\top\theta\}}{1+\sum_{k\in S} \exp\{v_{tk}^\top\theta\}} v_{tj}, \quad
%\mathbb E_{\theta,t}[v_{tj} v_{tj}^{\top}|S]= \sum_{j=1}^N \frac{\exp\{v_{tj}^\top\theta\}}{1+\sum_{k\in S} \exp\{v_{tk}^\top\theta\}} v_{tj} v_{tj}^{\top}.
%\]
%}

In the rest of this section, we give a regret analysis that shows an $\tilde O(d\sqrt{T})$ upper bound on the regret of the MLE-UCB policy.
Additionally, we prove a lower bound of $\tilde O(d\sqrt{T}/K)$ in Sec.~\ref{sec:negative} and show how the combinatorial optimization in Step 7 can be approximately computed efficiently
in Sec.~\ref{sec:comb}.

\subsection{Regret analysis}\label{subsec:regret}

%Next, we provide the regret analysis of the proposed MLE-UCB policy. We make the following assumptions:
To establish rigorous regret upper bounds on Algorithm \ref{alg:mle-ucb}, we impose the following assumptions:
\begin{enumerate}[leftmargin=0.5in]
	\item[(A1)] There exists a constant $\nu$ such that $\|v_{tj}\|_2\leq\nu$ for all $t$ and $j$.
	Moreover, for all $t\leq T_0$ and $j \in [N]$, $v_{tj}$ are i.i.d.~generated from an unknown distribution with the density $\mu$ satisfying that
	$\lambda_{\min}(\mathbb E_\mu vv^\top) \geq \lambda_0$ for some constant $\lambda_0>0$;
	%for any measurable set $A\subseteq\mathbb R^d$ with zero Lebesgue measure, $\mu(A)=0$.
	
	\item[(A2)] There exists a constant $\rho<\infty$ such that for all $t\in[T]$ and $S\subseteq[N]$ {with $|S|\leq K$}, $\frac{p_{\theta_0,t}(j|S)}{p_{\theta_0,t}(j'|S)} \leq \rho$
	for all $j,j'\in S\cup\{0\}$.
\end{enumerate}

%\xnote{(A1): I feel that we should relax the assumption of (A1) to the feature vector only in the exploration phase. Also, in the introduction, we claim we are generalizing contextual bandit and David's paper, but they allow adversarial feature. (A2)  I add another cardinality constraint. Please check}

The item (A1) assumes that the contextual information vectors $\{v_{tj}\}$  in the pure-exploration
phase with $t \leq T_0$ are \emph{randomly generated} from a non-degenerate density.
It also places a standard boundedness condition on $\{v_{tj}\}$ for all time periods $t$.
Note that after the pure-exploration phase, we allow the contextual vectors $\{v_{tj}\}$ to be adversarially chosen,
only subject to boundedness conditions.
(A2) additionally assumes a bounded ratio between the probability of choosing any two different items in an arbitrary assortment set. {We remark that if $\|\theta_0\|_2\leq C$, then the boundedness assumption in (A1) implies (A2) with $\rho\leq e^{2 \max\{1,C\nu\}}$.}

We are now ready to state our main result that upper bounds the worst-case accumulated regret of the proposed MLE-UCB policy
in Algorithm \ref{alg:mle-ucb}.

\begin{theorem}
	Suppose that $T_0\asymp \max\{\nu^2d\log T/\lambda_0^2, \rho^2(d+\log T)/(\tau^2\lambda_0)\}$ and $\tau\asymp 1/\sqrt{\rho^2\nu^2K^2}$ in Algorithm \ref{alg:mle-ucb}, then the
	regret of the MLE-UCB policy is upper bounded by
	\begin{equation}\label{eq:thm-main}
	C_1\left[ d\sqrt{T}\cdot \log(\lambda_0^{-1}\rho\nu TK)+ {d^2\lambda_0^{-2}\rho^4\nu^2 K^2\log T}\right] + C_2,
	\end{equation}
	where $C_1, C_2>0$ are universal constants.
	\label{thm:mle-ucb}
\end{theorem}

In addition to universal constants, the regret upper bound established in Theorem \ref{thm:mle-ucb} has two terms.
The first term, $d\sqrt{T}\cdot \log(\lambda_0^{-1}\rho\nu TK)$, is the main regret term that scales as $\tilde O(d\sqrt{T})$ dropping logarithmic dependency.
The second $d^2\lambda_0^{-2}\rho^4\nu^2K^2\log T$ term is a minor term,  because it only scales logarithmically with the time horizon $T$.
One remarkable aspect of Theorem \ref{thm:mle-ucb} is the fact that the regret upper bound has \emph{no} dependency on the total number of items $N$ (even in a logarithmic term).
This is an attractive property of the proposed policy, which allows $N$ to be very large, even exponentially large in $d$ and $K$.

%In the next subsection, we provide the proof of our main theorem in Theorem \ref{thm:mle-ucb}. The proofs of some technical lemmas are relegated to the online supplement.

\subsection{Proof sketch of Theorem \ref{thm:mle-ucb}}

We provide a proof sketch of Theorem \ref{thm:mle-ucb} in this section. The proofs of technical lemmas are relegated to the online supplement.

The proof is divided into four steps.
In the first step, we analyze the pilot estimator $\theta^*$ obtained from the pure exploration phase of Algorithm \ref{alg:mle-ucb},
and show as a corollary that the true model $\theta_0$ is feasible to all subsequent local MLE formulations with high probability (see Corollary \ref{cor:tau}).
In the second step, we use an $\varepsilon$-net argument to analyze the estimation error of the local MLE. % \xnote{should we change empirical process}
Afterwards, we show in the third step that an upper bound on the estimation error $\hat\theta_{t-1}-\theta_0$ implies
an upper bound on the estimation error of the expected revenue $R_t(S)$, hence showing that $\bar R_t(S)$ are valid upper confidence bounds.
Finally, we apply the \emph{elliptical potential lemma}, which also plays a key role in linear stochastic bandit and its variants,
to complete our proof.

%For the ease of presentation, we will decompose the proofs into several parts.
%We will first analyze the pure exploration and the property of the pilot estimator. \xnote{helpful to add some overview of the proof}
%In this section we prove our main regret upper bound in Theorem \ref{thm:mle-ucb}.

\subsubsection{Analysis of pure exploration and the pilot estimator}

Our first step is to establish an upper bound on the estimation error $\|\theta^*-\theta_0\|_2$ of the pilot estimator $\theta^*$,
built using pure exploration data.
It should be noted that in the pure exploration phase ($t\in\{1,\cdots,T_0\}$),
the assortments $\{S_t\}_{t=1}^{T_0}$ only consist of one item. Therefore the observation model reduces to
a standard \emph{generalized linear model} with the sigmoid function $\sigma(x)=1/(1+e^{-x})= e^x/(1+e^x)$ as the link function, which is essentially
a logistic  regression model of observing 1 if the customer makes a purchase.
%This allows us to cite existing works to upper bound the error $\|\theta^*-\theta_0\|_2$.

Because the choice model in the pure exploration phase reduces to a generalized linear model, we can cite existing works to upper bound the error $\|\theta^*-\theta_0\|_2$.
In particular, the following lemma is cited from \citep[Eq.~(18)]{li2017provably}, adapted to our model and parameter settings.
The details on how to adapt the result from \citep{li2017provably} provided in the supplementary material.
\begin{lemma}
	With probability $1-\delta$ it holds that
	\begin{equation}
	\|\theta^*-\theta_0\|_2 \leq \frac{2}{\kappa}\sqrt{\frac{d+\log(1/\delta)}{\lambda_{\min}(V)}}
	\;\;\;\;\text{where}\;\;\kappa = \frac{1}{2e(1+\rho)}\;\;\text{and}\;\; V = \sum_{t=1}^{T_0}v_{t,i_t}v_{t,i_t}^\top.
	\end{equation}
	\label{lem:pilot}
\end{lemma}

The following corollary immediately follows Lemma \ref{lem:pilot},
by lower bounding $\lambda_{\min}(V)$ using standard matrix concentration inequalities.
Its proof is again deferred to the supplementary material.

\begin{corollary}
	There exists a universal constant $C_0>0$ such that for arbitrary $\tau\in(0,1/2]$,
	if $T_0 \geq C_0\max\{\nu^2d\log T/\lambda_0^2, \rho^2 (d+\log T)/(\tau^2\lambda_0)\}$ then with probability $1-O(T^{-1})$, $\|\theta^*-\theta_0\|_2\leq \tau$.
	\label{cor:tau}
\end{corollary}

The purpose of Corollary \ref{cor:tau} is to establish a connection between the number of pure exploration iterations $T_0$ and
the critical radius $\tau$ used in the local MLE formulation.
It shows a lower bound on $T_0$ in order for the estimation error $\|\theta^*-\theta_0\|_2$ to be upper bounded by $\tau$ with high probability,
which certifies that the true model $\theta_0$ is also a \emph{feasible} local estimator in our MLE-UCB policy.
This is an important property for later analysis of local MLE solutions $\hat\theta_{t-1}$.

\subsubsection{Analysis of the local MLE}
\label{sec:ana_local}
The following lemma upper bounds a Mahalanobis distance between $\hat\theta_t$ and $\theta_0$.
For convenience, we adopt the notation that $r_{t0}=0$ and $v_{t0}=0$ for all $t$ throughout this section.
We also define
\begin{eqnarray}
I_t(\theta) & := &\sum_{t'=1}^t M_{t'}(\theta), \label{eq：pop_fisher} \\
M_{t'}(\theta)& := &-\mathbb E_{\theta_0,t'}[\nabla^2_\theta \log p_{\theta,t'}(j|S_{t'})]  \nonumber  \\
& =  &\mathbb E_{\theta_0,t'}[v_{t'j}v_{t'j}^\top] - \{\mathbb E_{\theta_0,t'}v_{t'j}\}\{\mathbb E_{\theta,t'}v_{t'j}\}^\top
- \{\mathbb E_{\theta,t'}v_{t'j}\} \{\mathbb E_{\theta_0,t'}v_{t'j}\}^\top  +  \{\mathbb E_{\theta,t'}v_{t'j}\}\{\mathbb E_{\theta,t'}v_{t'j}\}^\top \nonumber
\end{eqnarray}
where $\mathbb E_{\theta,t'}$ denotes the expectation evaluated under the law $j\sim p_{\theta,t'}(\cdot|S_{t'})$;
that is, $p_{\theta,t'}(j|S_{t'}) =\exp\{v_{t'j}^\top\theta\}/(1+\sum_{k\in S_{t'}}\exp\{v_{t'j}^\top\theta\})$ for $j\in S_{t'}$
and $p_{\theta,t'}(j|S_{t'})=0$ for $j\notin S_{t'}$.

\begin{lemma}
	Suppose $\tau\leq 1/\sqrt{8\rho \nu^2K^2}$. Then there exists a universal constant $C>0$ such that with probability $1-O(T^{-1})$ the following holds uniformly over all $t=T_0,\cdots,T-1$:
	\begin{equation}
	(\hat\theta_t-\theta_0)^\top I_t(\theta_0)(\hat\theta_t-\theta_0) \leq C\cdot d\log(\rho\nu TK).
	\end{equation}
	%`Here we define
	%is defined in the description of Algorithm \ref{alg:mle-ucb}.
	\label{lem:mle}
\end{lemma}
\begin{remark}
	For $\theta=\theta_0$, the expression of $M_{t'}(\theta)$ can be simplified as $M_{t'}(\theta_0) = \mathbb E_{\theta_0,t'}[v_{t'j}v_{t'j}^\top] - \{\mathbb E_{\theta_0,t'}v_{t'j}\}\{\mathbb E_{\theta_0,t'}v_{t'j}\}^\top$.
\end{remark}

The complete proof of Lemma \ref{lem:mle} is given in the supplementary material,
and here we provide some high-level ideas behind our proof.

Our proof is inspired by the classical convergence rate analysis of M-estimators \citep[Sec.~5.8]{van1998asymptotic}.
The main technical challenge is to provide \emph{finite-sample} analysis of several components in the proof of \citep[Sec.~5.8]{van1998asymptotic}.

In particular, for any $\theta\in\mathbb R^d$, consider
$$
F_{t}(\theta) := \sum_{t'\leq t}f_{t'}(\theta)\;\;\;\;\text{where}\;\;
f_{t'}(\theta) := \mathbb E_{\theta_0,t'}\left[\log\frac{p_{\theta,t'}(j|S_{t'})}{p_{\theta_0,t'}(j|S_{t'})}\right] = \sum_{j\in S_{t'}\cup\{0\}} p_{\theta_0,t'}(j|S_{t'})\log\frac{p_{\theta,t'}(j|S_{t'})}{p_{\theta_0,t'}(j|S_{t'})}
$$
and its ``sample'' version
$$
\hat F_t(\theta) := \sum_{t'\leq t} \hat f_{t'}(\theta)\;\;\;\;\text{where}\;\;
\hat f_{t'}(\theta) := \log\frac{p_{\theta,t'}(i_{t'}|S_{t'})}{p_{\theta_0,t'}(i_{t'}|S_{t'})}.
$$

It is easy to verify by definition that $F_t(\hat\theta_{t})\geq F_t(\theta_0)=0$ and $\hat F_t(\hat\theta_{t})\leq  \hat{F}_t(\theta_0)=0$, because
$F_t(\cdot)$ is a Kullback-Leibler divergence,
$\theta_0$ is feasible to the local MLE formulation
and $\hat\theta_{t-1}$ is the optimal solution.
On the other hand, it can be proved that $|F_t(\theta)-\hat F_t(\theta)|$ is small for all $\theta$ with high probability,
by using concentration inequalities for self-normalized empirical process (note that $\mathbb E\hat f_{t'}(\theta) = f_{t'}(\theta)$ for any $\theta$). Moreover, by constructing a local quadratic approximation of $F_t(\cdot)$ around $\theta_0$, we can show that $F_t(\theta)-F_t(\theta_0)$ is large when $\theta$ is far away from $\theta_0$.

{Following the above observations, we can use proof by contradiction to prove Lemma \ref{lem:mle}, which essentially claims that  $\hat\theta_{t}$ and $\theta_0$  are close under the quadratic distance $\|\cdot\|_{I_t(\theta_0)}$.
	Suppose by contradiction that $\hat\theta_{t}$ and $\theta_0$ are far apart, which implies that $|F_t(\hat\theta_{t})-F_t(\theta_0)|$ is large. On the other hand, by the fact that   $\hat F_t(\hat\theta_{t})\leq  0 = {F}_t(\theta_0) \leq F_t(\hat\theta_{t})$, we have
	\[
	|F_t(\hat\theta_{t})-F_t(\theta_0)| =    |F_t(\hat\theta_{t})|  \leq |F_t(\hat\theta_{t})-\hat F_t(\hat\theta_{t})  |.
	\]
	By the established concentration result, we have  $|F_t(\theta)-\hat F_t(\theta)|$ is small for all $\theta$ with high probability (including $\theta=\hat\theta_{t}$).
	This leads to the desired contradiction.}

%\xnote{I change $\hat\theta_{t-1}$ to $\hat\theta_{t}$.}

%On the other hand, because $F_t(\theta_0)=\hat F_t(\theta_0)=0$, $F_t(\hat\theta_{t-1})\geq 0$, $\hat F_t(\hat\theta_{t-1})\leq 0$ and $|F_t(\theta)-\hat F_t(\theta)|$ is small (with high probability)
%for both $\theta=\theta_0$ and $\theta=\hat\theta_{t-1}$, we conclude that $|F_t(\hat\theta_{t-1})-F_t(\theta_0)|$ must also be small.
%These properties lead to {\red a desired contradiction} and complete the proof of Lemma \ref{lem:mle}.

%\xnote{I feel that this paragraph needs to be expanded. For  example, the readers will be confused about ``what is a contradiction here?''}

\subsubsection{Analysis of upper confidence bounds}

The following technical lemma shows that the upper confidence bounds constructed in Algorithm \ref{alg:mle-ucb}
are valid with high probability.
Additionally, we establish an upper bound on the discrepancy between $\bar R_t(S)$ and the true value $R_t(S)$  defined in Eq.~\eqref{eq:rev}.
\begin{lemma}
	Suppose $\tau$ satisfies the condition in Lemma \ref{lem:mle}.
	With probability $1-O(T^{-1})$ the following holds uniformly for all $t>T_0$ and $S\subseteq[N]$, $|S|\leq K$ such that
	\begin{enumerate}
		\item $\bar R_t(S)\geq R_t(S)$;
		\item $|\bar R_t(S)-R_t(S)| \lesssim \min\{1,\omega\sqrt{\|I_{t-1}^{-1/2}(\theta_0) M_t(\theta_0|S) I_{t-1}^{-1/2}(\theta_0)\|_\op}\}$.
	\end{enumerate}
	\label{lem:ucb}
\end{lemma}

At a higher level, the proof of Lemma \ref{lem:ucb} can be regarded as a ``finite-sample'' version of the classical \emph{Delta's method},
which upper bounds estimation error of some functional $\varphi$ of parameters, i.e.,  $|\varphi(\hat\theta_{t-1})-\varphi(\theta_0)|$
using the estimation error of the parameters themselves $\hat\theta_{t-1}-\theta_0$.  The complete proof is relegated to the supplementary material.

\subsubsection{The elliptical potential lemma}\label{sec:ellipitical}

%We analyze the regret accumulated by the main components $T>T_0$ of Algorithm \ref{alg:mle-ucb},
%conditioned on the success event in Lemma \ref{lem:ucb}, which occurs with probability $1-O(T^{-1})$.

Let $S_t^*$ be the assortment that maximizes the expected revenue $R_t(\cdot)$ (defined in Eq.~\eqref{eq:rev}) at time period $t$,
and $S_t$ be the assortment selected by Algorithm \ref{alg:mle-ucb}.
Because $R_t(S)\leq \bar R_t(S)$ for all $S$ (see Lemma \ref{lem:ucb}), we have the following upper bound for each term in the regret (see Eq.~\eqref{eq:regret}):
\begin{equation}
R_t(S_t^*)-R_t(S_t)
\leq (\bar R_t(S_t^*)-\bar R_t(S_t)) + (\bar R_t(S_t)-R_t(S_t)) \leq \bar R_t(S_t)-R_t(S_t),
\label{eq:ucb-triangle}
\end{equation}
where the last inequality holds because $\bar R_t(S_t^*)-\bar R_t(S_t) \leq 0$ (note that $S_t$  maximizes $\bar{R}_t(\cdot)$).

Subsequently, invoking Lemma \ref{lem:ucb} and the Cauchy-Schwarz inequality, we have
\begin{align}
\sum_{t=T_0+1}^T& R_t(S_t^*)-R_t(S_t)
\lesssim \sqrt{d\log(\rho\nu TK)}\cdot \sum_{t=T_0+1}^T \sqrt{\min\{1,\|I_{t-1}^{-1/2}(\theta_0)M_t(\theta_0|S_t)I_{t-1}^{-1/2}(\theta_0)\|_\op\}}\nonumber\\
&\lesssim \sqrt{dT\log(\rho\nu TK)\cdot\sum_{t=T_0+1}^T\min\{1,\|I_{t-1}^{-1/2}(\theta_0)M_t(\theta_0|S_t)I_{t-1}^{-1/2}(\theta_0)\|_\op^2 \}}.\label{eq:ucb-regret-decomp}
\end{align}

The following lemma is a key result that upper bounds $\sum_{t=T_0+1}^T\min\{1,\|I_{t-1}^{-1/2}(\theta_0)M_t(\theta_0|S_t)I_{t-1}^{-1/2}(\theta_0)\|_\op^2\}$.
It is usually referred to as the \emph{elliptical potential lemma} and has found many applications in contextual bandit-type problems (see, e.g., \cite{Dani08stochastic,Rusmevichientong2010,filippi2010parametric,li2017provably}).

\begin{lemma}
	%With probability $1-O(T^{-1})$ it holds that
	It holds that
	\begin{equation*}
	\sum_{t=T_0+1}^T\min\{1,\|I_{t-1}^{-1/2}(\theta_0)M_t(\theta_0|S_t)I_{t-1}^{-1/2}(\theta_0)\|_\op^2\} \leq
	4\log\frac{\det I_{T}(\theta_0)}{\det I_{T_0}(\theta_0)} \lesssim d\log(\lambda_0^{-1}\rho\nu).
	\end{equation*}
	\label{lem:elliptical}
\end{lemma}

The proof of Lemma \ref{lem:elliptical} is placed in the supplementary material.
It is a routine proof following existing proofs of elliptical potential lemmas using matrix-determinant rank-1 updates.

We are now ready to give the final upper bound on $\mathrm{Regret}(\{S_t\}_{t=1}^T)$ defined in Eq.~\eqref{eq:regret}.
Note that the total regret incurred by the pure exploration phase is upper bounded by $T_0$, because the revenue parameters $r_{tj}$ are normalized so that they are upper bounded by 1.
In addition, as the failure event of $\bar R_t(S)\leq R_t(S)$ for some $S$ occurs with probability $1-O(T^{-1})$,
the total regret accumulated under the failure event is $O(T^{-1})\cdot T = O(1)$.
Further invoking Eq.~(\ref{eq:ucb-regret-decomp}) and Lemma \ref{lem:elliptical}, we have
\begin{align}
\mathrm{Regret}(\{S_t\}_{t=1}^T)
&\leq T_0 + O(1) + \mathbb E\sum_{t=T_0+1}^T R_t(S_t^*)-R_t(S_t)\nonumber\\
&\lesssim O(1) + \frac{\nu^2d\log T}{\lambda_0^2} + \frac{\rho^2(d+\log T)}{\tau^2\lambda_0} + d\sqrt{T}\cdot \log(\lambda_0^{-1}\rho\nu TK)\nonumber\\
&\lesssim O(1) + {d^2\lambda_0^{-2}\rho^4\nu^2 K^2\log T}+ d\sqrt{T}\cdot \log(\lambda_0^{-1}\rho\nu TK).
\end{align}

%\xnote{In Theorem \ref{thm:mle-ucb}, it only has a lower bound on $T_0$}

\section{Lower bound}\label{sec:negative}
\label{sec:lower}
%We establish the following regret lower bound for \emph{any} policy $\pi$:

To complement our regret analysis in Sec.~\ref{subsec:regret},
in this section we prove a \emph{lower} bound for worst-case regret.
Our lower bound is information theoretical, and therefore applies to \emph{any} policy for dynamic assortment optimization
with changing contextual features.

\begin{theorem}
	Suppose $d$ is divisible by 4.
	There exists a universal constant $C_0>0$ such that for any sufficiently large $T$ and policy $\pi$,
	there is a worst-case problem instance with $N\asymp K\cdot 2^d$ items and uniformly bounded feature and coefficient vector
	(i.e., $\|v_{ti}\|_2\leq 1$ and $\|\theta_0\|_2\leq 1$ for all $i\in[N]$, $t\in[T]$) such that the regret of $\pi$ is lower bounded by
	$C_2\cdot d\sqrt{T}/K$.
	\label{thm:lower}
\end{theorem}

Theorem \ref{thm:lower} essentially implies that the $\tilde O(d\sqrt{T})$ regret upper bound established in Theorem \ref{thm:mle-ucb} is tight
(up to logarithmic factors) in $T$ and $d$. Although there is an $O(K)$ gap between the upper and lower regret bounds,
in practical applications $K$ is usually  small and can be generally regarded as a constant. It is an interesting technical open problem to close this gap of $O(K)$.

We also remark that an $\Omega(d\sqrt{T})$ lower bound was established in \citep{Dani08stochastic} for contextual linear bandit problems.
%However, in our problem there is an additional problem parameter $K$ (assortment capacity).
However, in assortment selection, the reward function is \emph{not} coordinate-wise decomposable, making techniques in \cite{Dani08stochastic}
not directly applicable.
%Instead, we use a technique developed in \cite{Chen:18tight} for plain capacitated MNL models to prove Theorem \ref{thm:lower}.
In the following subsection, we provide a high-level proof sketch of Theorem \ref{thm:lower},
with complete proofs of technical lemmas relegated to the supplementary material.

\subsection{Proof sketch of Theorem \ref{thm:lower}}

At a higher level, the proof of Theorem \ref{thm:lower} can be divided into three steps (separated into three different sub-sections below).
In the first step, we construct an \emph{adversarial parameter} set and reduce the task of lower bounding the \emph{worst-case} regret of any policy
to lower bounding the \emph{Bayes risk} of the constructed parameter set.
In the second step, we use a ``counting argument'' similar to the one developed in \cite{Chen:18tight}
to provide an explicit lower bound on the Bayes risk of the constructed adversarial parameter set,
and finally we apply \emph{Pinsker's inequality} (see, e.g., \cite{tsybakov2009introduction}) to derive a complete lower bound.

\subsubsection{Adversarial construction and the Bayes risk}

Let $\epsilon\in(0, 1/d\sqrt{d})$ be a small positive parameter to be specified later.
For every subset $W\subseteq[d]$, define the corresponding parameter $\theta_W\in\mathbb R^d$ as $[\theta_W]_i=\epsilon$ for all $i\in W$,
and $[\theta_W]_i=0$ for all $i\notin W$.
The parameter set we consider is
\begin{equation}
\theta \in \Theta := \{\theta_W: W\in\mathcal W_{d/4}\} := \{\theta_W: W\subseteq[d], |W|=d/4\}.
\end{equation}
Note that $d/4$ is a positive integer because $d$ is divisible by 4, as assumed in Theorem \ref{thm:lower}.
Also, to simplify notation, we use $\mathcal W_k$ to denote the class of all subsets of $[d]$ whose size is $k$.

The feature vectors $\{v_{ti}\}$ are constructed to be invariant across time iterations $t$.
For each $t$ and $U\in\mathcal W_{d/4}$, $K$ identical feature vectors $v_{U}$ are constructed as (recall that $K$ is the maximum allowed assortment capacity)
\begin{equation}
[v_U]_i = 1/\sqrt{d} \;\;\;\;\text{for}\;\; i\in U; \;\;\;\;\;\;
[v_U]_i = 0 \;\;\;\;\text{for}\;\; i\notin U.
\end{equation}

It is easy to check that with the condition $\epsilon\in(0,1/\sqrt{d})$, $\|\theta_W\|_2\leq 1$ and $\|v_U\|_2\leq 1$ for all $W,U\in\mathcal W_{d/4}$.
Hence the worst-case regret of any policy $\pi$ can be lower bounded by the worst-case regret of parameters belonging to $\Theta$, which can be further lower bounded
by the ``average'' regret over a uniform prior over $\Theta$:
\begin{align}
\sup_{v,\theta}\mathbb E_{v,\theta}^\pi\sum_{t=1}^TR(S_\theta^*)-R(S_t)
&\geq \max_{\theta_W\in\Theta} \mathbb E_{v,\theta_W}^\pi\sum_{t=1}^TR(S_{\theta_W}^*)-R(S_t)\nonumber\\
&\geq \frac{1}{|\mathcal W_{d/4}|}\sum_{W\in\mathcal W_{d/4}}\mathbb E_{v,\theta_W}^\pi\sum_{t=1}^TR(S_{\theta_W}^*)-R(S_t).
\label{eq:bayes-risk}
\end{align}

Here $S_\theta^*$ is the optimal assortment of size at most $K$ that maximizes (expected) revenue under parameterization $\theta$.
By construction, it is easy to verify that $S_{\theta_W}^*$ consists of all $K$ items corresponding to feature $v_W$.
We also employ constant revenue parameters $r_{ti}\equiv 1$ for all $t\in[T]$, $i\in[N]$.

\subsubsection{The counting argument}

In this section we drive an explicit lower bound on the Bayes risk in Eq.~(\ref{eq:bayes-risk}).
For any sequences $\{S_t\}_{t=1}^T$ produced by the policy $\pi$, we first describe an alternative sequence $\{\tilde S_t\}_{t=1}^T$
that provably enjoys less regret under parameterization $\theta_W$, while simplifying our analysis.

Let $v_{U_1},\cdots,v_{U_L}$ be the distinct feature vectors contained in assortment $S_t$ (if $S_t=\emptyset$ then one may choose an arbitrary feature $v_U$)
with $U_1,\cdots,U_L\in\mathcal W_{d/4}$.
Let $U^*$ be the subset among $U_1,\cdots,U_L$ that maximizes $\langle v_{U^*},\theta_W\rangle$, where $\theta_W$ is the underlying parameter.
Let $\tilde S_t$ be the assortment consisting of all $K$ items corresponding to feature $v_U^*$.
We then have the following observation:
\begin{proposition}
	$R(S_t) \leq R(\tilde S_t)$ under $\theta_W$.
	\label{prop:reduction-surrogate-st}
\end{proposition}
\begin{proof}
	Because $r_{tj}\equiv 1$ in our construction, we have $R(S_t)=(\sum_{j\in S_t}u_{j})/(1+\sum_{j\in S_t}u_j)$ where $u_j=\exp\{v_j^\top\theta_W\}$ under $\theta_W$.
	Clearly $R(S)$ is a monotonically non-decreasing function in $u_j$.
	By replacing all $v_j\in S_t$ with $v_{U^*}\in\tilde S_t$, the $u_j$ values do not decrease and therefore the Proposition holds true.
\end{proof}

%Proposition \ref{prop:reduction-surrogate-st} is easily proved by definition of $U^*$ and the definition of the contextual MNL-logit choice model
%in Eq.~(\ref{eq:context_MNL}).
%Its complete proof is relegated to the supplementary material.
To simplify notation we also use $\tilde U_t$ to denote the unique $U^*\in\mathcal W_{d/4}$ in $\tilde S_t$.
We also use $\mathbb E_W$ and $\mathbb P_W$ to denote the law parameterized by $\theta_W$ and policy $\pi$.
The following lemma gives a lower bound on $R(\tilde S_t)-R(S_{\theta_W}^*)$ by comparing it with $W$,
which is also proved in the supplementary material.
\begin{lemma}
	Suppose $\epsilon\in(0,1/d\sqrt{d})$ and define $\delta := d/4 - |\tilde U_t\cap W|$. Then
	$$
	R(S_{\theta_W}^*)-R(\tilde S_t) \geq\frac{\delta\epsilon}{4K\sqrt{d}} .
	$$
	\label{lem:R-lb}
\end{lemma}

Define random variables $\tilde N_i := \sum_{t=1}^T\vct 1\{i\in\tilde U_t\}$.
Lemma \ref{lem:R-lb} immediately implies
\begin{equation}
\mathbb E_W\sum_{t=1}^TR(S_{\theta_W}^*)-R(\tilde S_t)
\geq \frac{\epsilon}{4K\sqrt{d}}\left(\frac{dT}{4} - \sum_{i\in W}\mathbb E_W[\tilde N_i] \right), \;\;\;\;\;\forall W\in\mathcal W_{d/4}.
\label{eq:count1}
\end{equation}

Denote $\mathcal W_{d/4}^{(i)} := \{W\in\mathcal W_{d/4}: i\in W\}$ and $\mathcal W_{d/4-1} := \{W\subseteq[d]: |W|=d/4-1\}$.
Averaging both sides of Eq.~(\ref{eq:count1}) with respect to all $W\in\mathcal W_{d/4}$ and swapping the summation order, we have
\begin{align*}
\frac{1}{|\mathcal W_{d/4}|}\sum_{W\in\mathcal W_{d/4}}&\mathbb E_{W}\sum_{t=1}^TR(S_{\theta_W}^*)-R(S_t)
\geq \frac{\epsilon}{4K\sqrt{d}} \frac{1}{|\mathcal W_{d/4}|}\sum_{W\in\mathcal W_{d/4}}\left(\frac{dT}{4} - \sum_{i\in W}\mathbb E_W[\tilde N_i]\right)\nonumber\\
&= \frac{\epsilon}{4K\sqrt{d}}\left(\frac{dT}{4} - \frac{1}{|\mathcal W_{d/4}|} \sum_{i=1}^d\sum_{W\in\mathcal W_{d/4}^{(i)}}\mathbb E_W[\tilde N_i]\right)\nonumber\\
&=  \frac{\epsilon}{4K\sqrt{d}}\left(\frac{dT}{4} - \frac{1}{|\mathcal W_{d/4}|} \sum_{W\in\mathcal W_{d/4-1}}\sum_{i\notin W}\mathbb E_{W\cup\{i\}}[\tilde N_i]\right)\nonumber\\
%&\geq \frac{\epsilon}{4K\sqrt{d}}\left(\frac{dT}{4} - \frac{1}{|\mathcal W_{d/4}|} \sum_{W\in\mathcal W_{d/4-1}}\sum_{i=1}^d\mathbb E_{W\cup\{i\}}[\tilde N_i]\right)\nonumber\\
&\geq \frac{\epsilon}{4K\sqrt{d}}\left(\frac{dT}{4} - \frac{|\mathcal W_{d/4-1}|}{|\mathcal W_{d/4}|} \max_{W\in\mathcal W_{d/4-1}}\sum_{i\notin W}\mathbb E_{W\cup\{i\}}[\tilde N_i]\right)\\
&= \frac{\epsilon}{4K\sqrt{d}}\left(\frac{dT}{4} - \frac{|\mathcal W_{d/4-1}|}{|\mathcal W_{d/4}|} \max_{W\in\mathcal W_{d/4-1}}\sum_{i\notin W}\mathbb E_{W}[\tilde N_i] + \mathbb E_{W\cup\{i\}}[\tilde N_i]-\mathbb E_W[\tilde N_i]\right).
\end{align*}

Note that for any fixed $W$, $\sum_{i\notin W}\mathbb E_W[\tilde N_i] \leq \sum_{i=1}^d\mathbb E_W[\tilde N_i]\leq dT/4$.
Also, $|\mathcal W_{d/4-1}|/|\mathcal W_{d/4}| = \binom{d}{d/4-1}/\binom{d}{d/4}=\frac{d/4}{3d/4+1} \leq 1/3$.
Subsequently,
\begin{equation}
\frac{1}{|\mathcal W_{d/4}|}\sum_{W\in\mathcal W_{d/4}}\mathbb E_{W}\sum_{t=1}^TR(S_{\theta_W}^*)-R(S_t)
\geq \frac{\epsilon}{4K\sqrt{d}}\left(\frac{dT}{6} - \max_{W\in\mathcal W_{d/4-1}}\sum_{i\notin W}|\mathbb E_{W\cup\{i\}}[\tilde N_i]-\mathbb E_W[\tilde N_i]|\right).
\label{eq:count2}
\end{equation}

\subsubsection{Pinsker's inequality}

In this section we concentrate on upper bounding $|\mathbb E_{W\cup\{i\}}[\tilde N_i]-\mathbb E_W[\tilde N_i]|$ for any $W\in\mathcal W_{d/4-1}$.
Let $P=\mathbb P_W$ and $Q=\mathbb P_{W\cup\{i\}}$ denote the laws under $\theta_W$ and $\theta_{W\cup\{i\}}$, respectively.
Then
\begin{align*}
\big|\mathbb E_P[\tilde N_i]-\mathbb E_Q[\tilde N_i]\big|
& \leq \sum_{j=0}^Tj\cdot \big|P[\tilde N_i=j] - Q[\tilde N_i=j]\big|\\
&\leq T\cdot \sum_{j=0}^T\big|P[\tilde N_i=j] - Q[\tilde N_i=j]\big|\\
&\leq T\cdot \|P-Q\|_{\mathrm{TV}} \leq T\cdot\sqrt{\frac{1}{2}\kl(P\|Q)},
\end{align*}
where $\|P-Q\|_{\mathrm{TV}}=\sup_A|P(A)-Q(A)|$ is the total variation distance between $P$, $Q$,
$\kl(P\|Q)=\int(\log\ud P/\ud Q)\ud P$ is the Kullback-Leibler (KL) divergence between $P$, $Q$,
and the inequality $\|P-Q\|_{\mathrm{TV}} \leq \sqrt{\frac{1}{2}\kl(P\|Q)}$ is the celebrated Pinsker's inequality.

For every $i\in[d]$ define random variables $N_i := \sum_{t=1}^T \frac{1}{K}\sum_{v_{U}\in S_t}\vct 1\{i\in U\}$.
The next lemma upper bound the KL divergence, which is proved in the supplementary material.
\begin{lemma}
	For any $W\in\mathcal W_{d/4-1}$ and $i\in[d]$,
	$\kl(P_W\|P_{W\cup\{i\}}) \leq C_{\kl}\cdot \mathbb E_W[N_i] \cdot\epsilon^2/{d}$ for some universal constant $C_{\kl}>0$.
	\label{lem:kl}
\end{lemma}

Combining Lemma \ref{lem:kl} and Eq.~(\ref{eq:count2}), we have
\begin{equation*}
\frac{1}{|\mathcal W_{d/4}|}\sum_{W\in\mathcal W_{d/4}}\mathbb E_{W}\sum_{t=1}^TR(S_{\theta_W}^*)-R(S_t)
\geq \frac{\epsilon}{4K\sqrt{d}}\left(\frac{dT}{6} - T\sum_{i=1}^d\sqrt{C_{\kl}\mathbb E_W[N_i]\epsilon^2/d}\right).
\end{equation*}

Further using Cauchy-Schwartz inequality, we have
$$\sum_{i=1}^d\sqrt{C_{\kl}\mathbb E_W[N_i]\epsilon^2/d}
\leq \sqrt{d}\cdot \sqrt{\sum_{i=1}^d C_{\kl}\mathbb E_W[N_i]\epsilon^2/d},$$
which is further upper bounded by $\sqrt{d}\cdot \sqrt{C_{\kl}T\epsilon^2/4}$ because $\sum_{i=1}^d\mathbb E_W[N_i]\leq dT/4$.
Subsequently,
\begin{equation}
\frac{1}{|\mathcal W_{d/4}|}\sum_{W\in\mathcal W_{d/4}}\mathbb E_{W}\sum_{t=1}^TR(S_{\theta_W}^*)-R(S_t) \geq  \frac{\epsilon}{4K\sqrt{d}}\left(\frac{dT}{6} -T \sqrt{C_{\kl}'dT\epsilon^2}\right),
\end{equation}
where $C_{\kl}'=C_{\kl}/4$.
Setting $\epsilon = \sqrt{d/144C_{\kl}'T}$ we complete the proof of Theorem \ref{thm:lower}.

\section{The combinatorial optimization subproblem}
\label{sec:comb}

The major computational bottleneck of our algorithm is its Step \ref{alg:step6},
which involves solving a \emph{combinatorial} optimization problem.
For notational simplicity,  we equivalently reformulate this problem as follows:
\begin{align}
\max_{S\subseteq[N], |S|\leq K} &\mathrm{ESTR}(S) + \min\left\{1, \omega\cdot \mathrm{CI}(S)\right\} \;\;\;\;\;\text{where}\;\;\mathrm{ESTR}(S):=\frac{\sum_{j\in S}r_{tj}\hat u_{tj}}{1+\sum_{j\in S}\hat u_{tj}}\;\;\text{and}\;\;
\label{eq:comb-opt}\\
& \mathrm{CI}(S) := \sqrt{\left\|\frac{\sum_{j\in S}\hat u_{tj} x_{tj} x_{tj}^\top}{1+\sum_{j\in S}\hat u_{tj}} -\left(\frac{\sum_{j\in S}\hat u_{tj}x_{tj}}{1+\sum_{j\in S}\hat u_{tj}}\right)\left(\frac{\sum_{j\in S}\hat u_{tj}x_{tj}}{1+\sum_{j\in S}\hat u_{tj}}\right)^\top \right\|_\op}.\nonumber
\end{align}

Here $\hat u_{tj} := \exp\{v_{tj}^\top\hat\theta_{t-1}\}$ and $x_{tj} := \hat I_{t-1}^{-1/2}(\hat\theta_{t-1})v_{tj}$, both of which can be pre-computed before solving Eq.~(\ref{eq:comb-opt}).
%\xnote{where is $\sum_t$ in this equivalent formulation?}

A brute-force way to compute Eq.~(\ref{eq:comb-opt}) is to enumerate all subsets $S\subseteq[N]$, $|S|\leq K$ and select the one with the largest objective value.
Such an approach is \emph{not} an efficient (polynomial-time) algorithm and is therefore not scalable.

In this section we provide two alternative methods for (approximately) solving the combinatorial optimization problem in Eq.~(\ref{eq:comb-opt}).
Our first algorithm is based on discretized dynamic programming and enjoys rigorous approximation guarantees.
The second algorithm is a computationally efficient greedy heuristic. Although the greedy heuristic does not have rigorous guarantees, our numerical result suggests it works reasonably well (see Sec.~\ref{sec:numerical}).

\subsection{Approximation algorithms for assortment optimization}\label{sec:approx}

In this section we introduce algorithms with polynomial running times and rigorous approximation guarantees for the optimization task described in Eq.~(\ref{eq:comb-opt}).
%With these algorithms, our Algorithm \ref{alg:mle-ucb} will be able to run in polynomial time with cumulative regret that is $\sqrt{d}$ times the one incurred if the optimization step could be exactly solved. If we allow Algorithm \ref{alg:mle-ucb} to run in time exponential in $d$ (which is usually small) but polynomial in the rest of parameters, the cumulative regret will be on the same order of the upper bound shown in Theorem \ref{thm:mle-ucb}.
%
We first formally introduce the concept of \emph{$(\alpha,\varepsilon,\delta)$-approximation} to characterize the approximation performance,
and show that such approximation guarantees imply certain upper bounds on the final regret.
\begin{definition}[$(\alpha,\varepsilon,\delta)$-approximation]
	Fix $\alpha\geq 1$, $\varepsilon\geq 0$ and $\delta\in[0,1)$.
	An algorithm is an \emph{$(\alpha,\varepsilon,\delta)$-approximation algorithm}
	if it produces $\hat S\subseteq[N]$, $|\hat S|\leq K$ such that with probability at least $1-\delta$,
	\begin{equation}
	\mathrm{ESTR}(\hat S) + \min\{1, \alpha\omega\cdot \mathrm{CI}(\hat S)\} + \varepsilon \geq \mathrm{ESTR}(S^*) + \min\{1,\omega\cdot\mathrm{CI}(S^*)\},
	\label{eq:approx}
	\end{equation}
	where $S^*$ is the assortment set maximizing the actual objective in Eq.~(\ref{eq:comb-opt}) \footnote{{We slightly abuse the notation $S^*$ here following the optimization convention that $S^*$ denotes the optimal solution. Note that $S^*$ is different from $S_t^*$ in \eqref{eq:regret}, where the latter means the assortment that maximizes the  expected revenue at time $t$.}}.
	\label{defn:approx}
\end{definition}

%As a remark, it is easy to verify that if an $(\alpha,\varepsilon,\delta)$-approximation algorithm is used instead of exact optimization
%in our proposed MLE-UCB Algorithm \ref{alg:mle-ucb}, the resulting cumulative regret can be upper bounded by
The following lemma shows how $(\alpha,\varepsilon,\delta)$-approximation algorithms imply an upper bound on the accumulated.
It is proved using standard analysis of UCB type algorithms, with the complete proof given in the supplementary material.
\begin{lemma}
	Suppose an $(\alpha,\varepsilon,\delta)$-approximation algorithm is used instead of exact optimization in the MLE-UCB policy at each time period $t$.
	Then its regret can be upper bounded by
	\begin{equation*}
	\alpha\cdot \mathrm{Regret}^* + \varepsilon T + \delta T^2 + O(1),
	\end{equation*}
	where $\mathrm{Regret}^*$ is the regret upper bound shown by Theorem \ref{thm:mle-ucb} for Algorithm \ref{alg:mle-ucb} with exact optimization in Step \ref{alg:step6}.
	%which is analyzed and upper bounded in Theorem \ref{thm:mle-ucb}.
	\label{lem:approx}
\end{lemma}

In the rest of this section we introduce our proposed approximation algorithm and the approximation guarantee.
To highlight the main idea of the approximation algorithm, we only describe how the algorithm operates in the \emph{univariate} ($d=1$) case,
while leaving the general multivariate ($d>1$) case to the appendix.

Our approximation algorithm can be roughly divided into three steps.
In the first step, we use a ``discretization'' trick to approximate the objective function
using ``rounded'' parameter values.
Such rounding motivates the second step, in which we define ``reachable states'' and
present a simple yet computationally expensive brute-force method to enumerate all reachable states,
and establish approximation guarantees for such methods. {This brute-force method is only presented for illustration purposes and will be replaced by a dynamic programing algorithm proposed in the third step.}
In particular, a \emph{dynamic programming} algorithm is developed to compute which states are ``reachable'' in polynomial time.
%Detailed approximation guarantee and time complexity analysis are also given.
%We start with the simpler univariate ($d=1$) case, and then extend to the multivariate ($d>1$) case. Finally, we put everything together by setting the approximation parameters ($\alpha$, $\epsilon$, and $\delta$) to achieve efficient algorithms with small cumulative regret.

%All technical proofs are relegated to the supplementary material.

\subsubsection{The discretization trick}

In the univariate case, $\{x_{tj}\}$ are scalars and therefore $x_{tj}x_{tj}^\top$ is simply $x_{tj}^2$.
Let $\Delta>0$ be a small positive discretization parameter to be specified later.
For all $i\in[N]$, define
%In this subsection, we will first design an approximation algorithm (Algorithm \ref{alg:approx-1d-restrict}) for the following restricted goal -- given the maximum utility item $q \in [N]$, to find an assortment $S \subseteq [N], |S| \leq K$ such that 1) $q \in S$; 2) $\forall j \in S, u_{tj} \leq u_{tq}$; 3) $\mathrm{ESTR}(S) + \min\{1, \omega \cdot \mathrm{CI}(S)\}$ is maximized. We will show that Algorithm \ref{alg:approx-1d-restrict} is a $(1, \epsilon, 0)$-approximation algorithm. Then, the general univariate case can be solved by invoking Algorithm \ref{alg:approx-1d-restrict} $N$ times, as described in Algorithm \ref{alg:approx-1d}, and it is straightforward to show that Algorithm \ref{alg:approx-1d} is also $(1, \epsilon, 0)$-approximation algorithm.
%
%
%Now we dive into the explanation and analysis of Algorithm \ref{alg:approx-1d-restrict}.
%Let $U = \max\{1, \hat{u}_{tq}\}$ and $\Delta = \epsilon_0 U / K$ be a \emph{discretization} parameter where $\epsilon_0$ will be decided later. For all $i\in[N]$, define
\begin{equation}
\mu_i := \left[\frac{\hat u_{ti}}{\Delta}\right]\Delta, \;\;
\alpha_i := \left[\frac{\hat u_{ti}x_{ti}}{\Delta}\right]\Delta,\;\;
\beta_i := \left[\frac{\hat u_{ti}x_{ti}^2}{\Delta}\right]\Delta,\;\;
\gamma_i := \left[\frac{\hat u_{ti}r_{ti}}{\Delta}\right]\Delta,
\label{eq:muabg}
\end{equation}
where $[a]$ denotes the nearest integer a real number $a$ is rounded into.
Intuitively, $\mu_i$ is the real number closest to $\hat u_{ti}$ that is an \emph{integer} multiple of the discretization parameter $\Delta$,
and similarly for $\alpha_i,\beta_i,\gamma_i$.

The motivation for the definitions of $\{\mu_i,\alpha_i,\beta_i,\gamma_i\}$ is their \emph{sufficiency} in computing the objective function $\mathrm{ESTR}(S)+\min\{1,\omega\cdot \mathrm{CI}(S)\}$.
Indeed, for any $S\subseteq[n]$, $|S|\leq K$, define $\mu=\sum_{j\in S}\mu_j$, $\alpha = \sum_{j\in S}\alpha_j$, $\beta=\sum_{j\in S}\beta_j$, $\gamma=\sum_{j\in S}\gamma_j$
and
\begin{equation*}
\hat{\mathrm{ESTR}}(S) := \frac{\gamma}{1+\mu}, \;\;\;\;\;\;
\hat{\mathrm{CI}}(S) := \max\left\{0,\sqrt{\frac{\beta}{1+\mu}-\left(\frac{\alpha}{1+\mu}\right)^2}\right\}.
\end{equation*}

Following the definition of $\mathrm{ESTR}(S)$ and $\mathrm{CI}(S)$, it is easy to see that $\hat{\mathrm{ESTR}}(S)\to \mathrm{ESTR}(S)$
and $\hat{\mathrm{CI}}(S)\to\mathrm{CI}(S)$ as $\Delta\to 0^+$.
The following lemma gives a more precise control of the error between $\hat{\mathrm{ESTR}}(S)$, $\hat{\mathrm{CI}}(S)$ and $\mathrm{ESTR}(S)$, $\mathrm{CI}(S)$
using the values of $\Delta$ and the maximum utility parameter in $S$.

\begin{lemma}
	For any $S\subseteq[N]$, $|S|\leq K$,
	suppose $U = \max_{j\in S}\{1, \hat{u}_{tj}\}$ and $\Delta = \epsilon_0 U / K$ for some $\epsilon_0>0$.
	Suppose also $|x_{tj}|\leq\nu$ for all $t,j$. Then
	\begin{equation}
	\big|\mathrm{ESTR}(S)-\hat{\mathrm{ESTR}}(S)\big| \leq 6\epsilon_0\;\;\;\;\;\text{and}\;\;\;\;\;
	\big|\mathrm{CI}(S)-\hat{\mathrm{CI}}(S)\big| \leq  \sqrt{24 \epsilon_0} (1 +\nu),
	\end{equation}
	\label{lem:approx-1d-estr-ci}
\end{lemma}

The complete proof of Lemma \ref{lem:approx-1d-estr-ci} is relegated to the supplementary material.
%
%One important aspect of Lemma \ref{lem:approx-1d-estr-ci} is that the discretization parameter $\Delta$
%(as well as the corresponding estimation error bounds) depends on the item with the largest $\hat u_{tj}$ in $S$.
%This means the discretization needs to be \emph{adaptive} with respect to $U=\max_{j\in S}\{1,\hat u_{tj}\}$.
%{\color{red} This can be done by first \emph{enumerating} the item $q\in[N]$ attaining $U$, and then searching for the remaining $K-1$ itmes
%in $[N]\backslash[q]$ to maximize the objective function.} \xnote{The last sentence needs to be better explained.}
%
\subsubsection{Reachable states and a brute-force algorithm}

To apply the estimation error bounds in Lemma \ref{lem:approx-1d-estr-ci}
one needs to first enumerate $q\in[N]$ giving rise to the item in $S$ with the largest utility parameter $\hat u_{tq}$.
After such an element $q$ is enumerated, the discretization parameter $\Delta=\epsilon_0 U/K=\epsilon_0\max\{1,\hat u_{tq}\}/K$
can be determined and discretized values $\mu_i,\alpha_i,\beta_i,\gamma_i$ can be computed for all $i\in[N]/\backslash\{q\}$.
It is also easy to verify that there are at most $O(K/\epsilon)$ possible values of $\mu_i,\gamma_i$,
$O(K\nu/\epsilon)$ possible values of $\alpha_i$ and $O(K\nu^2/\epsilon)$ possible values of $\beta_i$ (recall that $\nu$ is the upper bound of $|x_{tj}$ for all $t$ and $j$).

For any $i\in[N]\cup\{0\}$, $k\in[K]\cup\{0\}$ and $\mu,\alpha,\beta,\gamma\geq 0$ being \emph{integer} multiples of $\Delta$,
we use a tuple $\varsigma_i^k(\mu,\alpha,\beta,\gamma)$ to denote a \emph{state}. {Here the indices $i$ and $k$ mean that the assortment  $S\subseteq\{1,2,\cdots,i\}$ and $|S|=k$.}
Clearly there are at most $O(NK^5\nu^3/\epsilon^4)$ different types of states.
A state $\varsigma_i^k(\mu,\alpha,\beta,\gamma)$ can be either \emph{reachable} or \emph{non-reachable},
as defined below:
\begin{definition}
	Let $q\in[N]$ be the enumerated item with maximal utility parameter and $U=\max\{1,\hat u_{tq}\}$, $\Delta=\epsilon_0 U/K$.
	A state $\varsigma_i^k(\mu,\alpha,\beta,\gamma)$ is \emph{reachable} if there exists $S\subseteq[N]$ satisfying the following:
	\begin{enumerate}
		\item $S\subseteq\{1,2,\cdots,i\}$ and $|S|=k$;
		\item $\hat u_{tj}\leq \hat u_{tq}$ for all $j\in S$;
		\item if $i\geq q$ then $q\in S$;
		\item $\mu=\sum_{j\in S}\mu_j$, $\alpha=\sum_{j\in S}\alpha_j$, $\beta=\sum_{j\in S}\beta_j$ and $\gamma=\sum_{j\in S}\gamma_j$.
	\end{enumerate}
	On the other hand, a state $\varsigma_i^k(\mu,\alpha,\beta,\gamma)$ is \emph{non-reachable} if at least one condition above is violated.
	\label{defn:reachable}
\end{definition}
%\xnote{a bit strange for $q\in S$ if $i\geq q$                                     }

A simple way to find all reachable states is to enumerate all $S\subseteq[N]$, $|S|\leq K$ and verify the three conditions in Definition \ref{defn:reachable}.
While such a procedure is clearly computationally intractable, in the next section we will present a dynamic programming approach to compute
all reachable states in polynomial time.
After all reachable states are computed, enumerate over every $q\in[N]$ and reachable $\zeta_N^k(\cdot,\cdot,\cdot,\cdot)$ for $k\in[K]$
and find $\hat S$ that maximizes $\hat{\mathrm{ESTR}}(\hat S)+\min\{1,\omega\cdot\hat{\mathrm{CI}}(\hat S)\}$.
The following corollary establishes the approximation guarantee for $\hat S$, following Lemma \ref{lem:approx-1d-estr-ci}.

\begin{corollary}
	Let $\hat S\subseteq[N]$, $|\hat S|\leq K$ be a subset corresponding to a reachable state $\varsigma_N^k(\cdot,\cdot,\cdot,\cdot)$ for some $k\in[K]$, $q\in[N]$,
	that maximizes $\hat{\mathrm{ESTR}}(\hat S)+\min\{1,\omega\cdot\hat{\mathrm{CI}}(\hat S)\}$. Then
	$$
	{\mathrm{ESTR}}(\hat S)+\min\{1,\omega\cdot{\mathrm{CI}}(\hat S)\}
	\geq \max_{S\subseteq[N],|S|\leq K}\mathrm{ESTR}(S)+\min\{1,\omega\cdot\mathrm{CI}(S)\} - (6\epsilon_0 + \omega(1+\nu)\sqrt{24\epsilon_0}).
	$$
	\label{cor:approx-1d-estr-ci}
\end{corollary}

Corollary \ref{cor:approx-1d-estr-ci} follows easily by plugging in the upper bounds of estimation error in Lemma \ref{lem:approx-1d-estr-ci}.
By setting $\epsilon_0=\min\{\varepsilon/12, \varepsilon^2/96\omega^2(1+\nu)^2\}$,
the algorithm that produces $\hat S$ satisfies $(1,\varepsilon,0)$-approximation as defined in Definition \ref{defn:approx}.

\subsubsection{A dynamic programming method for computation of reachable states}

In this section we describe a \emph{dynamic programming} algorithm to compute reachable states in polynomial time.
The dynamic programming algorithm is exact and deterministic, therefore approximation guarantees in Corollary \ref{cor:approx-1d-estr-ci} remain valid.

The first step is again to enumerate $q\in[N]$ corresponding to the item in $S$ with the largest utility parameter $\hat u_{tq}$,
and calculating the discretization parameter $\Delta=\epsilon\max\{1,\hat u_{tq}\}/K$.
Afterwards, reachable states are computed in an iterative manner, from $i=0,1,\cdots$ until $i=N$.
The initialization is that $\varsigma_0^0(0,0,0,0)$ is reachable.
Once a state $\varsigma_i^k(\mu,\alpha,\beta,\gamma)$ is determined to be reachable,
the following two states are \emph{potentially} reachable:
$$
\varsigma_{i+1}^k(\mu,\alpha,\beta,\gamma) \;\;\;\;\;\text{and}\;\;\;\;\; \varsigma_{i+1}^{k+1}(\mu+\mu_{i+1},\alpha+\alpha_{i+1}, \beta+\beta_{i+1}, \gamma+\gamma_{i+1}).
$$

The first future state $\varsigma_{i+1}^k(\mu,\alpha,\beta,\gamma)$ corresponds to the case of $i+1\notin S$.
To determine when such a state is reachable, we review the conditions in Definition \ref{defn:reachable} and
observe that whenever $i+1\neq q$, the decision $i+1\notin S$ is legal because $q$ must belong to $S$ whenever $i\geq q$ {(note that $q$ is the item in $S$ with the largest estimated utility)}.
The second future state $\varsigma_{i+1}^{k+1}(\mu+\mu_{i+1},\alpha+\alpha_{i+1}, \beta+\beta_{i+1}, \gamma+\gamma_{i+1})$
corresponds to the case of $i+1\in S$.
Reviewing again the conditions listed in Definition \ref{defn:reachable}, such a state is reachable if $k+1\leq K$ (meaning that there is still room to include a new item in $S$)
and $\hat u_{t,i+1}\leq \hat u_{tq}$ (meaning that the new item ($i+1$) to be included has an estimated utility smaller than $\hat u_{tq}$).
Combining both cases, we arrive at the following updated rule of reachability:
\begin{enumerate}
	\item If $i+1\neq q$, then $\varsigma_{i+1}^k(\mu,\alpha,\beta,\gamma)$ is reachable;
	\item If $k<K$ and $\hat u_{t,i+1}\leq \hat u_{tq}$, then $\varsigma_{i+1}^{k+1}(\mu+\mu_{i+1},\alpha+\alpha_{i+1},\beta+\beta_{i+1},\gamma+\gamma_{i+1})$ is reachable.
\end{enumerate}

%the reachability of new states is implied according the following rule:
%\begin{enumerate}
%\item If $i+1\neq q$, then $\varsigma_{i+1}^k(\mu,\alpha,\beta,\gamma)$ is reachable;
%\item If $k<K$ and $\hat u_{t,i+1}\leq \hat u_{tq}$, then $\varsigma_{i+1}^{k+1}(\mu+\mu_{i+1},\alpha+\alpha_{i+1},\beta+\beta_{i+1},\gamma+\gamma_{i+1})$ is reachable.
%\end{enumerate}
%\xnote{still a bit unclear to me why $i+1 \neq q$?                                                                                                                                                                                                                                                                                                                                                                                                                                                                                                                                                                                                                                                                                                                                                                                                                                                                                                                                     }
Algorithms \ref{alg:approx-1d} and \ref{alg:approx-1d-restrict} give pseudo-codes for the proposed dynamic programming approach
of computing reachable states and an approximate optimizer of $\hat{\mathrm{ESTR}}(S)+\min\{1,\omega\cdot\hat{\mathrm{CI}}(S)\}$.

\begin{algorithm}[t]
	\KwInput{$\{\hat u_{ti},r_{ti},x_{ti}\}_{i=1}^N$, the designated maximum utility item $q$, and approximation parameter $\epsilon$.}
	\KwOutput{An approximate maximizer $\hat S$ of $\mathrm{ESTR}(\hat S)+\min\{1,\omega\cdot\mathrm{CI}(\hat S)\}$}
	%such that 1) $q \in S$; 2) $\forall j \in S, u_{tj} \leq u_{tq}$.}
	\caption{Approximate combinatorial optimization, the univariate ($d=1$) case, and with the designated maximum utility item.}
	\label{alg:approx-1d-restrict}

	Initialization: compute $\mu_i,\alpha_i,\beta_i,\gamma_i$ for all $i\in[N]$ as in Eq.~(\ref{eq:muabg})\;
	Declare $\varsigma_1^0(0,0,0,0)$ as reachable\;
	\For{$i=0,1,\dots ,N-1$}{
		\For {all reachable states $\varsigma_{i-1}^k(\mu,\alpha,\beta,\gamma)$}{
			\If{$i+1\neq q$}{Declare $\varsigma_{i+1}^k(\mu,\alpha,\beta,\gamma)$ as reachable\;}
			\If{$\hat u_{t,i+1} \leq \hat u_{tq}$ and $k + 1 \leq K$} {Decare $\varsigma_{i+1}^{k+1}(\mu+\mu_{i+1},\alpha+\alpha_{i+1},\beta+\beta_{i+1},\gamma+\gamma_{i+1})$  as reachable\;}
		}
	}
	
	For all reachable states $\varsigma_N^k(\mu,\alpha,\beta,\gamma)$, trace back the actual assortment $S\subseteq[N]$, $|S|\leq K$
	and select the one with the largest $\hat{\mathrm{ESTR}}(S)+\min\{1,\omega\cdot\hat{\mathrm{CI}}(S)\}$ as the output $\hat S$.
\end{algorithm}

\begin{algorithm}[!t]
	\KwInput{$\{\hat u_{ti},r_{ti},x_{ti}\}_{i=1}^N$ and additive approximation parameter $\epsilon$.}
	\KwOutput{An approximate maximizer $\hat S$ of $\mathrm{ESTR}(\hat S)+\min\{1,\omega\cdot\mathrm{CI}(\hat S)\}$.}
	\caption{Approximate combinatorial optimization, the univariate ($d=1$) case.}
	\label{alg:approx-1d}
	
	\For{$i=1,2,\dots,N$}{
		Invoke Algorithm \ref{alg:approx-1d-restrict} with parameters $q = i$ and $\epsilon$ and denote the returned assortment by $\hat{S}_i$.
	}
	
	Among $\hat{S}_1, \dots, \hat{S}_N$, select the one with the largest $\hat{\mathrm{ESTR}}(S)+\min\{1,\omega\cdot\hat{\mathrm{CI}}(S)\}$ as the output $\hat S$.
\end{algorithm}
%
%\paragraph{Approximation guarantees} For any $S\subseteq[N]$, let $\mathrm{CI}(S)$ and $\hat{\mathrm{CI}}(S)$ be defined as in Eqs.~(\ref{eq:comb-opt})
%and (\ref{eq:discretized-comb-opt}).
%%Also denote $U$ as a uniform upper bound on $\{\hat u_{ti}\}_{t,i}$.
%The following lemma upper bounds the discrepancy between $\mathrm{ESTR}(S),\mathrm{CI}(S)$ and $\hat{\mathrm{ESTR}}(S),\hat{\mathrm{CI}}(S)$:
%
%Proof of Lemma \ref{lem:approx-1d-estr-ci} is deferred to the supplementary material.
%Consequently, Algorithm \ref{alg:approx-1d} is a $(1,6\epsilon_0 + \omega \sqrt{6\nu(1 + 3\nu) \epsilon_0},0)$-approximation algorithm as defined in Definition \ref{defn:approx}.
%Alternatively, for any $\varepsilon>0$, by setting $\epsilon_0=\min\{\varepsilon/6,\varepsilon^2/(24(1+\nu)^2\omega^2)\}$
%the resulting algorithm is a $(1,\varepsilon,0)$-approximation.

Finally, we remark on the time complexity of the proposed algorithm.
Because the items $j$ we consider in the assortment satisfy $|\hat u_{ti}| \leq U$, $|r_{ti}| \leq 1$, and $|x_{ti}| \leq \nu$, and all $\mu_i,\alpha_i,\beta_i,\gamma_i$ are integral multiples of $\Delta$, we have (1) $\mu_i$ and $\gamma_i$ take at most $O(K\epsilon_0^{-1})$ possible values; (2) $\alpha_i$ takes at most $(K \nu \epsilon_0^{-1})$ possible values; and (3) $\beta_i$ takes at most $(K \nu^2 \epsilon_0^{-1})$ values. Therefore, the total number of states $\varsigma_i^k(\cdot,\cdot,\cdot,\cdot)$ for fixed $i\in[N]\cup\{0\}$, $k\in[K]$ can be upper bounded by $O(K^8 \nu^3\epsilon_0^{-4})$.
The time complexity of Algorithm \ref{alg:approx-1d} is thus upper bounded by
$O(K^9N \nu^3\epsilon_0^{-4})$.
Alternatively, to achieve $(1,\varepsilon,0)$-approximation, one may set $\epsilon_0=\min\{\varepsilon/12,\varepsilon^2/(96(1+\nu)^2\omega^2)\}$
as suggested by Corollary \ref{cor:approx-1d-estr-ci},
resulting in a time complexity of $O(K^9 N \nu^3  \max\{\varepsilon^{-4}, (1 + \nu)^8 \omega^8\varepsilon^{-8}\})$.

{This dynamic programming based approximation algorithm can be extended to multivariate feature vector with $d>1$. The details are presented in Appendix \ref{appsec:multivariate}.
}
%Now we use Algorithm \ref{alg:approx-md} to solve the combinatorial optimization problem in Step \ref{alg:step6} of Algorithm \ref{alg:mle-ucb} and examine the cumulative regret. If we let Algorithm \ref{alg:approx-md}  achieve to $(\sqrt{d},\varepsilon,\delta)$-approximation guarantee with $\varepsilon = T^{-1/2}$ and $\delta = T^{-2}$, the computational time complexity at each time slot will be $\tilde{O}(K^9 N \nu^3 (1+\nu)^8 d^4 T^4)$,\footnote{A polylogarithmic factor dependent on $T, K, \delta^{-1}, \nu, \rho$ is hidden in the $\tilde{O}(\cdot)$ notation.} and the cumulative regret will be upper bounded by $O(\sqrt{d}) \cdot \mathrm{Regret}^*$.
%If we let Algorithm \ref{alg:approx-md}  to achieve $(1/2,\varepsilon,\delta)$-approximation guarantee with $\varepsilon = T^{-1/2}$ and $\delta = T^{-2}$, the computational time complexity at each time slot will be $e^{O(d)} \cdot \tilde{O}(K^9 N \nu^3 (1+\nu)^8 d^4 T^4)$, and the cumulative regret will be upper bounded by $O(1) \cdot \mathrm{Regret}^*$

\subsection{Greedy swapping heuristics}\label{subsec:heuristic}

{While the proposed approximation has rigorous approximation guarantees and runs in polynomial time, the large time complexity still prohibits its application to moderately large scale problem instances.} In this subsection, we consider a practically efficient greedy swapping heuristic to approximately solve the combinatorial optimization problem in Eq.~(\ref{eq:comb-opt}).

At a higher level, the heuristic algorithm is a ``local search'' method similar to the Lloyd's algorithm for K-means clustering \citep{lloyd1982least},
which continuously tries to improve an assortment solution by considering local swapping/addition/deletions
until no further improvements are possible.
A pseudo-code description of our heuristic method is given in Algorithm \ref{alg:greedy}.

\begin{algorithm}[!t]
	\caption{A greedy heuristic for combinatorial assortment optimization}
	\KwInput{problem parameters $\{\hat u_{ti},r_{ti},x_{ti}\}_{i=1}^N$.}
	\KwOutput{approximate maximizer $\hat S$ of $\mathrm{ESTR}(\hat S)+\min\{1,\omega\cdot\mathrm{CI}(\hat S)\}$.}
	\label{alg:greedy}
	
	Initialization: select $S\subseteq[N]$, $|S|= K$ uniformly at random\;
	\While{$\mathrm{ESTR}(S)+\min\{1,\omega\cdot\mathrm{CI}(S)\}$ can be improved}{
		For every $i\notin S$ and $j\in S$, consider new candidate assortments $S'=S\cup\{i\}\backslash \{j\}$ (swapping), $S'=S\cup\{i\}$ if $|S|<K$ (addition)
		and $S'=S\backslash \{j\}$ if $|S|>1$ (deletion)\;
		let $S'_*$ be the considered assortments with the largest $\mathrm{ESTR}(S'_*)+\min\{1,\omega\cdot\mathrm{CI}(S'_*)\}$\;
		If $S$ can be improved update $S\gets S'_*$\;
	}
\end{algorithm}

%\xnote{maybe explain " it can be shown that the greedy swapping heuristic always produces the optimal solution. Finally, we remark that this favorable property allows the combination of greedy swapping and the MLE-UCB algorithm to achieve a diminishing average regret as $T$ grows."?}
%{\color{red}
%}

While the greedy heuristic does not have rigorous guarantees in general, we would like to mention a special case of $\omega=0$,
in which Algorithm \ref{alg:greedy} does converge to the optimal assortment $S$ maximizing $\mathrm{ESTR}(S)+\min\{1,\omega\cdot\mathrm{CI}(S)\}$
in polynomial time.
More specifically, we have the following proposition which is proved in the supplementary material. %(and also note that we made no effort to optimize the constant exponent in the polynomial):
\begin{proposition}
	If $\omega=0$, then Algorithm \ref{alg:greedy} terminates in $O(N^4)$ iterations and produces an output $S$
	that maximizes $\mathrm{ESTR}(S)$.
	\label{prop:greedy-omega-zero}
\end{proposition}

\section{Numerical studies}
\label{sec:numerical}

In this section, we present numerical results of our proposed MLE-UCB algorithm. We use the greedy swapping heuristics (Algorithm~\ref{alg:greedy}) as the subroutine to solve the combinatorial optimization problem in Eq.~(\ref{eq:comb-opt}). We will also study the quality of the solution of the greedy swapping heuristics.

\paragraph{Experiment setup.} The unknown model parameter $\theta_0 \in \R^d$ is generated as a uniformly random unit $d$-dimensional vector. The revenue parameters $\{r_{tj}\}$ for $j \in [N]$ are independently and identically generated from the uniform distribution $[0.5, 0.8]$. For the feature vectors $\{v_{tj}\}$, each of them is independently generated as a uniform random vector $v$ such that $\|v\| = 2$ and $v^\top \theta_0 < -0.6$. Here we set an upper bound of $-0.6$ for the inner product so that the utility parameters $u_{tj} = \exp\{v_{tj}^\top \theta_0\}$ are upper bounded by $\exp(-0.6)\approx 0.55$. We set such an upper bound because if the utility parameters are uniformly large, the optimal assortment is likely to pick very few items, leading to degenerated problem instances. In the implementation of our MLE-UCB algorithm, we simply set $T_0 = \lfloor \sqrt{T} \rfloor$ and $\omega = \sqrt{d \ln (T K)}$.

\paragraph{The greedy swapping heuristics.} We first numerically evaluate the solution quality of the greedy swapping heuristic algorithm by focusing on the optimization problem in Eq.~\eqref{eq:comb-opt}. We compare the obtained objective values in Eq.~\eqref{eq:comb-opt} to the proposed greedy heuristic and the optimal solution (obtained by brute-force search). Instead of generating purely random instances, we consider more realistic instances generated from a dynamic assortment planning process.  In particular, for a given $T$, we generate a dynamic assortment optimization problem with parameters $N = 10$, $K = 4$ and  $d = 5$, and run the MLE-UCB algorithm till the $T$-th time period. Now the combinatorial optimization problem in Eq.~(\ref{eq:comb-opt}) to be solved at the $T$-th time is kept as one testing instance for the greedy swapping algorithm.

For each $T \in \{50, 200, 800\}$, we generate $1000$ such test instances, and compare the solution of the greedy swapping heuristics with the optimal solution obtained by brute-force search in terms of the objective value in Eq.~\eqref{eq:comb-opt}. Table \ref{table:greedy-swap} shows the relative differences between the two solutions at several percentiles, and the mean relative differences. We can see that the approximation quality of the greedy swapping algorithm has already been desirable when $T = 50$, and becomes even better as $T$ grows.

%This is because the confidence interval shrinks when $T$ is larger; and a smaller CI is more friendly to the greedy swapping algorithm. %Indeed, when all CIs are $0$, it can be shown that the greedy swapping heuristic always produces the optimal solution. Finally, we remark that this favorable property allows the combination of greedy swapping and the MLE-UCB algorithm to achieve a diminishing average regret as $T$ grows.
%\xnote{the last sentence needs to be revised a bit}

\begin{table}[!t]
	\centering
	\begin{tabular}{cccccccc}
		\hline
		\multirow{2}{*}{$T$} & \multicolumn{5}{c}{percentile rank} & mean relative difference in \\
		\cline{2-6}
		& $94$th & $96$th & $98$th & $99$th & $99.5$th &   objective value \\ \hline
		50 &  0 & 0.0159 & 0.0293 & 0.0393 & 0.0687 & 0.00207  \\
		200 & 0 & 0.0001 & 0.0040 & 0.0080 & 0.0123 & 0.00024\\
		800 & 0 & 0 & 0 & 0.0014 & 0.0037 & 0.00004\\ \hline
	\end{tabular}
	\caption{relative differences in terms of objective value in Eq.~\eqref{eq:comb-opt} between the greedy swapping algorithm and the optimal solution.}
	\label{table:greedy-swap}
\end{table}

\paragraph{Performance of the MLE-UCB algorithm.} In Figure \ref{fig:1-a} we plot the average regret (i.e.\ $\mathrm{regret}/T$) of MLE-UCB algorithm with $N = 1000, K = 10, d = 5$ for the first $T = 10000$ time periods. For each experiment (in both  Figure \ref{fig:1-a} and other figures), we repeat the experiment for 100 times and report the average value. In Figure \ref{fig:1-b} we compare our algorithm with the UCB algorithm for multinomial logit bandit (MNL-UCB) from \cite{Agrawal16MNLBandit} without utilizing the feature information. Since the MNL-UCB algorithm assumes fixed item utilities that do not change over time, in this experiment we randomly generate one feature vector for each of the $N=1000$ items and this feature vector will be fixed for the entire time span. We can observe that our MLE-UCB algorithm performs much better than MNL-UCB, which suggests the importance of taking the advantage of the contextual information.

% Again, the experiment is repeated for 10 times and the average regret is reported.

\begin{figure}[!t]
	\centering
	\begin{subfigure}[t]{0.5\textwidth}
		\centering
		\includegraphics[height=6cm]{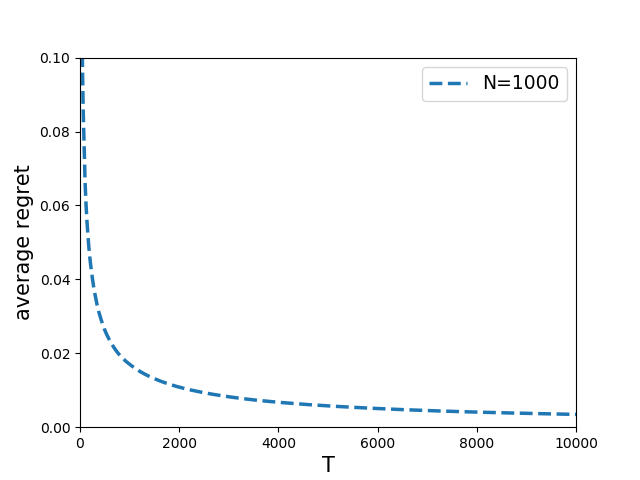}
		\caption{Average regret of MLE-UCB}\label{fig:1-a}
	\end{subfigure}%
	~
	\begin{subfigure}[t]{0.5\textwidth}
		\centering
		\includegraphics[height=6cm]{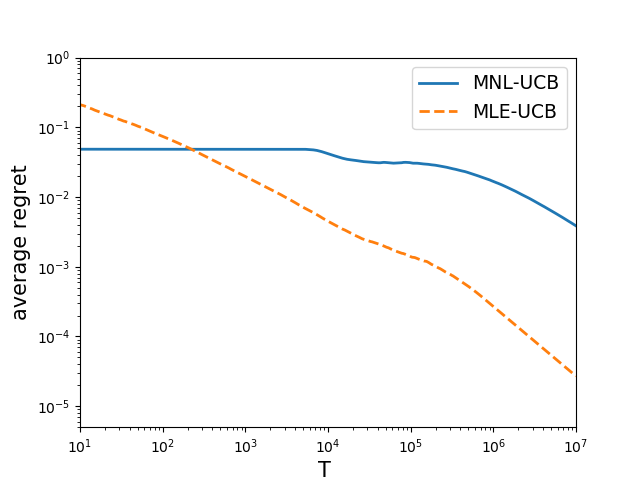}
		\caption{Comparison of MLE-UCB and MNL-UCB}\label{fig:1-b}
	\end{subfigure}
	\caption{Illustration of the performance of MLE-UCB.}
\end{figure}

\paragraph{Impact of the dimension size $d$.} We study how the dimension of the feature vector impacts the performance of our MLE-UCB algorithm. We fix $N = 1000$ and $K = 10$ and test our algorithm for dimension sizes in $5, 7, 9, 11, \dots, 25$. In Figure \ref{fig:2}, we report the average regret at times $T \in \{4000, 6000, 8000, 10000\}$. We can see that the average regret increases approximately linearly with $d$. This phenomenon matches the linear dependency on $d$ of the main term of the regret Eq.~\eqref{eq:thm-main} of the MLE-UCB.

\begin{figure}[!t]
	\begin{minipage}[t]{0.5\textwidth}
		\includegraphics[width=0.98\textwidth]{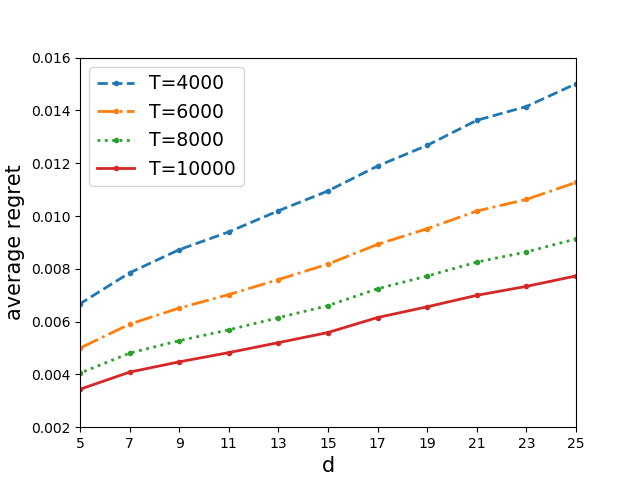}
		\caption{Average regret of MLE-UCB for various $d$'s. }
		\label{fig:2}
	\end{minipage}	
	\begin{minipage}[t]{0.5\textwidth}
		\includegraphics[width=0.98\textwidth]{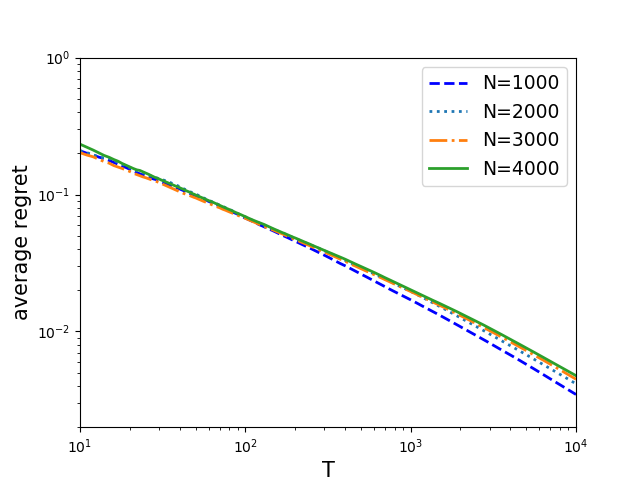}
		\caption{Average regret of MLE-UCB for various $N$'s. }
		\label{fig:3}
	\end{minipage}
\end{figure}

\paragraph{Impact of the number of items $N$.} We compare the performance of our MLE-UCB algorithm for the varying number of items $N$. We fix $K = 10$ and $d = 5$, and test MLE-UCB for $N \in \{1000, 2000, 3000, 4000\}$. In Figure \ref{fig:3}, we report the average regret for the first $T = 10000$ time periods. We observe that the regret of the algorithm is almost not affected by a bigger $N$. This confirms the fact that the regret Eq.~\eqref{eq:thm-main} of MLE-UCB is totally independent of $N$.

% Each experiment is repeated for $10$ times.

%\iffalse
%\begin{multicols}{2}
%\begin{figure}[H]
%        \centering
%        \includegraphics[height=6cm]{comp-d.png}
%        \caption{Average regret of MLE-UCB for various $d$'s}\label{fig:2}
%\end{figure}
%

%\begin{figure}[H]
%        \centering
%        \includegraphics[height=6cm]{comp-N.png}
%        \caption{Average regret of MLE-UCB for various $N$'s}\label{fig:3}
%\end{figure}
%
%\end{multicols}
%\fi

\section{Conclusion and future directions}
\label{sec:conclusion}

In this paper, we study the dynamic assortment planning problem under a contextual MNL model, which incorporates rich feature information into choice modeling. We propose an upper confidence bound (UCB) algorithm based on the local MLE that simultaneously learns the underlying coefficient and makes the decision on the assortment selection. We establish both the upper and lower bounds of the regret. Moreover, we develop an approximation algorithm and a greedy heuristic for solving the key optimization problem in our UCB algorithm.

There are a few possibilities for future work. Technically, there is still a gap of $1/K$ between our upper and lower bounds on regret. Although the cardinality constraint of an assortment $K$ is usually small in practice, it is still a technically interesting question to close this gap. Second, introducing contextual information into choice model is a natural idea for many online applications. This paper explores the standard MNL model, and it would be interesting to extend this work to contextual nested logit and other popular choice models. Finally, it is interesting to incorporate other operational considerations into the model, such as prices or inventory constraints.

\begin{appendices}
	\section{Multivariate approximation algorithm}\label{appsec:multivariate}
	\label{sec:multi}
	In this appendix we describe an approximation algorithm for the combinatorial optimization problem studied in Sec.~\ref{sec:approx}
	for the general multivariate ($d>1$) case.
	The multivariate case is dealt with by \emph{randomized} reductions to several univariate problem instances.
	
	More specifically, for any $y\in\mathbb R^d$, $\|y\|_2=1$,
	a univariate problem instance can be constructed by replacing every occurrences of $x_{ti}$ with $x_{ti}^\top y$.
	The univariate approximation Algorithm \ref{alg:approx-1d} is then invoked on $L$ independent univariate problem instances,
	each corresponding to a $y$ vector sampled uniformly at random from the $d$-dimensional unit sphere.
	The $L$ output maximizers $\hat S$ of Algorithm \ref{alg:approx-1d} are then compared against each other and the one leading to the largest
	value of $\mathrm{ESTR}(\hat S)+\min\{1,\alpha\omega\cdot\mathrm{CI}(R)\}$ is selected, where $\alpha$ is the preset multiplicative approximation parameter.
	A pseudo-code description is given in Algorithm \ref{alg:approx-md}.
	
	\begin{algorithm}[t]
		\KwInput{$\{\hat u_{ti},r_{ti},x_{ti}\}_{i=1}^N$, multiplicative approximation parameter $\alpha$, additive approximation parameter $\epsilon$, repetition $L\in\mathbb N$.}
		\KwOutput{An approximate maximizer $\hat S$ of $\mathrm{ESTR}(\hat S)+\min\{1,\omega\cdot\mathrm{CI}(\hat S)\}$.}
		
		Generalize $L$ vectors $y^{(1)},\cdots,y^{(L)}\in\mathbb R^d$ independently and uniformly from the unit sphere\;
		\For{$\ell=1,2,\cdots,L$}{
			Replace each $x_{ti}$ with $\langle x_{ti}, y^{(\ell)}\rangle$\;
			Invoke Algorithm \ref{alg:approx-1d} on the reduced univariate problem instance,
			and let $\hat S^{(\ell)}$ be the output\;
		}
		
		Output $\hat S^{(\ell)}$ that maximizes $\mathrm{ESTR}(\hat S^{(\ell)})+\min\{1,\alpha \omega\cdot\mathrm{CI}(\hat S^{(\ell)})\}$.
		
		\caption{Approximate combinatorial optimization, the multivariate ($d>1$) case}
		\label{alg:approx-md}
	\end{algorithm}
	
	\subsection{Approximation guarantees}
	The approximation performance of Algorithm \ref{alg:approx-md} can be analyzed based on the following observation:
	if $y$ is close to $y^*$, the leading eigenvector of
	$$
	\frac{\sum_{j\in S^*}\hat u_{tj}x_{tj}x_{tj}^\top}{1+\sum_{j\in S^*}\hat u_{tj}} - \left(\frac{\sum_{j\in S^*}\hat u_{tj}x_{tj}}{1+\sum_{j\in S^*}\hat u_{tj}}\right)\left(\frac{\sum_{j\in S^*}\hat u_{tj}x_{tj}}{1+\sum_{j\in S^*}\hat u_{tj}}\right)^\top,
	$$
	where $S^*$ is the exact maximizer of Eq.~(\ref{eq:comb-opt}), then the reduction to a univariate problem instance $x_{tj}\mapsto x_{tj}^\top y$ does not lose much accuracy.
	More specifically, we have the following lemma:
	\begin{lemma}
		Suppose there exists $\ell\in[L]$ such that $\langle y^{(\ell)},y^*\rangle \geq 1/\alpha$ for some $\alpha\geq 1$ in Algorithm \ref{alg:approx-md}, then
		$\mathrm{ESTR}(\hat S^{(\ell)}) + \min\{1,\alpha\omega\cdot\mathrm{CI}(\hat S^{(\ell)})\}+\varepsilon \geq \mathrm{ESTR}(S^*)+\min\{1,\omega\cdot\mathrm{CI}(S^*)\}$,
		where $\varepsilon>0$ is the approximation parameter of the univariate problem instances.
		\label{lem:approx-md}
	\end{lemma}
	
	{
		Lemma \ref{lem:approx-md} is proved in the supplementary material using elementary linear algebra.
		At a higher level, Lemma \ref{lem:approx-md} shows that when the sampled vector $y^{(\ell)}$ is close to the underlying leading eigenvector $y^*$
		(in the sense that the inner product between $y^{(\ell)}$ and $y^*$ is large),
		the produced subset $\hat S^{(\ell)}$ will have good performance in maximizing the objective function $\mathrm{ESTR}(S)+\min\{1,\omega\cdot\mathrm{CI}(S)\}$.
		
	}
	
	The following proposition additionally gives the proximity between a random $y$ and $y^*$.
	\begin{proposition}
		Assume that $d \geq 2$.  Let $y^*\in\mathbb R^d$, $\|y^*\|_2=1$ be fixed and $y$ be sampled uniformly at random from the unit $d$-dimensional sphere.
		Then
		\begin{equation*}
		\Pr[\langle y,y^*\rangle \geq 1/\sqrt{d}] = \Omega(1) \;\;\;\;\;\text{and}\;\;\;\;\; \Pr[\langle y,y^*\rangle \geq 1/2] = \exp\{-O(d)\}.
		\end{equation*}
		\label{prop:approx-init}
	\end{proposition}
	
	Proposition \ref{prop:approx-init} is again proved in the supplementary material, using isotropy of $y$ and classical concentration inequalities.
	
	%Both Lemma \ref{lem:approx-md} and Proposition \ref{prop:approx-init} are proved in the supplementary material.
	Combining Lemma \ref{lem:approx-md} and Proposition \ref{prop:approx-init} we can give some recommendations on the choice of $L$ in Algorithm \ref{alg:approx-md},
	which is the number of random $y^{(\ell)}$ vectors sampled.
	First, if $L\asymp \log(1/\delta)$ initializations are taken,
	then with probability $1-\delta$ Lemma \ref{lem:approx-md} is satisfied with $\alpha=\sqrt{d}$,
	yielding a $(\sqrt{d}, \epsilon, \delta)$-approximation.
	Additionally, if $L\asymp e^{O(d)}\log(1/\delta)$ initializations are taken,
	then with probability $1-\delta$ Lemma \ref{lem:approx-md} is satisfied with $\alpha=2$,
	yielding a $(2,\epsilon,\delta)$-approximation.
	
	\subsection{Time complexity analysis}
	To achieve a $(\sqrt{d},\varepsilon,\delta)$-approximation
	$L$ is set to $L\asymp \log(1/\delta)$ and the overall running time of Algorithm \ref{alg:approx-md}
	$O(K^9 N \nu^3  \max\{\varepsilon^{-4}, (1 + \nu)^8 \omega^8\varepsilon^{-8}\} \log \delta^{-1})$. To achieve a $(2,\varepsilon,\delta)$-approximation $L$ is set to $L\asymp e^{O(d)}\log(1/\delta)$
	and the overall running time of Algorithm \ref{alg:approx-md} is $e^{O(d)} K^9 N \nu^3  \max\{\varepsilon^{-4}, (1 + \nu)^8 \omega^8\varepsilon^{-8}\}$.

	%\subsubsection{Putting things together}
	
	Now we use Algorithm \ref{alg:approx-md} to solve the combinatorial optimization problem in Step \ref{alg:step6} of Algorithm \ref{alg:mle-ucb} and examine the cumulative regret. If we let Algorithm \ref{alg:approx-md}  achieve to $(\sqrt{d},\varepsilon,\delta)$-approximation guarantee with $\varepsilon = T^{-1/2}$ and $\delta = T^{-2}$, the computational time complexity at each time slot will be $\tilde{O}(K^9 N \nu^3 (1+\nu)^8 d^4 T^4)$,\footnote{A polylogarithmic factor dependent on $T, K, \delta^{-1}, \nu, \rho$ is hidden in the $\tilde{O}(\cdot)$ notation.} and the cumulative regret will be upper bounded by $O(\sqrt{d}) \cdot \mathrm{Regret}^*$.
	If we let Algorithm \ref{alg:approx-md}  to achieve $(1/2,\varepsilon,\delta)$-approximation guarantee with $\varepsilon = T^{-1/2}$ and $\delta = T^{-2}$, the computational time complexity at each time slot will be $e^{O(d)} \cdot \tilde{O}(K^9 N \nu^3 (1+\nu)^8 d^4 T^4)$, and the cumulative regret will be upper bounded by $O(1) \cdot \mathrm{Regret}^*$.
\end{appendices}

% Acknowledgments here
\paragraph{Acknowledgement}: {The authors would like to thank Vineet Goyal for helpful discussions, and Zikai Xiong for helping with the numerical studies.
}

\bibliographystyle{apa-good}
\bibliography{refs}
%%%%%%%%%%%%%%%%%

%%%%%%%%%% Merge with supplemental materials %%%%%%%%%%
\pagebreak
%%%%%%%%%% Merge with supplemental materials %%%%%%%%%%
%%%%%%%%%% Prefix a "S" to all equations, figures, tables and reset the counter %%%%%%%%%%
\setcounter{equation}{0}
\setcounter{figure}{0}
\setcounter{table}{0}
\setcounter{page}{1}
\setcounter{section}{0}
\makeatletter
\renewcommand{\theequation}{S\arabic{equation}}
\renewcommand{\thefigure}{S\arabic{figure}}
\renewcommand{\bibnumfmt}[1]{[S#1]}
\renewcommand{\citenumfont}[1]{S#1}

\renewcommand\thesection{\Alph{section}}
\renewcommand\thesubsection{\thesection.\arabic{subsection}}
\renewcommand\thesubsubsection{\thesubsection.\arabic{subsubsection}}

\title{Supplementary Material for: Dynamic Assortment Optimization with Changing Contextual Information}
\maketitle

%% Here starts the e-companion (EC)
%%%%%%%%%%%%%%%%%%%%%%%%%%%%%%%%%%%%%%%%%%%%%%%%%%%%%%%%%%

%\ECDisclaimer
%%%%%%%%%%%%%%%%%%%%%%%%%%%%%%%%%%%%%%%%%%%%%%%%%%%%%%%%%%

%%% Main head for the e-companion

This supplementary material provides detailed proofs for technical lemmas whose proofs are omitted in the main text.

%We give complete proofs for technical lemmas whose proofs are omitted due to space constraints in the main text and the appendix.

%\section{Proofs of technical lemmas in the proof of Theorem \ref{thm:main-upper}}

%\subsection{Proof of Proposition \ref{prop:re-sound}.}\label{sec:proof_re_sound}
\section{Proofs of technical lemmas for Theorem \ref{thm:mle-ucb} (upper bound)}
\subsection{Proof of Lemma \ref{lem:pilot}}
{
	\renewcommand{\thelemma}{\ref{lem:pilot}}
	\begin{lemma}[restated]
		With probability $1-\delta$ it holds that
		\begin{equation}
		\|\theta^*-\theta_0\|_2 \leq \frac{2}{\kappa}\sqrt{\frac{d+\log(1/\delta)}{\lambda_{\min}(V)}}
		\;\;\;\;\text{where}\;\;\kappa = \frac{1}{2e(1+\rho)}\;\;\text{and}\;\; V = \sum_{t=1}^{T_0}v_{t,i_t}v_{t,i_t}^\top.
		\end{equation}
	\end{lemma}
}
\begin{proof}
	Because the noise in a logistic regression model is clearly centered and sub-Gaussian with parameter at most $1/4$,
	it only remains to check \citep[Assumption 1]{li2017provably}, that $\inf_{\|x\|_2\leq 1, \|\theta-\theta_0\|_2\leq 1}\sigma'(x^\top\theta) \geq \kappa=2e(1+\rho)$
	where $\sigma(x)=1/(1+e^{-x})$ is the sigmoid link function.
	Because $\sigma'(x)=\sigma(x)(1-\sigma(x))$, we have $\sigma'(x^\top\theta) = \wp_\theta(1-\wp_\theta) \geq 0.5\wp_\theta$
	where $\wp_\theta = \min\{p_{\theta}(1), 1-p_{\theta}(1)\}$ and $p_\theta(1)=\sigma(x^\top\theta)=1/(1+\exp\{-x^\top\theta\})$.
	By (A2), we know that $\wp_{\theta_0}\geq 1/(1+\rho)$. Subsequently,
	for any $\|x\|_2\leq 1$ and $\|\theta-\theta_0\|_2\leq 1$, we have
	\begin{equation*}
	\wp_\theta = \frac{1}{1+\exp\{-x^\top\theta\}} = \frac{1}{1+\exp\{-x^\top(\theta-\theta_0)\} \exp\{-x^\top\theta_0\}} \geq \frac{1}{e}\frac{1}{1+\exp\{x^\top\theta_0\}} \geq \frac{1}{e(1+\rho)}.
	\end{equation*}
	Lemma \ref{lem:pilot} is then an immediate consequence of \citep[Eq.~(18)]{li2017provably}.
\end{proof}

\subsection{Proof of Corollary \ref{cor:tau}}
{\renewcommand{\thecorollary}{\ref{cor:tau}}
	\begin{corollary}[restated]
		There exists a universal constant $C_0>0$ such that for arbitrary $\tau\in(0,1/2]$,
		if $T_0 \geq C_0\max\{\nu^2d\log T/\lambda_0^2, \rho^2 (d+\log T)/(\tau^2\lambda_0)\}$ then with probability $1-O(T^{-1})$, $\|\theta^*-\theta_0\|_2\leq \tau$.
		%\label{cor:tau}
	\end{corollary}
}
\begin{proof}
	Denote $\Lambda := \mathbb E_\mu xx^\top$ and $\hat\Lambda := V/T_0 = \frac{1}{T_0}\sum_{t=1}^{T_0} x_{t,i_t}x_{t,i_t}^\top$.
	Clearly $\mathbb E\hat\Lambda = \Lambda$.
	In addition, because $\|v_{tj}\|_2\leq\nu$ almost surely, $v_{tj}$ are sub-Gaussian random variables with parameter $\nu^2$.
	By standard concentration inequalities (see, e.g., \citep[Proposition 2.1]{vershynin2012close}),
	we have with probability $1-O(T^{-2})$ that $\|\hat\Lambda-\Lambda\|_\op \lesssim \nu\sqrt\frac{d\log T}{T_0}$.
	Hence, if $T_0\geq C_0\nu^2 d\log T/\lambda_0^2$ for some sufficiently large universal constant $C_0$,
	we have $\|\hat\Lambda-\Lambda\|_\op \leq 0.5\lambda_0 = \lambda_{\min}(\Lambda)$
	and therefore $\lambda_{\min}(V) = T_0\lambda_{\min}(\hat\Lambda) \geq 0.5T_0\lambda_0$.
	The corollary then immediately follows Lemma \ref{lem:pilot}.
\end{proof}

\subsection{Proof of Lemma \ref{lem:mle}}

{\renewcommand{\thelemma}{\ref{lem:mle}}
	\begin{lemma}[restated]
		Suppose $\tau\leq 1/\sqrt{8\rho \nu^2K^2}$. Then there exists a universal constant $C>0$ such that with probability $1-O(T^{-1})$ the following holds uniformly over all $t=T_0,\cdots,T-1$:
		\begin{equation}
		(\hat\theta_t-\theta_0)^\top I_t(\theta_0)(\hat\theta_t-\theta_0) \leq C\cdot d\log(\rho\nu TK).
		\end{equation}
		%`Here we define
		%is defined in the description of Algorithm \ref{alg:mle-ucb}.
		%\label{lem:mle}
	\end{lemma}
}
\begin{proof}
	For any $\theta\in\mathbb R^d$ define
	$$
	f_{t'}(\theta) := \mathbb E_{\theta_0,t'}\left[\log\frac{p_{\theta,t'}(j|S_{t'})}{p_{\theta_0,t'}(j|S_{t'})}\right] = \sum_{j\in S_{t'}\cup\{0\}} p_{\theta_0,t'}(j|S_{t'})\log\frac{p_{\theta,t'}(j|S_{t'})}{p_{\theta_0,t'}(j|S_{t'})}.
	$$
	
	By simple algebra calculations, the first and second order derivatives of $f_{t'}$ with respect to $\theta$ can be computed as
	\begin{align}
	\nabla_\theta f_{t'}(\theta) &= \mathbb E_{\theta_0,t'}[v_{t'j}] - \mathbb E_{\theta,t'}[v_{t'j}];\label{eq:grad}\\
	\nabla_\theta^2 f_{t'}(\theta) &= -\mathbb E_{\theta_0,t'}[v_{t'j}v_{t'j}^\top] + \{\mathbb E_{\theta_0,t'}v_{t'j}\}\{\mathbb E_{\theta,t'}v_{t'j}\}^\top\nonumber\\
	&\;\;\;\; + \{\mathbb E_{\theta,t'}v_{t'j}\}\{\mathbb E_{\theta_0,t'}v_{t'j}\}^\top - \{\mathbb E_{\theta,t'}v_{t'j}\}\{\mathbb E_{\theta,t'}v_{t'j}\}^\top.\label{eq:hessian}
	\end{align}
	In the rest of the section we drop the subscript in $\nabla_\theta$, $\nabla^2_\theta$, and the $\nabla$, $\nabla^2$ notations should always be understood as with respect to $\theta$.
	
	Define $F_t(\theta) := \sum_{t'=1}^t f_{t'}(\theta)$.
	It is easy to verify that $-F_t(\theta)$ is the Kullback-Leibler divergence between the conditional distribution of $(i_1,\cdots,i_t)$ parameterized by $\theta$ and $\theta_0$, respectively.
	Therefore, $F_t(\theta)$ is always non-positive.
	Note also that $F_t(\theta_0) = 0$, $\nabla F_t(\theta_0)=0$, $\nabla^2 f_{t'}(\theta) = -M_{t'}(\theta)$ and $\nabla^2 F_t(\theta) \equiv -I_t(\theta)$.
	By Taylor expansion with Lagrangian remainder, there exists $\bar\theta_t=\alpha\theta_0+(1-\alpha)\hat\theta_t$ for some $\alpha\in(0,1)$ such that
	\begin{equation}
	F_t(\hat\theta_t) = -\frac{1}{2}(\hat\theta_t-\theta_0)^\top I_t(\bar\theta_t) (\hat\theta_t-\theta_0).
	\label{eq:taylor_F}
	\end{equation}
	
	Our next lemma shows that, if $\bar\theta_t$ is close to $\theta_0$ (guaranteed by the constraint that $\|\hat\theta_t-\theta^*\|_2\leq\tau$), then $I_t(\bar\theta_t)$ can be
	\emph{spectrally} lower bounded by $I_t(\theta_0)$.
	It is proved in the supplementary material.
	\begin{lemma}
		Suppose $\tau\leq 1/\sqrt{8\rho \nu^2K^2}$. Then $I_t(\bar\theta_t)\succeq \frac{1}{2}I_t(\theta_0)$ for all $t$.
		\label{lem:Itheta-Itheta0}
	\end{lemma}

	As a corollary of Lemma \ref{lem:Itheta-Itheta0}, we have
	\begin{equation}
	F_t(\hat\theta_t) \leq -\frac{1}{4}(\hat\theta_t-\theta_0)^\top I_t(\theta_0)(\hat\theta_t-\theta_0).
	\label{eq:taylor-F-population}
	\end{equation}
	
	On the other hand, consider the ``empirical'' version $\hat F_t(\theta) := \sum_{t'=1}^t\hat f_{t'}(\theta)$, where
	\begin{equation}
	\hat f_{t'}(\theta) := \log\frac{p_{\theta,t'}(i_{t'}|S_{t'})}{p_{\theta_0,t'}(i_{t'}|S_{t'})}.
	\end{equation}
	
	It is easy to verify that $\hat F_t(\theta_0)=0$ remains true;
	in addition, for any \emph{fixed} $\theta\in\mathbb R^d$, $\{\hat F_t(\theta)\}_t$ forms a \emph{martingale}
	\footnote{$\{X_k\}_k$ forms a martingale if $\mathbb E[X_{k+1}|X_1,\cdots,X_k]=X_k$ for all $k$.}
	and satisfies $\mathbb E\hat F_t(\theta)=F_t(\theta)$ for all $t$.
	This leads to our following lemma, which upper bounds the \emph{uniform} convergence of $\hat F_t(\theta)$ towards $F_t(\theta)$
	for all $\|\theta-\theta_0\|\leq 2\tau$.
	\begin{lemma}
		Suppose $\tau \leq 1/\sqrt{8\rho^2 \nu^2K^2}$. Then there exists a universal constant $C>0$ such that with probability $1-O(T^{-1})$ the following holds uniformly for all $t\in\{T_0+1,\cdots,T\}$ and $\|\theta-\theta_0\|_2\leq 2\tau$:
		\begin{equation}
		\big|\hat F_t(\theta)-F_t(\theta)\big| \leq C\left[d\log(\rho\nu TK) + \sqrt{|F_t(\theta)| d\log(\rho\nu TK)}\right].
		\label{eq:hatF-uniform}
		\end{equation}
		\label{lem:hatF-uniform}
	\end{lemma}
	
	Lemma \ref{lem:hatF-uniform} can be proved by using a standard $\varepsilon$-net argument.
	Since the complete proof is quite involved, we defer it to the supplementary material.
	
	%We next prove Lemma \ref{lem:hatF-uniform} by using a standard $\varepsilon$-net argument.
	
	We are now ready to prove Lemma \ref{lem:mle}.
	%Combining Eqs~(\ref{eq:taylor-F-population},\ref{eq:hatF-uniform}) and the facts that $\hat F_t(\hat\theta_t)\geq 0$, $F_t(\hat\theta_t)\leq 0$, we have
	%\begin{equation}
	%(\hat\theta_t-\theta_0)^\top I_t(\theta_0)(\hat\theta_t-\theta_0) \lesssim
	%\end{equation}
	By Eq.~(\ref{eq:hatF-uniform}) and the fact that $\hat F_t(\hat\theta_t)\leq 0\leq F_t(\hat\theta_t)$, we have
	\begin{equation}
	|F_t(\hat\theta_t)| \leq |\hat F_t(\hat\theta_t)- F_t(\hat\theta_t)| \lesssim d\log(\rho\nu TK) + \sqrt{|F_t(\hat\theta_t)| d\log(\rho\nu TK)}.
	\end{equation}
	Subsequently,
	\begin{equation}
	|F_t(\hat\theta_t)| \lesssim d\log(\rho\nu NT).
	\end{equation}
	In addition, because $F_t(\hat\theta_t)\leq 0$, by Eq.~(\ref{eq:taylor-F-population}) we have
	\begin{equation}
	-\frac{1}{2}(\hat\theta_t-\theta_0)^\top I_t(\theta_0) (\hat\theta_t-\theta_0) \geq F_t(\hat\theta_t) \geq d\log(\rho\nu TK).
	\end{equation}
	Lemma \ref{lem:mle} is thus proved.
\end{proof}

\subsubsection*{Proof of Lemma \ref{lem:Itheta-Itheta0}}
{
	\renewcommand{\thelemma}{\ref{lem:Itheta-Itheta0}}
	\begin{lemma}[restated]
		Suppose $\tau\leq 1/\sqrt{8\rho \nu^2K^2}$. Then $I_t(\bar\theta_t)\succeq \frac{1}{2}I_t(\theta_0)$ for all $t$.
		%\label{lem:Itheta-Itheta0}
	\end{lemma}
}
\begin{proof}
	Because $\hat\theta_t$ is a feasible solution of the local MLE, we know $\|\hat\theta_t-\theta^*\|_2\leq\tau$.
	Also by Corollary \ref{cor:tau} we know that $\|\theta^*-\theta_0\|_2\leq\tau$ with high probability.
	By triangle inequality and the definition of $\bar\theta_t$ we have that $\|\bar\theta_t-\theta_0\|_2 \leq 2\tau$.
	
	To prove $I_t(\bar\theta_t)\succeq \frac{1}{2}I_t(\theta_0)$ we only need to show that $M_{t'}(\bar\theta_t)-M_{t'}(\theta_0)\preceq \frac{1}{2}M_{t'}(\theta_0)$
	for all $1\leq t'\leq t$.
	This reduces to proving
	\begin{equation}
	\{\mathbb E_{\bar\theta_t,t'}v_{t'j}-\mathbb E_{\theta_0,t'}v_{t'j}\}\{\mathbb E_{\bar\theta_t,t'}v_{t'j}-\mathbb E_{\theta_0,t'}v_{t'j}\}^\top \preceq
	\frac{1}{2} \mathbb E_{\theta_0,t'}\left[(v_{t'j}-\mathbb E_{\theta_0,t'}v_{t'j})(v_{t'j}-\mathbb E_{\theta_0,t'}v_{t'j})^\top\right].
	\label{eq:preceq-1}
	\end{equation}
	
	Fix arbitrary $S_{t'}\subseteq[N]$, $|S_{t'}|=J\leq K$ and for convenience denote $x_1,\cdots,x_J\in\mathbb R^d$ as the feature vectors of items in $S_{t'}$ (i.e., $\{v_{t'j}\}_{j\in S_{t'}}$).
	Let also $p_{\theta_0}(j)$ and $p_{\bar\theta_t}(j)$ be the probability of choosing action $j\in[J]$ corresponding to $x_j$ parameterized by $\theta_0$ or $\bar\theta_t$.
	Define $\bar x := \sum_{j=1}^Jp_{\theta_0}(j)x_j$, $w_j := x_j - \bar x$ and $\delta_j := p_{\bar\theta_t}(j)-p_{\theta_0}(j)$.
	Recall also that $x_0=0$ and $w_0=-\bar x$.
	Eq.~(\ref{eq:preceq-1}) is then equivalent to
	\begin{equation}
	\left\{\sum_{j=0}^J\delta_j w_j\right\}\left\{\sum_{j=0}^J\delta_j w_j\right\}^\top \preceq \frac{1}{2}\sum_{j=0}^J p_{\theta_0}(j)w_jw_j^\top.
	\label{eq:preceq-2}
	\end{equation}
	
	Let $L=\SPAN\{w_j\}_{j=0}^J$ and $H\in\mathbb R^{L\times d}$ be a whitening matrix such that $H(\sum_j p_{\theta_0}(j)w_jw_j^\top)H^\top = I_{L\times L}$,
	where $I_{L\times L}$ is the identity matrix of size $L$.
	Denote $\tilde w_j := Hw_j$.
	We then have $\sum_{j=0}^J p_{\theta_0}(j) \tilde w_j\tilde w_j^\top = I_{L\times L}$.
	Eq.~(\ref{eq:preceq-2}) is then equivalent to
	\begin{equation}
	\left\|\sum_{j=0}^J\delta_j \tilde w_j\right\|_2^2 \leq \frac{1}{2}.
	\end{equation}
	
	On the other hand, by (A2) we know that $p_{\theta_0}(j)\geq 1/\rho K$ for all $j$ and therefore $\|\tilde w_j\|_2 \leq \sqrt{\rho K}$ for all $j$.
	Subsequently, we have
	\begin{equation}
	\left\|\sum_{j=0}^J\delta_j \tilde w_j\right\|_2^2 \leq\left(\max_j |\delta_j|\cdot \sum_{j=0}^J\|\tilde w_j\|_2\right)^2 \leq \max_j|\delta_j|^2\cdot \rho K^2.
	\end{equation}
	
	Recall that $\delta_i = p_{\bar\theta_t}(i)-p_{\theta_0}(i)$ where $p_\theta(i) = \exp\{x_i^\top\theta\}/(1+\sum_{j\in S_{t'}}\exp\{x_i^\top\theta\})$.
	Simple algebra yields that $\nabla_\theta p_\theta(i) = p_{\theta}(i)[x_i - \mathbb E_\theta x_j]$, where $\mathbb E_\theta x_j = \sum_{j\in S_{t'}}p_\theta(j)x_j$.
	Using the mean-value theorem, there exists $\tilde\theta_t=\tilde\alpha\bar\theta_t+(1-\tilde\alpha)\theta_0$ for some $\tilde\alpha\in(0,1)$ such that
	\begin{equation}
	\delta_i = \langle\nabla_\theta p_{\tilde\theta_t}(i), \hat\theta_t-\theta_0\rangle = p_{\tilde\theta_t}(i)\langle x_i-\mathbb E_{\tilde\theta_t}x_j, \bar\theta_t-\theta_0\rangle.
	\label{eq:deltai}
	\end{equation}
	%
	%Because $x_i$ are sub-Gaussian random vectors (A1), by standard concentration inequalities of quadratic forms of sub-Gaussian random vectors (e.g., \citep[Theorem 2.1]{Hsu2012})
	%and the union bound we have
	%\begin{equation}
	%\|v_{ti}\|_2^2\leq 5\nu^2d\log(NT)
	%\label{eq:xi}
	%\end{equation}
	%%inequalities we know that $\|x_i\|_2\lesssim \nu\sqrt{\log T}$ with probability $1-O(T^{-1})$
	%with probability $1-O(T^{-1})$ uniformly for all $t\in[T]$ and $i\in[N]$.
	
	Because $\|x_{ti}\|_2\leq \nu$ almost surely for all $t\in[T]$ and $i\in[N]$, we have
	\begin{equation}
	\max_j |\delta_j|^2\cdot \rho K^2 \leq 4\cdot\max_i\|x_i\|_2^2\cdot \|\bar\theta_t-\theta_0\|_2^2\cdot \rho K^2 \leq 4\rho \nu^2 K^2\cdot \tau^2.
	\label{eq:xi-all}
	\end{equation}
	The lemma is then proved by plugging in the condition on $\tau$.
\end{proof}

\subsubsection*{Proof of Lemma \ref{lem:hatF-uniform}}
{
	\renewcommand{\thelemma}{\ref{lem:hatF-uniform}}
	\begin{lemma}[restated]
		Suppose $\tau \leq 1/\sqrt{8d\rho^2 \nu^2K^2}$. Then there exists a universal constant $C>0$ such that with probability $1-O(T^{-1})$ the following holds uniformly for all $t\in\{T_0+1,\cdots,T\}$ and $\|\theta-\theta_0\|_2\leq 2\tau$:
		\begin{equation}
		\big|\hat F_t(\theta)-F_t(\theta)\big| \leq C\left[d\log(\rho\nu TK) + \sqrt{|F_t(\theta)| d\log(\rho\nu TK)}\right].
		%\label{eq:hatF-uniform}
		\end{equation}
		%\label{lem:hatF-uniform}
	\end{lemma}
}
\begin{proof}
	We first consider a \emph{fixed} $\theta\in\mathbb R^d$, $\|\theta-\theta_0\|_2\leq 2\tau$.
	Define
	\begin{equation}
	\mathcal M := \max_{t'\leq t}|\hat f_{t'}(\theta)| \;\;\;\;\;\text{and}\;\;\;\;\; \mathcal V^2 := \sum_{t'=1}^t \mathbb E_{j\sim\theta_0,t'}\left|\log\frac{p_{\theta,t'}(j|S_{t'})}{p_{\theta_0,t'}(j|S_{t'})}\right|^2.
	\end{equation}
	Using an Azuma-Bernstein type inequality (see, for example, \citep[Theorem A]{fan2015exponential}, \citep[Theorem~(1.6)]{freedman1975tail}), we have
	\begin{equation}
	\big|\hat F_t(\theta)-F_t(\theta)\big| \lesssim \mathcal M\log(1/\delta) + \sqrt{\mathcal V^2\log(1/\delta)} \;\;\;\;\;\;\text{with probability $1-\delta$}.
	\end{equation}
	
	The following lemma upper bounds $\mathcal M$ and $\mathcal V^2$ using $F_t(\theta)$ and the fact that $\theta$ is close to $\theta_0$.
	It will be proved right after this proof.
	\begin{lemma}
		If $\tau\leq 1/\sqrt{8\rho^2 \nu^2K^2}$ then $\mathcal M\leq 1$ and $\mathcal V^2\leq 8 |F_t(\theta)|$.
		\label{lem:MV}
	\end{lemma}
	
	\begin{corollary}
		Suppose $\tau$ satisfies the condition in Lemma \ref{lem:MV}. Then for any $\|\theta-\theta_0\|_2\leq 2\tau$,
		\begin{equation}
		\big|\hat F_t(\theta)-F_t(\theta)\big| \lesssim \log(1/\delta) + \sqrt{|F_t(\theta)|\log(1/\delta)} \;\;\;\;\;\;\text{with probability $1-\delta$}.
		\end{equation}
		\label{cor:concentration-Ft}
	\end{corollary}
	
	Our next step is to construct an $\epsilon$-net over $\{\theta\in\mathbb R^d: \|\theta-\theta_0\|_2\leq 2\tau\}$ and apply union bound on the constructed $\epsilon$-net.
	This together with a deterministic perturbation argument delivers \emph{uniform} concentration of $\hat F_t(\theta)$ towards $F_t(\theta)$.
	
	For any $\epsilon>0$,
	let $\mathcal H(\epsilon)$ be a finite covering of $\{\theta\in\mathbb R^d:\|\theta-\theta_0\|_2\leq 2\tau\}$ in $\|\cdot\|_2$ up to precision $\epsilon$.
	That is, $\sup_{\|\theta-\theta_0\|_2\leq 2\tau} \min_{\theta'\in\mathcal H(\epsilon)} \|\theta-\theta'\|_2\leq\epsilon$.
	By standard covering number arguments (e.g., \citep{van2000empirical}), such a finite covering set $\mathcal H(\epsilon)$ exists
	whose size can be upper bounded by $\log|\mathcal H(\epsilon)|\lesssim d\log(\tau/\epsilon)$.
	Subsequently, by Corollary \ref{cor:concentration-Ft} and the union bound, we have with probability $1-O(T^{-1})$ that
	\begin{equation}
	\big|\hat F_t(\theta)-F_t(\theta)\big|\lesssim d\log(T/\epsilon) + \sqrt{{|F_t(\theta)|d\log(T/\epsilon)}} \;\;\;\;\;\;\forall T_0<t\leq T, \theta\in\mathcal H(\epsilon).
	\label{eq:net}
	\end{equation}
	
	On the other hand, with probability $1-O(T^{-1})$ such that Eq.~(\ref{eq:deltai}) holds, we have for arbitrary $\|\theta-\theta'\|_2\leq\epsilon$ that
	\begin{align}
	\big|\hat F_t(\theta)-\hat F_t(\theta')\big|
	& \leq t\cdot\sup_{t'\leq t,j\in S_{t'}\cup\{0\}}\bigg|\log\frac{p_{\theta,t'}(j|S_{t'})}{p_{\theta',t'}(j|S_{t'})}\bigg|\nonumber\\
	&\leq t\cdot\sup_{t'\leq t,j\in S_{t'}\cup\{0\}}\frac{|p_{\theta,t'}(j|S_{t'})-p_{\theta',t'}(j|S_{t'})|}{p_{\theta',t'}(j|S_{t'})}\label{eq:cover-intermediate-1}\\
	&\leq 2\rho TK\cdot \sup_{t'\leq t, j\in S_{t'}\cup\{0\}} \big|p_{\theta,t'}(j|S_{t'})-p_{\theta',t'}(j|S_{t'})\big|\label{eq:cover-intermediate-2}\\
	&\leq 2\rho TK\cdot \sup_{t'\leq t,j\in[N]}4\|v_{t'j}\|_2^2\cdot \|\theta-\theta'\|_2\nonumber\\
	&\lesssim \rho TK\cdot \nu^2  \cdot \epsilon.\label{eq:cover-intermediate-3}
	\end{align}
	Here Eq.~(\ref{eq:cover-intermediate-1}) holds because $\log(1+x)\leq x$; Eq.~(\ref{eq:cover-intermediate-2}) holds because $p_{\theta',t'}(j|S_{t'}) \geq p_{\theta_0,t'}(j|s_{t'}) - |p_{\theta',t'}(j|S_{t'})-p_{\theta_0,t'}(j|S_{t'})|
	\geq 1/2\rho K$ thanks to (A2) and Eq.~(\ref{eq:xi-all}).
	
	Combining Eqs.~(\ref{eq:net},\ref{eq:cover-intermediate-3}) and setting $\epsilon\asymp 1/(\rho\nu^2  TK)$ we have with probability $1-O(T^{-1})$ that
	\begin{equation}
	\big|\hat F_t(\theta)-F_t(\theta)\big|\lesssim d\log(\rho\nu TK) + \sqrt{|F_t(\theta)|d\log(\rho\nu TK)}\;\;\;\;\forall T_0<t\leq T, \|\theta-\theta_0\|_2\leq2\tau,
	\end{equation}
	which is to be demonstrated in Lemma \ref{lem:hatF-uniform}.
\end{proof}

\subsubsection*{Proof of Lemma \ref{lem:MV}}
{\renewcommand{\thelemma}{\ref{lem:MV}}
	\begin{lemma}[restated]
		If $\tau\leq 1/\sqrt{8\rho^2 \nu^2K^2}$ then $\mathcal M\leq 1$ and $\mathcal V^2\leq 8 |F_t(\theta)|$.
		%\label{lem:MV}
	\end{lemma}
}
\begin{proof} %[Proof of Lemma \ref{lem:MV}]
	We first derive an upper bound for $M$.
	By (A2), we know that $p_{\theta_0,t'}(j|S_{t'})\geq 1/\rho K$ for all $j$.
	Also, Eqs.~(\ref{eq:deltai},\ref{eq:xi-all}) shows that $|p_{\theta,t'}(j|S_{t'})-p_{\theta_0,t'}(j|S_{t'})| \leq 4\nu^2\cdot \tau^2$.
	If $\tau^2\leq 1/\sqrt{8\rho \nu^2K}$ we have $|p_{\theta,t'}(j|S_{t'})-p_{\theta_0,t'}(j|S_{t'})| \leq 0.5p_{\theta_0,t'}(j|S_{t'})$
	and therefore $|\hat f_{t'}(\theta)|\leq \log^22\leq 1$.
	
	We next give upper bounds on $\mathcal V^2$.
	Fix arbitrary $t'$, and for notational simplicity let $p_j=p_{\theta_0,t'}(j|S_{t'})$ and $q_j=p_{\theta,t'}(j|S_{t'})$.
	Because $\log(1+x)\leq x$ for all $x\in(-1,\infty)$, we have
	\begin{equation}
	\mathbb E_{j\sim\theta_0,t'}\left|\log\frac{p_{\theta,t'}(j|S_{t'})}{p_{\theta_0,t'}(j|S_{t'})}\right|^2 = \sum_{j\in S_{t'}\cup\{0\}} p_j\log^2\left(1+\frac{q_j-p_j}{p_j}\right) \leq
	\sum_{j\in S_{t'}\cup\{0\}} \frac{(q_j-p_j)^2}{p_j}.
	\end{equation}
	On the other hand, by Taylor expansion we know that for any $x\in(-1,\infty)$, there exists $\bar x\in(0,x)$ such that $\log(1+x)=x-x^2/2(1+\bar x)^2$.
	Subsequently,
	\begin{align}
	-f_{t'}(\theta)
	&= -\mathbb E_{j\sim\theta_0,t'}\left[\log\frac{p_{\theta,t'}(j|S_{t'})}{p_{\theta_0,t'}(j|S_{t'})}\right] = -\sum_{j\in S_{t'}\cup \{0\}}p_j\log\left(1+\frac{q_j-p_j}{p_j}\right)\\
	&= -\sum_{j\in S_{t'}\cup\{0\}} p_j\left(\frac{q_j-p_j}{p_j} -\frac{1}{2(1+\bar\delta_j)^2} \frac{|q_j-p_j|^2}{p_j^2} \right)\\
	&\geq \frac{1}{2(1+\max_j{|p_j-q_j|/p_j})^2}\cdot \sum_{j\in S_{t'}\cup\{0\}}\frac{(q_j-p_j)^2}{p_j}.
	\end{align}
	Here $\bar\delta_j \in(0, (q_j-p_j)/p_j)$ and the last inequality holds because $\sum_j p_j=\sum_jq_j=1$.
	
	By Eqs.~(\ref{eq:deltai}) and (\ref{eq:xi-all}), we have that $|q_j-p_j|^2\leq 4 \nu^2\cdot\tau^2$.
	In addition, (A2) implies that $p_j\geq 1/\rho K$ for all $j$.
	Therefore, if $\tau\leq1/\sqrt{4\rho^2\nu^2K^2}$ we have $|p_j-q_j|/p_j\leq 1$ for all $j$ and hence
	\begin{equation}
	\mathbb E_{j\sim\theta_0,t'}\left|\log\frac{p_{\theta,t'}(j|S_{t'})}{p_{\theta_0,t'}(j|S_{t'})}\right|^2\leq
	\sum_{j\in S_{t'}\cup\{0\}} \frac{(q_j-p_j)^2}{p_j} \leq 8|f_{t'}(\theta)|.
	\end{equation}
	Summing over all $t'=1,\cdots,t$ and noting that $f_{t'}(\theta)$ is always non-positive, we complete the proof of Lemma \ref{lem:MV}.
\end{proof}

\subsection{Proof of Lemma \ref{lem:ucb}}
{\renewcommand{\thelemma}{\ref{lem:ucb}}
	\begin{lemma}[restated]
		Suppose $\tau$ satisfies the condition in Lemma \ref{lem:mle}.
		With probability $1-O(T^{-1})$ the following holds uniformly for all $t>T_0$ and $S\subseteq[N]$, $|S|\leq K$ such that
		\begin{enumerate}
			\item $\bar R_t(S)\geq R_t(S)$;
			\item $|\bar R_t(S)-R_t(S)| \lesssim \min\{1,\omega\sqrt{\|I_{t-1}^{-1/2}(\theta_0) M_t(\theta_0|S) I_{t-1}^{-1/2}(\theta_0)\|_\op}\}$.
		\end{enumerate}
		%\label{lem:ucb}
	\end{lemma}
}
\begin{proof}
	Without explicit clarification, all statements are conditioned on the success event in Lemma \ref{lem:mle},
	which occurs with probability $1-O(T^{-1})$ if $\tau$ is sufficiently large and satisfies the condition in Lemma \ref{lem:mle}.
	
	We present below a key technical lemma in the proof of Lemma \ref{lem:ucb},
	which is an upper bound on the absolute value difference between $R_t(S) := \mathbb E_{\theta_0,t}[r_{tj}|S]$ and $\hat R_t(S) := \mathbb E_{\hat\theta_{t-1},t}[r_{tj}|S]$
	using $I_{t-1}(\theta_0)$ and $M_t(\theta_0|S)$,
	where $I_{t-1}(\theta)=\sum_{t'=1}^{t-1}M_{t'}(\theta)$ and $M_{t'}(\theta) =
	\mathbb E_{\theta_0,t'}[v_{t'j}v_{t'j}^\top] - \{\mathbb E_{\theta_0,t'}v_{t'j}\}\{\mathbb E_{\theta,t'}v_{t'j}\}^\top - \{\mathbb E_{\theta,t'}v_{t'j}\} \{\mathbb E_{\theta_0,t'}v_{t'j}\}^\top  +  \{\mathbb E_{\theta,t'}v_{t'j}\}\{\mathbb E_{\theta,t'}v_{t'j}\}^\top$.
	This key lemma can be regarded as a finite sample version of the celebrated \emph{Delta's method} (e.g., \citep{van1998asymptotic}) used widely in classical statistics to estimate and/or infer a functional
	of unknown quantities.
	\begin{lemma}
		For all $t>T_0$ and $S\subseteq[N]$, $|S|\leq K$, it holds that
		$|\hat R_t(S)-R_t(S)| \lesssim\sqrt{d\log(\rho\nu TK)}\cdot \sqrt{\|I_{t-1}^{-1/2}(\theta_0) M_t(\theta_0|S) I_{t-1}^{-1/2}(\theta_0)\|_\op}$,
		where in $\lesssim$ notation we only hide numerical constants.
		\label{lem:delta-method}
	\end{lemma}
	
	Below we state our proof of Lemma \ref{lem:delta-method}, while deferring the proof of some detailed technical lemmas to the supplementary material.
	Fix $S\subseteq[N]$.
	We use $\mathfrak R_t(\theta) = \mathbb E_{\theta,t}[r_{tj}] = [\sum_{j\in S}r_{tj}\exp\{v_{tj}^\top\theta\}]/[1+\sum_{j\in S}\exp\{v_{tj}^\top\theta\}]$
	to denote the expected revenue of assortment $S$ at time $t$, evaluated using a specific model $\theta\in\mathbb R$.
	Then
	\begin{align}
	\nabla_\theta\mathfrak R_t(\theta)
	&= \frac{\sum_{j\in S}r_{tj}\exp\{v_{tj}^\top\theta\}(1+\sum_{j\in S}\exp\{v_{tj}^\top\theta\})^2 - (\sum_{j\in S}r_{tj}\exp\{v_{tj}^\top\theta\})(\sum_{j\in S}\exp\{v_{tj}^\top\theta\})}{(1+\sum_{j\in S}\exp\{v_{tj}^\top\theta\})^2}\nonumber\\
	&= \mathbb E_{\theta,t}[r_{tj}v_{tj}] - \{\mathbb E_{\theta,t}r_{tj}\}\{\mathbb E_{\theta,t}v_{tj}\}.
	\end{align}
	
	By the mean value theorem, there exists $\tilde\theta_{t-1}=\theta_0+\xi(\hat\theta_{t-1}-\theta_0)$ for some $\xi\in(0,1)$ such that
	\begin{align}
	\big|\hat R_t(S)-R_t(S)\big|
	&= \big|\mathfrak R_t(\hat\theta_{t-1})-\mathfrak R_t(\theta_0)\big|
	= \big|\langle \nabla\mathfrak R_t(\tilde\theta_{t-1}), \hat\theta_{t-1}-\theta_0\rangle\big|\nonumber\\
	&= \sqrt{(\hat\theta_{t-1}-\theta_0)^\top[\nabla\mathfrak R_t(\tilde\theta_{t-1})\nabla\mathfrak R_t(\tilde\theta_{t-1})^\top)](\hat\theta_{t-1}-\theta_0)}.
	\label{eq:delta-intermediate-1}
	%&= \sqrt{(\hat\theta_{t-1}-\theta_0)^\top\mathbb E_{\tilde\theta_{t-1},t}\left[(r_{tj}v_{tj}-\{\mathbb E_{\theta_0,t}r_{tj}\}\{\mathbb E_{\theta_0,t}v_{tj}\})(r_{tj}v_{tj}-\{\mathbb E_{\theta_0,t}r_{tj}\}\{\mathbb E_{\theta_0,t}v_{tj}\})^\top\right](\hat\theta_{t-1}-\theta_0)}
	\end{align}
	
	Recall that $\nabla\mathfrak R_t(\tilde\theta_{t-1}) = \mathbb E_{\tilde\theta_{t-1},t}[r_{tj}v_{tj}] - \{\mathbb E_{\tilde\theta_{t-1},t}r_{tj}\}\{\mathbb E_{\tilde\theta_{t-1},t}v_{tj}\}= \mathbb E_{\tilde\theta_{t-1},t}[(r_{tj}-\mathbb E_{\tilde\theta_{t-1},t}r_{tj})(v_{tj}-\mathbb E_{\tilde\theta_{t-1},t}v_{tj})]$.
	Subsequently, by Jenson's inequality and the fact that $r_{tj}\in[0,1]$ almost surely,
	\begin{align}
	\nabla\mathfrak R_t(\tilde\theta_{t-1})\nabla\mathfrak R_t(\tilde\theta_{t-1})^\top
	&\preceq \mathbb E_{\tilde\theta_{t-1,t}}\left[(r_{tj}-\mathbb E_{\tilde\theta_{t-1},t}r_{tj})^2(v_{tj}-\mathbb E_{\tilde\theta_{t-1},t}v_{tj})(v_{tj}-\mathbb E_{\tilde\theta_{t-1},t}v_{tj})^\top\right]\nonumber\\
	&\preceq \mathbb E_{\tilde\theta_{t-1,t}}\left[(v_{tj}-\mathbb E_{\tilde\theta_{t-1},t}v_{tj})(v_{tj}-\mathbb E_{\tilde\theta_{t-1},t}v_{tj})^\top\right]
	= \hat M_t(\tilde\theta_{t-1}|S).
	\label{eq:delta-intermediate-2}
	\end{align}
	
	Define $\hat M_t(\theta|S) := \mathbb E_{\theta,t}[(v_{tj}-\mathbb E_{\theta,t}v_{tj})(v_{tj}-\mathbb E_{\theta,t}v_{tj})^\top]$,
	where $S\subseteq[N]$ is the assortment supplied at iteration $t$.
	Combining Eqs.~(\ref{eq:delta-intermediate-1},\ref{eq:delta-intermediate-2}) with Lemma \ref{lem:mle}, we have
	\begin{equation}
	\big|\hat R_t(S)-R_t(S)\big|
	\lesssim \sqrt{d\log(\rho\nu NT)}\cdot \sqrt{\|I_{t-1}(\theta_0)^{-1/2}\hat M_t(\tilde\theta_{t-1}|S)I_{t-1}(\theta_0)^{-1/2}\|_\op}.
	\label{eq:delta-intermediate-3}
	\end{equation}
	
	It remains to show that $\hat M_t(\tilde\theta_{t-1}|S)$ and $M_t(\theta_0|S)$ are close,
	for which we first recall the definitions of both quantities:
	\begin{align*}
	\hat M_t(\tilde\theta_{t-1}|S) &=  \mathbb E_{\tilde\theta_{t-1,t}}\left[(v_{tj}-\mathbb E_{\tilde\theta_{t-1},t}v_{tj})(v_{tj}-\mathbb E_{\tilde\theta_{t-1},t}v_{tj})^\top\right];\\
	M_t(\theta_0|S) &= \mathbb E_{\theta_0,t}[v_{tj}v_{tj}^\top] - \{\mathbb E_{\theta_0,t}v_{tj}\}\{\mathbb E_{\theta_0,t}v_{tj}\}^\top = \hat M_t(\theta_0|S).
	\end{align*}
	
	%We show that if the conditions in Lemma \ref{lem:mle} are met and the success condition in Lemma \ref{lem:mle} is conditioned upon,
	The next lemma shows that under suitable conditions
	$\hat M_t(\tilde\theta_{t-1}|S)$ is close to $\hat M_t(\theta_0|S)=M_t(\theta_0|S)$,
	implying that $\frac{1}{4}M_t(\theta_0|S)\preceq \hat M_t(\tilde\theta_{t-1}|S) \preceq 4M_t(\theta_0|S)$.
	It is proved in the supplementary material.
	\begin{lemma}
		Suppose $\tau \leq 1/\sqrt{8\rho^2\nu^2 K^2}$. Then $\frac{1}{4}M_t(\theta_0|S)\preceq \hat M_t(\tilde\theta_{t-1}|S) \preceq 4M_t(\theta_0|S)$ for all $t$, $S$ and $\theta$.
		\label{lem:close-M}
	\end{lemma}
	
	As a consequence of Lemma \ref{lem:close-M},
	the right-hand side of Eq.~(\ref{eq:delta-intermediate-3}) can be upper bounded by
	\begin{equation*}
	\sqrt{d\log(\rho\nu TK)}\cdot \sqrt{4\|I_{t-1}(\theta_0)^{-1/2} M_t(\theta_0|S)I_{t-1}(\theta_0)^{-1/2}\|_\op}.
	\end{equation*}
	
	Lemma \ref{lem:delta-method} is thus proved.
	We are now ready to prove Lemma \ref{lem:ucb}.
	%We first prove $\bar R_t(S) =\hat R_t(S)+\max\{1,\omega\sqrt{\|\hat I_{t-1}^{-1/2}(\hat\theta_{t-1})\hat M_t(\hat\theta_{t-1}|S)\hat I_{t-1}^{-1/2}\|_\op}\} \geq R_t(S)$
	%with high probability.
	By Lemma \ref{lem:delta-method}, we know that with high probability
	\begin{equation}
	\big|\hat R_t(S)-R_t(S)\big| \lesssim \sqrt{d\log(\rho\nu TK)}\cdot \sqrt{\|I_{t-1}(\theta_0)^{-1/2} M_t(\theta_0|S)I_{t-1}(\theta_0)^{-1/2}\|_\op}
	\end{equation}
	
	In addition, by Lemma \ref{lem:close-M} and the fact that $\|\hat\theta_{t-1}-\theta_0\|_2\leq\tau$ thanks to the local MLE formulation, we have
	$\frac{1}{4}M_t(\theta_0|S)\preceq \hat M_t(\hat\theta_{t-1}|S)\preceq 4M_t(\theta_0|S)$ and subsequently
	$\frac{1}{4}I_{t-1}(\theta_0)\preceq\hat I_{t-1}(\hat\theta_{t-1})\preceq 4I_{t-1}(\theta_0)$ because $I_{t-1}(\cdot)$ and $\hat I_{t-1}(\cdot)$ are summations of
	$M_{t'}(\cdot)$ and $\hat M_{t'}(\cdot)$ terms.
	Setting $\omega \gtrsim \sqrt{d\log(\rho\nu TK)}$ we proved that $\bar R_t(S)\geq R_t(S)$.
	The second property of Lemma \ref{lem:ucb} can be proved similarly, by invoking the spectral similarities between $I_{t-1}(\cdot)$, $M_{t'}(\cdot)$
	and $\hat I_{t-1}(\cdot)$, $\hat M_{t'}(\cdot)$.
\end{proof}

\subsubsection*{Proof of Lemma \ref{lem:close-M}}
{
	\renewcommand{\thelemma}{\ref{lem:close-M}}
	\begin{lemma}[restated]
		Suppose $\tau \leq 1/\sqrt{8\rho^2\nu^2 K^2}$. Then $\frac{1}{4}M_t(\theta_0|S)\preceq \hat M_t(\tilde\theta_{t-1}|S) \preceq 4M_t(\theta_0|S)$ for all $t$, $S$ and $\theta$.
		%\label{lem:close-M}
	\end{lemma}
}
\begin{proof} %[Proof of Lemma \ref{lem:close-M}]
	Define $\bar M_t(\theta|S) := \mathbb E_{\theta_0,t}[(v_{tj}-\mathbb E_{\theta,t}v_{tj})(v_{tj}-\mathbb E_{\theta,t}v_{tj})^\top]$,
	where only the outermost expectation is replaced by taking with respect to the probability law under $\theta_0$.
	Denote also $\tilde w_j := v_{tj}-\mathbb E_{\theta,t}v_{tj}$. Then $\bar M_t(\theta|S) = \sum_j p_{\theta_0,t}(j)\tilde w_j\tilde w_j^\top$ and
	$\bar M_t(\theta|S)-\hat M_t(\theta|S) = \sum_j \delta_j\tilde w_j\tilde w_j^\top$,
	where $\delta_j = p_{\theta_0,t}(j)-p_{\theta,t}(j)$.
	By Eq.~(\ref{eq:deltai}) and the fact that $\|v_{ti}\|_2\leq\nu$, $\|\theta-\theta_0\|_2\leq\tau$, we have %with probability $1-O(T^{-1})$ that
	\begin{equation}
	\max_j|\delta_j|\leq \sqrt{4 \nu^2\cdot  \tau}.
	\label{eq:propclosem-intermediate1}
	\end{equation}
	
	On the other hand, by (A2) we know that $\min_jp_{\theta_0,t}(j)\geq 1/\rho K$ and therefore
	\begin{equation}
	\bar M_t(\theta|S) = \sum_j p_{\theta_0,t}\tilde w_j\tilde w_j^\top \succeq \frac{1}{\rho  K}\sum_j \tilde w_j\tilde w_j^\top.
	\label{eq:propclosem-intermediate2}
	\end{equation}
	
	Combining Eqs.~(\ref{eq:propclosem-intermediate1},\ref{eq:propclosem-intermediate2}) and the fact that $\bar M_t(\theta|S)-\hat M_t(\theta|S) = \sum_j \delta_j\tilde w_j\tilde w_j^\top$,
	we have $\bar M_t(\theta|S)-\hat M_t(\theta|S)\preceq \bar M_t(\theta|S)/2$ and
	$\hat M_t(\theta|S)-\bar M_t(\theta|S)\preceq  \bar M_t(\theta|S)/2$,
	provided that $\tau \leq 1/\sqrt{8\rho^2\nu^2 K^2}$.
	This also implies $\frac{1}{2}\bar M_t(\theta|S)\preceq\hat M_t(\theta|S)\preceq 2\bar M_t(\theta|S)$.
	
	We next prove that $\frac{1}{2}M_t(\theta_0|S)\preceq \bar M_t(\theta|S)\preceq 2M_t(\theta_0|S)$ which, together with $\frac{1}{2}\bar M_t(\theta|S)\preceq\hat M_t(\theta|S)\preceq 2\bar M_t(\theta|S)$ established in the previous section, implies Lemma \ref{lem:close-M}.
	Recall the definitions that
	\begin{align*}
	M_t(\theta_0|S) &= \mathbb E_{\theta_0,t}\left[(v_{tj}-\mathbb E_{\theta_0,t}v_{tj})(v_{tj}-\mathbb E_{\theta_0,t}v_{tj})^\top\right];\\
	\bar M_t(\theta|S) &=  \mathbb E_{\theta_0,t}\left[(v_{tj}-\mathbb E_{\theta,t}v_{tj})(v_{tj}-\mathbb E_{\theta,t}v_{tj})^\top\right].
	\end{align*}
	
	Adding and subtracting $\mathbb E_{\theta,t}v_{tj}, \mathbb E_{\theta_0,t}v_{tj}$ terms, we have
	\begin{align*}
	&\bar M_t(\theta|S)-M_t(\theta_0|S)\\
	&= \mathbb E_{\theta_0,t}\left[(v_{tj}-\mathbb E_{\theta_0,t}v_{tj}+\mathbb E_{\theta_0,t}v_{tj}-\mathbb E_{\theta,t}v_{tj})(v_{tj}-\mathbb E_{\theta_0,t}v_{tj}+\mathbb E_{\theta_0,t}v_{tj}-\mathbb E_{\theta,t}v_{tj})^\top\right]\\
	&\;\;\;\; - \mathbb E_{\theta_0,t}\left[(v_{tj}-\mathbb E_{\theta_0,t}v_{tj})(v_{tj}-\mathbb E_{\theta_0,t}v_{tj})^\top\right]\\
	&= \mathbb E_{\theta_0,t}\left[(\mathbb E_{\theta_0,t}v_{tj}-\mathbb E_{\theta,t}v_{tj})(v_{tj}-\mathbb E_{\theta_0,t}v_{tj})^\top\right]
	+ \mathbb E_{\theta_0,t}\left[(v_{tj}-\mathbb E_{\theta_0,t}v_{tj})(\mathbb E_{\theta_0,t}v_{tj}-\mathbb E_{\theta,t}v_{tj})^\top\right]\\
	&\;\;\;\; + (\mathbb E_{\theta_0,t}v_{tj}-\mathbb E_{\theta,t}v_{tj})(\mathbb E_{\theta_0,t}v_{tj}-\mathbb E_{\theta,t}v_{tj})^\top\\
	&= (\mathbb E_{\theta_0,t}v_{tj}-\mathbb E_{\theta,t}v_{tj})(\mathbb E_{\theta_0,t}v_{tj}-\mathbb E_{\theta,t}v_{tj})^\top.
	\end{align*}
	
	By Eq.~(\ref{eq:preceq-1}) in the proof of Lemma \ref{lem:Itheta-Itheta0}, we have that %with high probability
	$$(\mathbb E_{\theta_0,t}v_{tj}-\mathbb E_{\theta,t}v_{tj})(\mathbb E_{\theta_0,t}v_{tj}-\mathbb E_{\theta,t}v_{tj})^\top
	\lesssim \frac{1}{2}\mathbb E_{\theta_0,t}[(v_{tj}-\mathbb E_{\theta_0,t}v_{tj})(v_{tj}-\mathbb E_{\theta_0,t}v_{tj})^\top] = \frac{1}{2}M_t(\theta_0|S)$$
	provided that $\tau \leq 1/\sqrt{8\rho^2\nu^2 K^2}$,
	thus implying $\frac{1}{2}M_t(\theta_0|S)\preceq \bar M_t(\theta|S)\preceq 2M_t(\theta_0|S)$.
	%This completes the proof of Proposition \ref{prop:close-M}.
\end{proof}

\subsection{Proof of Lemma \ref{lem:elliptical}}
{\renewcommand{\thelemma}{\ref{lem:elliptical}}
	\begin{lemma}[restated]
		%With probability $1-O(T^{-1})$ it holds that
		It holds that
		\begin{equation*}
		\sum_{t=T_0+1}^T\min\{1,\|I_{t-1}^{-1/2}(\theta_0)M_t(\theta_0|S_t)I_{t-1}^{-1/2}(\theta_0)\|_\op^2\} \leq
		4\log\frac{\det I_{T}(\theta_0)}{\det I_{T_0}(\theta_0)} \lesssim d\log(\lambda_0^{-1}\rho\nu).
		\end{equation*}
		%\label{lem:elliptical}
	\end{lemma}
}
\begin{proof}
	Denote $A_t := I_{t-1}^{-1/2}(\theta_0)M_t(\theta_0|S_t)I_{t-1}^{-1/2}(\theta_0)$ as $d$-dimensional positive semi-definite matrices
	with eigenvalues sorted as $\sigma_1(A_t)\geq\cdots\geq\sigma_d(A_t)\geq 0$.
	By simple algebra,
	\begin{align}
	\sum_{t=T_0+1}^T&\min\{1,\|I_{t-1}^{-1/2}(\theta_0)M_t(\theta_0|S_t)I_{t-1}^{-1/2}(\theta_0)\|_\op^2\}
	= \sum_{t=T_0+1}^T\min\{1,\sigma_1(A_t)^2\}\nonumber\\
	&\leq \sum_{t=T_0+1}^T2\log(1+\sigma_1(A_t)^2) \leq \sum_{t=T_0+1}^T4\log(1+\sigma_1(A_t)).
	\label{eq:epl-intermediate-1}
	\end{align}
	
	On the other hand, note that $I_t(\theta_0) = I_{t-1}(\theta_0) + M_t(\theta_0|S_t)
	= I_{t-1}(\theta_0)^{1/2}[I_{d\times d} + A_t] I_{t-1}(\theta_0)^{1/2}$.
	Hence,
	\begin{equation}
	\log \det I_t(\theta_0) =\log \det I_{t-1}(\theta_0) + \sum_{j=1}^d\log(1+\sigma_j(A_t)).
	\label{eq:epl-intermediate-2}
	\end{equation}
	
	Comparing Eqs.~(\ref{eq:epl-intermediate-1}) and (\ref{eq:epl-intermediate-2}), we have
	\begin{equation}
	\sum_{t=T_0+1}^T\min\{1,\|I_{t-1}^{-1/2}(\theta_0)M_t(\theta_0|S_t)I_{t-1}^{-1/2}(\theta_0)\|_\op^2\} \leq 4\log\frac{\det I_T(\theta_0)}{\det I_{T_0}(\theta_0)},
	\end{equation}
	which proves the first inequality in Lemma \ref{lem:elliptical}.
	
	We next prove the second inequality in Lemma \ref{lem:elliptical}.
	Because assortments have size 1 throughout the pure exploration phase ($t\leq T_0$), we have
	\begin{align}
	I_{T_0}(\theta_0) = \sum_{t=1}^{T_0} p_{\theta_0,t}(j_t)(1-p_{\theta_0,t}(j_t))^2 v_{t,j_t}v_{t,j_t}^\top
	\geq\frac{1}{(1+\rho)^3}\cdot \sum_{t=1}^{T_0} v_{t,j_t}v_{t,j_t}^\top,
	\end{align}
	where the last inequality holds thanks to assumption (A2), which implies $p_{\theta_0,t}(j_t) \in [1/(1+\rho), \rho/(1+\rho)]$.
	In addition, by the proof of Corollary \ref{cor:tau}, with high probability $\lambda_{\min}(\sum_{t=1}^{T_0}v_{t,j_t}v_{t,j_t}^\top)\geq 0.5 T_0\lambda_0$,
	where $\lambda_0>0$ is a parameter specified in assumption (A1).
	Therefore,
	\begin{equation}
	\det I_{T_0}(\theta_0) \gtrsim [T_0\lambda_0/\rho^3]^d.
	\label{eq:epl-intermediate-3}
	\end{equation}
	
	On the other hand, because $\max_{t,j}\|v_{tj}\|_2 \leq \nu$
	% with high probability (see Eq.~(\ref{eq:xi}) in the appendix),
	we have $I_T(\theta_0) \lesssim T\cdot \nu^2$ and subsequently
	\begin{equation}
	\det I_T(\theta_0) \lesssim [\nu^2  T]^d.
	\label{eq:epl-intermediate-4}
	\end{equation}
	
	Combining Eqs.~(\ref{eq:epl-intermediate-3}) and (\ref{eq:epl-intermediate-4}) we proved the second inequality in Lemma \ref{lem:elliptical}.
\end{proof}

\section{Proofs of technical lemmas for Theorem \ref{thm:lower} (lower bound)}
\subsection{Proof of Lemma \ref{lem:R-lb}}
{
	\renewcommand{\thelemma}{\ref{lem:R-lb}}
	\begin{lemma}[restated]
		Suppose $\epsilon\in(0,1/d\sqrt{d})$ and define $\delta := d/4- |\tilde U_t\cap W|$. Then
		$$
		R(S_{\theta_W}^*)-R(\tilde S_t) \geq\frac{\delta\epsilon}{4K\sqrt{d}} .
		$$
		%\label{lem:R-lb}
	\end{lemma}
}
\begin{proof}
	Let $v=v_W$ and $\hat v=v_{\tilde U_t}$ be the corresponding feature vectors. Then
	\begin{align*}
	R(S_{\theta_W}^*) - R(\tilde S_t)
	&= \frac{K\exp\{v^\top\theta_W\}}{1+K\exp\{v^\top\theta_W\}} - \frac{K\exp\{\hat v^\top\theta_W\}}{1+K\exp\{\hat v^\top\theta_W\}} \\
	&= \frac{K[\exp\{v^\top\theta_W\} - \exp\{\hat v^\top\theta_W\}]}{(1+K\exp\{v^\top\theta_W\})(1+K\exp\{\hat v^\top\theta_W\})}\\
	&\geq \frac{\exp\{v^\top\theta_W\} - \exp\{\hat v^\top\theta_W\}}{2Ke}.
	\end{align*}
	
	Here the last inequality holds because $\max(\exp\{v^\top\theta_W\},\exp\{\hat v^\top\theta_W\})\leq e$.
	In addition, by Taylor expansion we know that $1+x\leq e^x\leq 1+x+x^2/2$ for all $x\in[0,1]$.
	Subsequently,
	\begin{align*}
	R(S_{\theta_W}^*) - R(\tilde S_t)
	&\geq \frac{(v-\hat v)^\top \theta_W - (\hat v^\top\theta_W)^2/2}{2Ke}
	\geq \frac{\delta\epsilon/\sqrt{d} - (\sqrt{d}\epsilon)^2/2}{2Ke}.
	\end{align*}
	
	Finally, noting that $d\epsilon^2/2\leq \delta\epsilon/2\sqrt{d}$ provided that $\epsilon\in(0,1/d\sqrt{d})$, we finish the proof of Lemma \ref{lem:R-lb}.
\end{proof}

\subsection{Proof of Lemma \ref{lem:kl}}

{
	\renewcommand{\thelemma}{\ref{lem:kl}}
	\begin{lemma}[restated]
		For any $W\in\mathcal W_{d/4-1}$ and $i\in[d]$,
		$\kl(P_W\|P_{W\cup\{i\}}) \leq C_{\kl}\cdot \mathbb E_W[N_i] \cdot\epsilon^2/{d}$ for some universal constant $C_{\kl}>0$.
		%\label{lem:kl}
	\end{lemma}
}
\begin{proof}
	Fix a time $t$ with policy's assortment choice $S_t$, and define $n_i(S_t) := \sum_{v_U\in S_t}\vct 1\{i\in U\}/K$.
	Let $\{p_j\}_{j\in S_t\cup\{0\}}$ and $\{q_j\}_{j\in S_t\cup \{0\}}$ be the probabilities of purchasing item $j$ under parameterization $\theta_W$ and $\theta_{W\cup\{i\}}$, respectively.
	Then
	\begin{equation}
	\kl(P_W(\cdot|S_t)\|P_{W\cup\{i\}}(\cdot|S_t))
	= \sum_{j\in S_t\cup\{0\}} p_j\log\frac{q_j}{p_j}
	\leq \sum_{j}p_j\frac{p_j-q_j}{q_j}
	\leq \sum_j \frac{|p_j-q_j|^2}{q_j},
	\label{eq:kl-intermediate-1}
	\end{equation}
	where the only inequality holds because $\log(1+x)\leq x$ for all $x>-1$.
	Because $q_j\geq e^{-1}/(1+Ke) \geq 1/(2Ke^2)$ for all $j\in S_t\cup\{0\}$, Eq.~(\ref{eq:kl-intermediate-1}) is reduced to
	\begin{equation}
	\kl(P_W(\cdot|S_t)\|P_{W\cup\{i\}}(\cdot|S_t)) \leq 2e^2 K\cdot \sum_{j\in S_t\cup\{0\}} |p_j-q_j|^2.
	\label{eq:kl-intermediate-2}
	\end{equation}
	
	We next upper bound $|p_j-q_j|$ separately.
	First consider $j=0$. We have
	\begin{align*}
	|p_j-q_j|
	&= \left|\frac{1}{1+\sum_{j\in S_t}\exp\{v_j^\top\theta_W\}} - \frac{1}{1+\sum_{j\in S_t}\exp\{v_j^\top\theta_{W\cup\{i\}}\}}\right|\\
	&\leq \frac{1}{(1+K/e)^2}\cdot 2\sum_{j\in S_t}\big|v_j^\top(\theta_W-\theta_{W\cup\{i\}})\big|\\
	&\leq \frac{2Kn_i(S_t)\epsilon/\sqrt{d}}{(1+K/e)^2}\leq \frac{8e^2n_i(S_t)\epsilon}{K\sqrt{d}}.
	\end{align*}
	
	Here the first inequality holds because $e^x\leq 1+2x$ for all $x\in[0,1]$.
	
	For $j>0$ corresponding to $v_j=v_U$ where $i\notin U$, we have
	\begin{align*}
	|p_j-q_j| &= \left|\frac{\exp\{v_U^\top\theta_W\}}{1+\sum_{j\in S_t}\exp\{v_j^\top\theta_W\}} - \frac{\exp\{v_U^\top\theta_{W\cup\{i\}}\}}{1+\sum_{j\in S_t}\exp\{v_j^\top\theta_{W\cup\{i\}}\}}\right|\\
	&\leq  \left|\frac{1}{1+\sum_{j\in S_t}\exp\{v_j^\top\theta_W\}} - \frac{1}{1+\sum_{j\in S_t}\exp\{v_j^\top\theta_{W\cup\{i\}}\}}\right|\\
	&\leq \frac{8e^2n_i(S_t)\epsilon}{K\sqrt{d}}.
	\end{align*}
	
	Here the first inequality holds because $\exp\{v_U^\top\theta_W\} = \exp\{v_U^\top\theta_{W\cup\{i\}}\} \leq 1$, since $i\notin U$.
	
	For $j>0$ corresponding to $v_j=v_U$ and $i\in U$, we have
	\begin{align*}
	|p_j-q_j| &= \left|\frac{\exp\{v_U^\top\theta_W\}}{1+\sum_{j\in S_t}\exp\{v_j^\top\theta_W\}} - \frac{\exp\{v_U^\top\theta_{W\cup\{i\}}\}}{1+\sum_{j\in S_t}\exp\{v_j^\top\theta_{W\cup\{i\}}\}}\right|\\
	&\leq \exp\{v_u^\top\theta_{W\cup\{i\}}\}\cdot \left|\frac{1}{1+\sum_{j\in S_t}\exp\{v_j^\top\theta_W\}} - \frac{1}{1+\sum_{j\in S_t}\exp\{v_j^\top\theta_{W\cup\{i\}}\}}\right|\\
	&\;\; + \big|\exp\{v_u^\top\theta_W\} - \exp\{v_u^\top\theta_{W\cup\{i\}}\}\big|\cdot \left|\frac{1}{1+\sum_{j\in S_t}\exp\{v_j^\top\theta_W\}}\right|\\
	&\leq  \frac{8e^2n_i(S_t)\epsilon}{K\sqrt{d}} + \frac{\epsilon}{\sqrt{d}}\cdot \frac{1}{1+K/e}.
	\leq \frac{8e^2n_i(S_t)\epsilon}{K\sqrt{d}} + \frac{2e\epsilon}{K\sqrt{d}}.
	\end{align*}
	
	%Here the last inequality holds because $n_i(S_t) \leq 1$.
	
	Combining all upper bounds on $|p_j-q_j|$ and Eq.~(\ref{eq:kl-intermediate-2}), we have
	\begin{align*}
	\kl(P_W(\cdot|S_t)\|P_{W\cup\{i\}}(\cdot|S_t))
	&\leq 2e^2K\cdot \left[\frac{128e^4n_i(S_t)^2\epsilon^2}{K^2d}(1+K)+ Kn_i(S_t)\cdot \frac{8e^4\epsilon^2}{K^2d} \right]\nonumber\\
	&\lesssim n_i(S_t)\epsilon^2/{d}.
	\end{align*}
	
	Here the last inequality holds because $n_i(S_t) \leq 1$.
	Note also that $N_i = \sum_{t=1}^Tn_i(S_t)$ by definition, and subsequently summing over all $t=1$ to $T$ we have
	$$
	\kl(P_W\|P_{W\cup\{i\}}) \lesssim \mathbb E_W[N_i]\cdot \epsilon^2/{d},
	$$
	which is to be demonstrated.
\end{proof}

\section{Proofs of approximation algorithms}

\subsection{Proof of Lemma \ref{lem:approx}}
{\renewcommand{\thelemma}{\ref{lem:approx}}
	\begin{lemma}[restated]
		Suppose an $(\alpha,\varepsilon,\delta)$-approximation algorithm is used instead of exact optimization in the MLE-UCB policy at each time period $t$.
		Then its regret can be upper bounded by
		\begin{equation*}
		\alpha\cdot \mathrm{Regret}^* + \varepsilon T + \delta T^2 + O(1),
		\end{equation*}
		where $\mathrm{Regret}^*$ is the regret upper bound shown by Theorem \ref{thm:mle-ucb} for Algorithm \ref{alg:mle-ucb} with exact optimization in Step \ref{alg:step6}.
		%which is analyzed and upper bounded in Theorem \ref{thm:mle-ucb}.
		%\label{lem:approx}
	\end{lemma}
}
\begin{proof}
	By union bound, we know the approximation guarantee in Eq.~(\ref{eq:approx}) \emph{for all $t$} with probability at least $1-\delta T$.
	In the event of failure, the accumulated regret is upper bounded by $T$ almost surely, because the regret incurred by each time period $t$ is at most $1$.
	This gives rises to the $\delta T^2$ term in Lemma \ref{lem:approx}, and in the rest of the proof we shall assume Eq.~(\ref{eq:approx}) holds for all $t$.
	
	Let $S_t^*$ be the solution to the exact optimization problem in Step \ref{alg:step6} of Algorithm \ref{alg:mle-ucb}, $S_t^\#$ be the assortment with the optimal revenue the same step,
	and $\hat S_t$ be the solution by an $(\alpha,\varepsilon,\delta)$-approximation algorithm.
	
	For each $t > T_0$, we bound the expected regret incurred at time $t$ by
	\begin{align*}
	&\mathrm{ESTR}(S^\#_t) - \mathrm{ESTR}(S_t)\\
	=& \left( \mathrm{ESTR}(S^\#_t) + \min\{1 + \omega \cdot \mathrm{CI}(S^\#_t)\} \right) -  \left( \mathrm{ESTR}(S^*_t) + \min\{1 + \omega \cdot \mathrm{CI}(S^*_t)\} \right) \\
	& \qquad + \left( \mathrm{ESTR}(S^*_t) + \min\{1 + \omega \cdot \mathrm{CI}(S^*_t)\}\right) - \left(\mathrm{ESTR}(S_t) + \min\{1 + \alpha \omega \cdot \mathrm{CI}(S_t)\} - \varepsilon\right) \\
	& \qquad + \left(\varepsilon + \min\{1 + \alpha \omega \cdot \mathrm{CI}(S_t)\} - \min\{1 + \omega \cdot \mathrm{CI}(S^\#_t)\} \right)\\
	\leq  & \varepsilon + \min\{1 + \alpha \omega \cdot \mathrm{CI}(S_t)\} - \min\{1 + \omega \cdot \mathrm{CI}(S^\#_t)\}  \leq   \varepsilon + \min\{1 + \alpha \omega \cdot \mathrm{CI}(S_t)\}  .
	\end{align*}
	
	Therefore, the total expected regret is bounded by
	\begin{align*}
	T_0 + \sum_{t = T_0 + 1}^{T}  \left(\varepsilon + \min\{1 + \alpha \omega \cdot \mathrm{CI}(S_t)\} \right)
	\leq  \varepsilon T + T_0 + \alpha \sum_{t = T_0 + 1}^{T} \min\{1 +  \omega \cdot \mathrm{CI}(S_t)\},
	\end{align*}
	which, by the same analysis in Section~\ref{sec:ellipitical}, can be bounded by $\alpha \cdot \mathrm{Regret}^* +  \varepsilon T$.
	%{\color{red}\bf [Yuan: finish the proof]}
	%By Eq.~(\ref{eq:ucb-triangle}) and Lemma \ref{lem:ucb}, the accumulated regret of $\{S_t^*\}$ is upper bounded by
	%$$
	%\textstyle O(1) + \sum_{t} \bar R_t(S_t^*)-R_t(S_t^*) \leq O(1) + C\cdot \sum_t \mathrm{ESTR}(S_t^*) + \min\{1,\omega\cdot \mathrm{CI}(S_t^*)\}
	%$$
	%for some universal constant $C>0$.
	%Similarly, the accumulated regret of $\{\hat S_t\}$ can be upper bounded by
	%$$
	%\textstyle O(1) + C\cdot \sum_t\mathrm{ESTR}(\hat S_t) + \min\{1,\omega\cdot\mathrm{CI}(\hat S_t)\}.
	%$$
	%
	%The proposition is then proved by invoking the approximation guarantee in Eq.~(\ref{eq:approx}).
\end{proof}

\subsection{Proof of Lemma \ref{lem:approx-1d-estr-ci}}

{
	\renewcommand{\thelemma}{\ref{lem:approx-1d-estr-ci}}
	\begin{lemma}[restated]
		For any $S\subseteq[N]$, $|S|\leq K$,
		suppose $U = \max_{j\in S}\{1, \hat{u}_{tj}\}$ and $\Delta = \epsilon_0 U / K$ for some $\epsilon_0>0$.
		Suppose also $|x_{tj}|\leq\nu$ for all $t,j$. Then
		\begin{equation}
		\big|\mathrm{ESTR}(S)-\hat{\mathrm{ESTR}}(S)\big| \leq 6\epsilon_0\;\;\;\;\;\text{and}\;\;\;\;\;
		\big|\mathrm{CI}(S)-\hat{\mathrm{CI}}(S)\big| \leq  \sqrt{24 \epsilon_0} (1 +\nu),
		\end{equation}
		%\label{lem:approx-1d-estr-ci}
	\end{lemma}
}

\begin{proof}
	We first prove the upper bound on $|\mathrm{ESTR}(S)-\hat{\mathrm{ESTR}(S)}|$, which is
	\begin{equation}
	\left|\frac{\sum_{i\in S}\hat u_{ti}r_{ti}}{1+\sum_{i\in S}\hat u_{ti}} - \frac{\sum_{i\in S}\gamma_i}{1+\sum_{i\in S}\mu_i}\right| \leq 6\epsilon_0,
	\label{eq:approx-estr}
	\end{equation}
	where $\mu_i = [\hat u_{ti}/\Delta]\cdot \Delta$, $\gamma_i = [\hat u_{ti}r_{ti}/\Delta]\cdot \Delta$.
	
	Denote $A := \sum_{i\in S}\hat u_{ti}r_{ti}$ and $B := 1 + \sum_{i\in S}\hat u_{ti}$.
	Because $r_{ti}\leq 1$, we have $A\leq B$.
	Let also $\tau_1 := \sum_{i\in S}\gamma_i - A$ and $\tau_2 := 1 + \sum_{i\in S}\mu_i - B$.
	Because $\max\{|\gamma_i - \hat u_{ti}r_{ti}|, |\mu_i-\hat u_{ti}|\} \leq \Delta$,
	we have $\max\{|\tau_1|,|\tau_2|\}\leq \Delta\cdot K$.
	Subsequently,
	\begin{align*}
	\left|\frac{\sum_{i\in S}\hat u_{ti}r_{ti}}{1+\sum_{i\in S}\hat u_{ti}} - \frac{\sum_{i\in S}\gamma_i}{1+\sum_{i\in S}\mu_i}\right|
	&= \left|\frac{A}{B} - \frac{A+\tau_1}{B+\tau_2}\right|
	= \left|\frac{A\tau_2-B\tau_1}{B(B+\tau_2)}\right|
	= \left|\frac{A\tau_2-B\tau_2+B\tau_2-B\tau_1}{B(B+\tau_2)}\right|\\
	&\leq \left|\frac{(A-B)\tau_2}{B(B+\tau_2)}\right| + \frac{|\tau_1|+|\tau_2|}{B-|\tau_2|}
	\leq \frac{|\tau_1|+2|\tau_2|}{B-|\tau_2|},
	\end{align*}
	where the last inequality holds because $A\leq B$.
	Using $B=1+\sum_{i\in S}\hat u_{ti} \geq 1 + \hat{u}_{tq} \geq U$ (since $q \in S$ and $U = \max\{1, u_{tq}\}$, and $\max\{|\tau_1|,|\tau_2|\}\leq \Delta\cdot K = \epsilon_0 U$, we have
	$$
	\frac{|\tau_1|+2|\tau_2|}{B-|\tau_2|} \leq \frac{3 \epsilon_0 U}{U - \epsilon_0 U} \leq 6\epsilon_0,
	$$
	provided that $\epsilon_0\in(0,1/2]$.
	Eq.~(\ref{eq:approx-estr}) is thus proved.
	
	We next prove the upper bound on $|\mathrm{CI}(S)-\hat{\mathrm{CI}}(S)|$, which is
	\begin{equation}
	\left|\sqrt{\frac{\sum_{i\in S}\hat u_{ti}x_{it}^2}{1+\sum_{i\in S}\hat u_{ti}}-\left(\frac{\sum_{i\in S}\hat u_{ti}x_{ti}}{1+\sum_{i\in S}\hat u_{ti}}\right)^2} - \sqrt[*]{\frac{\sum_{i\in S}\beta_i}{1+\sum_{i\in S}\mu_i} - \left(\frac{\sum_{i\in S}\alpha_i}{1+\sum_{i\in S}u_i}\right)^2}\right| \leq  \sqrt{24 \epsilon_0} (1 +\nu),
	\label{eq:approx-ci}
	\end{equation}
	where $\sqrt[*]{\cdot} = \sqrt{\max\{0,\cdot\}}$, $\mu_i = [\hat u_{ti}/\Delta]\cdot \Delta$, $\alpha_i = [\hat u_{ti}x_{ti}/\Delta]\cdot \Delta$,
	$\beta_i = [\hat u_{ti}x_{ti}^2/\Delta]\cdot \Delta$.
	
	Denote $C := \frac{\sum_{i\in S}\hat u_{ti}x_{ti}}{1+\sum_{i\in S}\hat u_{ti}}$ and $D := \frac{\sum_{i\in S}\hat u_{ti}x_{ti}^2}{1+\sum_{i\in S}\hat u_{ti}}$.
	Because $|x_{ti}|\leq\nu$ for all $t$ and $i$, we have $C\in[-\nu,\nu]$ and $D\in[0,\nu^2]$.
	Denote also $\tau_3 := \frac{\sum_{i\in S}\alpha_i}{1+\sum_{i\in S}\mu_i}-C$ and $\tau_4 := \frac{\sum_{i\in S}\beta_i}{1+\sum_{i\in S}\mu_i} - D$.
	Using the same analysis as in the proof of Eq.~(\ref{eq:approx-estr}), we have $|\tau_3| \leq 6 \epsilon_0 (1+\nu)$ and $|\tau_4| \leq 6 \epsilon_0 (1+\nu^2)$.
	
	With the definitions of $C$, $D$, $\tau_3$ and $\tau_4$, the left-hand side of Eq.~(\ref{eq:approx-ci}) can be re-written as
	\begin{equation}
	\big|\sqrt{D-C^2} - \sqrt[*]{(D+\tau_4) - (C+\tau_3)^2}\big|.
	\label{eq:approx-ci-ref}
	\end{equation}
	
	\paragraph{Case 1: $D-C^2>-(\tau_4-2\tau_3C-\tau_3^2)$.}
	In this case, we have
	\begin{align*}
	\text{Eq.~(\ref{eq:approx-ci-ref})}
	&= \frac{|\tau_4-2\tau_3C-\tau_3^2|}{\sqrt{D-C^2} + \sqrt{D-C^2+(\tau_4-2\tau_3 C-\tau_3^2)}}\leq \sqrt{|\tau_4-2\tau_3 C-\tau_3^2|}\\
	&\leq \sqrt{6\epsilon_0 (1 + \nu^2) + 2 \cdot 6\epsilon_0 (1 + \nu)^2 + 6\epsilon_0 (1+\nu)} \leq \sqrt{24 \epsilon_0} (1 +\nu).
	\end{align*}
	
	\paragraph{Case 2: $D-C^2\leq -(\tau_4-2\tau_3C-\tau_3^2)$.}
	In this case, we have $(D+\tau_4)-(C+\tau_3)^2\leq 0$ and subsequently
	\begin{equation*}
	\text{Eq.~(\ref{eq:approx-ci-ref})}
	= \sqrt{D-C^2} \leq  \sqrt{|\tau_4-2\tau_3 C-\tau_3^2|} \leq  \sqrt{24 \epsilon_0} (1 +\nu).
	\end{equation*}
	
	Combining both cases we prove Eq.~(\ref{eq:approx-ci}).
\end{proof}

\subsection{Proof of Lemma \ref{lem:approx-md}}

{
	\renewcommand{\thelemma}{\ref{lem:approx-md}}
	\begin{lemma}[restated]
		Suppose there exists $\ell\in[L]$ such that $\langle y^{(\ell)},y^*\rangle \geq 1/\alpha$ for some $\alpha\geq 1$ in Algorithm \ref{alg:approx-md}, then
		$\mathrm{ESTR}(\hat S^{(\ell)}) + \min\{1,\alpha\omega\cdot\mathrm{CI}(\hat S^{(\ell)})\}+\varepsilon \geq \mathrm{ESTR}(S^*)+\min\{1,\omega\cdot\mathrm{CI}(S^*)\}$,
		where $\varepsilon>0$ is the approximation parameter of the univariate problem instances.
	\end{lemma}
}

\begin{proof}	
	For each assortment $S$, define $\mathrm{CI}^{(\ell)} (S)$ by
	\[
	\mathrm{CI}^{(\ell)} (S) :=
	{y^{(\ell)}}^\top \left(\frac{\sum_{j\in S}\hat u_{tj}x_{tj}x_{tj}^\top}{1+\sum_{j\in S}\hat u_{tj}} - \left(\frac{\sum_{j\in S}\hat u_{tj}x_{tj}}{1+\sum_{j\in S}\hat u_{tj}}\right)\left(\frac{\sum_{j\in S}\hat u_{tj}x_{tj}}{1+\sum_{j\in S}\hat u_{tj}}\right)^\top\right) y^{(\ell)},
	\]
	Since $\mathrm{CI}^{(\ell)} (S) \leq \mathrm{CI} (S)$, we have
	\begin{equation}\label{eq:approx-md-eq1}
	\mathrm{ESTR}(\hat S^{(\ell)}) + \min\{1,\alpha\omega\cdot\mathrm{CI}(\hat S^{(\ell)})\}+\varepsilon  \geq \mathrm{ESTR}(\hat S^{(\ell)}) + \min\{1,\alpha \omega\cdot\mathrm{CI}^{(\ell)}(\hat S^{(\ell)})\}+\varepsilon .
	\end{equation}
	By the approximation guarantee of Algorithm \ref{alg:approx-1d}, we have
	\begin{equation}\label{eq:approx-md-eq2}
	\mathrm{ESTR}(\hat S^{(\ell)}) + \min\{1,\alpha\omega\cdot\mathrm{CI}^{(\ell)}(\hat S^{(\ell)})\}+\varepsilon \geq \mathrm{ESTR}( S^*) + \min\{1,\alpha\omega\cdot\mathrm{CI}^{(\ell)}( S^*)\} .
	\end{equation}
	Since $\langle y^{(\ell)},y^*\rangle \geq 1/\alpha$, we have $\mathrm{CI}^{(\ell)} (S^*) \geq (1/\alpha) \cdot\mathrm{CI}(S^*)$. Therefore,
	\begin{equation}\label{eq:approx-md-eq3}
	\mathrm{ESTR}( S^*) + \min\{1,\alpha\omega\cdot\mathrm{CI}^{(\ell)}( S^*)\} \geq \mathrm{ESTR}( S^*) + \min\{1,\omega\cdot\mathrm{CI}( S^*)\} .
	\end{equation}
	
	The lemma is proved by combining Eq.~\eqref{eq:approx-md-eq1}, Eq.~\eqref{eq:approx-md-eq2}, and Eq.~\eqref{eq:approx-md-eq3}.
\end{proof}

\subsection{Proof of Proposition \ref{prop:greedy-omega-zero}}
{
	\renewcommand{\theproposition}{\ref{prop:greedy-omega-zero}}
	\begin{proposition}[restated]
		If $\omega=0$, then Algorithm \ref{alg:greedy} terminates in $O(N^4)$ iterations and produces an output $S$
		that maximizes $\mathrm{ESTR}(S)$.
	\end{proposition}
}

\begin{proof}
	We first show that when the algorithm terminates with $\mathrm{ESTR}(S) = r$, $S$ is one of the optimal assortments. Suppose $S$ is not an optimal assortment, i.e.\ there exists $S^\#$ such that $\mathrm{ESTR}(S^\#) > r$, we show that the algorithm will not terminate. By the definition of $\mathrm{ESTR}(\cdot)$ we have $\sum_{i \in S^\#} \hat{u}_{ti} (r_{ti} - r) > r$ and  $\sum_{i \in S} \hat{u}_{ti} (r_{ti} - r) = r$. By comparing $S^\#$ and $S$, one can find a new candidate assortment $S'$ via swapping, adding, or deleting an item from/to $S$ such that $\sum_{i \in S'} \hat{u}_{ti} (r_{ti} - r) > r$. Therefore, $\mathrm{ESTR}(S') > r$ and the algorithm will not terminate.
	
	It remains to show that the algorithm terminates in $O(N^4)$ iterations.
	
	For each $r \in [0, 1]$, we define a total order $\geq_r$ on $[N] \cup \{\bot\}$, where $[N]$ corresponds to the $N$ items and $\bot$ is a special element with the definition $\hat{u}_{t\bot} = r_{t\bot} = 0$ for convenience, as follows: $i \geq_r j$ if and only if $\hat{u}_{ti} (r_{ti} - r) \geq \hat{u}_{tj} (r_{tj} - r)$ (and consequently  $i >_r j$ if and only if $\hat{u}_{ti} (r_{ti} - r) > \hat{u}_{tj} (r_{tj} - r)$). It is straightforward to verify that there exists $O(N^2)$ section points $\theta_0 < \theta_1 < \theta_2 < \dots < \theta_{L - 1} < \theta_L = 1$ so that for any two $r_1, r_2$ that sandwiched by the same pair of neighboring section points (i.e.\ $\exists \ell \in [L]: r_1, r_2 \in (\theta_{\ell - 1}, \theta_\ell)$), we have $\geq_{r_1} \equiv \geq_{r_2}$. Indeed, one can set the section points to be the solutions to the equalities $\hat{u}_{ti} (r_{ti} - r) = \hat{u}_{tj} (r_{tj} - r)$ for every pair of $i, j \in [N] \cup \{\bot\}$.
	
	We will show that if $\mathrm{ESTR}(S) \in (\theta_{\ell-1}, \theta_\ell)$ for some $\ell \in [L]$, after at most $O(N^2)$ iterations, either the algorithm terminates or $\mathrm{ESTR}(S) \geq \theta_{\ell}$. This directly leads to an $O(N^4)$ upper bound on the total number of iterations that the algorithm performs. We pick an arbitrary $r \in (\theta_{\ell-1}, \theta_\ell)$ and define the following two potential functions: $I(S) = |\{(i, j): i \in S, j \in [N] \backslash S, i <_r j \}|$, and $J(S) = |\{i \in S: i <_r \bot\}|$. We have the following observations:
	\begin{itemize}
		\item When a swapping operation is performed on $S$, $I(S)$ strictly decreases and $J(S)$ does not increase.
		\item When a deletion operation is performed on $S$, $I(S)$ increases by at most $N$ and $J(S)$ strictly decreases.
		\item When an addition operation is performed on $S$, $I(S)$ increases by at most $N$ and $J(S)$ does not increase.
	\end{itemize}
	We let $F(S) = I(S) + (2N+1) \cdot J(S) \in [0, O(N^2)]$. Suppose there are $a$ swapping operations, $b$ deletion operations, and $c$ addition operations done in total, $S$ is the assortment that the algorithm begins with and $T$ is the last assortment satisfying $\mathrm{ESTR}(T) < \theta_\ell$. Observe that there are at most $c \leq b + N$ addition operations. Together with the three observations above, we have
	\begin{multline*}
	0 \leq F(T) \leq F(S) - a + N b - (2N + 1) b + N c \leq F(S) - a + N b - (2N + 1) b + N (b + N) \\
	\leq F(S) + N^2 - a - b \leq O(N^2) - a - b .
	\end{multline*}
	In total, we have $a + b \leq O(N^2)$. Therefore, the total number of iterations where $\mathrm{ESTR}(S) \in (\theta_{\ell -1}, \theta_\ell)$ is $a + b + c \leq a + 2b + N \leq O(N^2)$.
\end{proof}

\subsection{Proof of Proposition \ref{prop:approx-init}}
{
	\renewcommand{\theproposition}{\ref{prop:approx-init}}
	\begin{proposition}[restated]
		Assume that $d \geq 2$. Let $y^*\in\mathbb R^d$, $\|y^*\|_2=1$ be fixed and $y$ be sampled uniformly at random from the unit $d$-dimensional sphere.
		Then
		\begin{equation*}
		\Pr[\langle y,y^*\rangle \geq 1/\sqrt{d}] = \Omega(1) \;\;\;\;\;\text{and}\;\;\;\;\; \Pr[\langle y,y^*\rangle \geq 1/2] = \exp\{-O(d)\}.
		\end{equation*}
	\end{proposition}
}
\begin{proof}
	Assume without loss of generality that $y^* = (1, 0, 0, \dots, 0)$, and let $y$ be sampled as follows. Sample $z_i \sim N(0, 1)$ independently for each $i \in [d]$, and let $y = z / \|z\|_2$. Now, $\langle y, y^* \rangle = z_1 /  \|z\|_2$.
	
	We first prove $\Pr[\langle y,y^*\rangle \geq 1/\sqrt{d}] = \Pr[z_1 / \|z\|_2 \geq 1 / \sqrt{d}] = \Omega(1)$. Note that when $z_1 \geq 10$ and $ \sqrt{z_2^2 + \dots + z_d^2} \leq 5\sqrt{d}$, we have $z_1 / \|z\|_2 = 1 / \sqrt{1 + (z_2^2 + \dots + z_d^2)/z_1^2} \geq 1 / \sqrt{1 + (5 \sqrt{d})^2 / 10^2} \geq 1 / \sqrt{d}$, where the last inequality holds for $d \geq 2$. Therefore,
	\begin{align*}
	\Pr[z_1 / \|z\|_2  \geq 1/\sqrt{d}] &\geq \Pr\left[z_1 \geq 10 \wedge \sqrt{z_2^2 + \dots + z_d^2} \leq 5\sqrt{d}\right] \\
	&= \Pr[z_1 \geq 10] \cdot \Pr\left[ \sqrt{z_2^2 + \dots + z_d^2} \leq 5\sqrt{d}\right] = \Omega(1) .
	\end{align*}
	
	Now we prove $\Pr[\langle y,y^*\rangle \geq 1/2] = \Pr[z_1 / \|z\|_2 \geq 1 / 2] = \exp\{-O(d)\}$. Similarly, when $z_1 \geq 5\sqrt{d}$ and $ \sqrt{z_2^2 + \dots + z_d^2} \leq 5\sqrt{d}$, we have $z_1 / \|z\|_2 = 1 / \sqrt{1 + (z_2^2 + \dots + z_d^2)/z_1^2} \geq 1 / \sqrt{1 + 1} > 1 / 2$. Therefore,
	\begin{align*}
	\Pr[z_1 / \|z\|_2  \geq 1/2] &\geq \Pr\left[z_1 \geq 5\sqrt{d} \wedge \sqrt{z_2^2 + \dots + z_d^2} \leq 5\sqrt{d}\right] \\
	&= \Pr[z_1 \geq 5\sqrt{d}] \cdot \Pr\left[ \sqrt{z_2^2 + \dots + z_d^2} \leq 5\sqrt{d}\right] \\
	&= \exp\{-O(d)\} \cdot \Omega(1) = \exp\{-O(d)\}.
	\end{align*}
\end{proof}

% References here (outcomment the appropriate case)

% CASE 1: BiBTeX used to constantly update the references
%   (while the paper is being written).
%\bibliographystyle{informs2014} % outcomment this and next line in Case 1
%\bibliography{<your bib file(s)>} % if more than one, comma separated

% CASE 2: BiBTeX used to generate mypaper.bbl (to be further fine tuned)
%\input{mypaper.bbl} % outcomment this line in Case 2

%If you don't use BiBTex, you can manually itemize references as shown below.

\end{document}